\documentclass[ejsv2]{imsart}

\RequirePackage[numbers]{natbib}
\RequirePackage[colorlinks,citecolor=blue,urlcolor=blue]{hyperref}
\RequirePackage{graphicx}
\RequirePackage{graphicx}
\usepackage{xcolor}
\usepackage{lscape}
\usepackage{enumerate}
\arxiv{2010.00000}
\startlocaldefs
\theoremstyle{plain}

\newtheorem{theorem}{Theorem}[section]
\newtheorem{lemma}[theorem]{Lemma}
\theoremstyle{definition}

\theoremstyle{remark}

\def\bc{\boldsymbol{c}}
\def\be{\boldsymbol{e}}
\def\bell{\boldsymbol{\ell}}
\def\bt{\boldsymbol{t}}

\def\bG{\boldsymbol{G}}
\def\bH{\boldsymbol{H}}

\def\bS{\boldsymbol{S}}
\def\bU{\boldsymbol{U}}
\def\bV{\boldsymbol{V}}
\def\bW{\boldsymbol{W}}
\def\bX{\boldsymbol{X}}
\def\bY{\boldsymbol{Y}}
\def\bZ{\boldsymbol{Z}}
\def\bOmega{\boldsymbol{\Omega}}
\def\bTheta{\boldsymbol{\Theta}}
\def\bSigma{\boldsymbol{\Sigma}}
\def\bnu{\boldsymbol{\nu}}
\def\IE{{\mathbb E}}
\def\bzero{\mathbf{0}}
\DeclareMathOperator*{\argmin}{arg\,min}

\newcommand{\revise}{\color{black}}
\newcommand{\reviseagain}{\color{black}}
\endlocaldefs

\begin{document}
\begin{frontmatter}
\title{A sample article title}
\runtitle{A sample running head title}

\begin{aug}
\author[A]{\fnms{Nicholas}~\snm{Woolsey}\ead[label=e1]{nwoolsey@email.sc.edu}},
\author[B]{\fnms{Xianzheng}~\snm{Huang}\ead[label=e2]{huang@stat.sc.edu}\orcid{0000-0001-7077-0869
}}
\address[A]{Statistics,
University of South Carolina\printead[presep={,\ }]{e1}}

\address[B]{Statistics,
University of South Carolina\printead[presep={,\ }]{e2}}
\runauthor{N. Woolsey et al.}
\end{aug}

\begin{abstract}
This study considers regression analysis of a circular response with an error-prone linear covariate. Starting with an existing estimator of the circular regression function that assumes error-free covariate, three approaches are proposed to revise this estimator, leading to three nonparametric estimators for the circular regression function accounting for measurement error. The proposed estimators are intrinsically connected through some deconvoluting operator that is exploited differently in different estimators. Moreover, a new bandwidth selection method is developed that is more computationally efficient than an existing method well-received in the context of tuning parameter selection in the presence of measurement error. The efficacy of these new estimators and their relative strengths are demonstrated through a thorough investigation of their asymptotic properties and extensive empirical study of their finite-sample performance. 
\end{abstract}

\begin{keyword}[class=MSC]
\kwd[Primary ]{62G08}
\kwd[; secondary ]{62H11}
\end{keyword}

\begin{keyword}
\kwd{Bandwidth selection}
\kwd{complex error}
\kwd{deconvolution} 
\kwd{measurement error}
\end{keyword}

\end{frontmatter}
\tableofcontents

\section{Introduction}

\label{sec:intro}

Circular data are data of a periodic nature that arise in a host of applications. Examples of circular data include data on the direction of animal movements that are of interest in zoology and ecology \citep{batschelet1981circular, rivest2016general}, and wind direction data as a subject of investigation in meteorology \citep{kamisan2011distribution}. Times of certain events, such as hospital visits and crimes, as observed on the 24-hour clock can be transformed to circular data, so are dates of certain events on a calendar, such as births and migration. In these applications, researchers are often interested in studying the association between a circular response and linear covariates. For instance,  \citet{birkett2012animal} considered regressing elephant's movements on rainfall;  \citet{hodel2022circular} analyzed the association between a moving animal's turn‐angle and step‐length; \citet{garcia2013exploring} formulated a regression model for wind direction and $\text{SO}_2$ concentration; \citet{gill2010circular} developed a circular regression model for the timing of domestic terrorism events, with linear covariates such as the number of attackers and death count. 

Early developments of circular-linear regression analysis led to parametric models that usually assume a von Mises distribution \citep{mardia2000directional} for the circular response given covariates. Frequently referenced regression models of this type include those proposed by \citet{gould1969regression}, \citet{Johnson}, and \citet{fisherlee}. More recent parametric models that adopt different distributional assumptions on the circular response include those considered in \cite{wang2013directional}, \cite{scealy2019scaled}, and  \cite{paine2020spherical},  among many others. A complication in formulating such regression models is the choice of a link function that relates predictors in the Euclidean space to the mean direction in a circular space. There is also the concern that the actual circular response may violate parametric assumptions implied by the chosen circular distribution. Nonparametric circular-linear regression models have been proposed to avoid choosing a link function or imposing stringent assumptions on the circular response. One of the first notable developments along this line is the local polynomial regression of a circular response introduced by \citet{di2013non}. \citet{meilan2020nonparametric} generalized this work by allowing a multivariate linear covariate. 

Oftentimes covariates of scientific interest cannot be measured precisely, either due to human error or imprecise measuring instruments. For example,  rainfall measurements from a rain gauge are affected by the exposure conditions of the rain gauge; the measurement accuracy of $\text{SO}_2$ concentration is also highly device-dependent. This motivates our study of circular regression models with a linear covariate prone to measurement error. The topic of regression analysis with error-prone covariates has been studied extensively. Comprehensive reviews of existing methodolgies for dealing with error-in-covariate can be found in \cite{carroll2006measurement}, \cite{fuller2009measurement}, \cite{buonaccorsi2010measurement}, \cite{yi2017statistical}, and
\cite{Yi}. Despite the existing extensive research on regression models with error-in-covariate, parallel problems in the context of regression models for circular responses have received little attention. Our study contributes in this underinvestigated area.  Concurrently, \cite{di2023kernel} considered regression models for a response that can be circular or linear, with error-prone covariates that also can be circular or linear. 

Focusing on circular-linear regression models, we propose in this study three methods based on local linear estimation adapted for a circular response while accounting for covariate measurement error. To prepare for the methodology development, Section~\ref{sec:regnome} briefly reviews local polynomial regression of a circular response with error-free covariates proposed by \cite{di2013non}. In Section~\ref{sec:regwithme} we develop three strategies to revise their method to account for measurement error in the linear covariate. Section~\ref{sec:asymp} summarizes asymptotic properties of our proposed estimators. Bandwidth selection required for local polynomial regression is especially challenging when covariates are prone to measurement error, which is an important issue not addressed in \cite{di2023kernel}. We present in Section~\ref{sec:band} two bandwidth selection methods utilized in our proposed estimation methods. Section~\ref{sec:simu} reports simulation studies where the proposed estimators' finite-sample performance is inspected and compared with a naive estimator and a benchmark estimator. We implement these methods in a real-life application in Section~\ref{sec:real}. Lastly, Section~\ref{sec:disc1} summarizes contributions of our study and discuss extensions of the study to different relevant settings. 

\section{Local polynomial circular-linear regression}
\label{sec:regnome}
\subsection{The circular mean function}
Consider $n$ independent bivariate observations, $\{(\Theta_j, X_j)\}_{j=1}^n$, identically distributed according to the joint distribution of $(\Theta, X)$ supported on $[-\pi, \pi) \times \mathbb{R}$. Given the linear covariate $X=x$, the distribution of the circular response $\Theta$ is specified by 
\begin{align}
\Theta & =\{m(x)+\epsilon\}(\text{mod }2\pi), 
\label{eq:thetagivenX}
\end{align}
where $\epsilon$ is an angular error following an unspecified distribution with a circular mean of zero, and $m(x)$ is the circular mean of $\Theta$, i.e., the circular regression function that maps from a linear space to a circular space. The circular mean of a circular random variable $\Theta$ is defined as the angle $\alpha\in [-\pi, \pi)$ that minimizes the mean cosine dissimilarity between $\Theta$ and $\alpha$, $D(\alpha)=\IE\{1-\cos(\Theta-\alpha)\}$, which can be shown to be $\argmin D(\alpha)=\text{atan2}[\IE(\sin \Theta), \IE(\cos \Theta)]$ \citep[][Chapter 2]{mardia2000directional}. Here,   $\text{atan2}[y, x]$ is the angle between the $x$-axis and the vector from the origin to the point $(x,y)$ in the Cartesian coordinate system. 

Because $\text{atan2}[y, x]$ depends on $(x, y)$ only via the ratio $y/x$, one may infer the circular mean function $m(x)$ via inferring $m_1(x)/m_2(x)$, where $m_1(x)=\IE(\sin \Theta|X=x)$ and $m_2(x)=\IE(\cos \Theta|X=x)$, both functions mapping to a linear space as in the traditional regression setting.

\subsection{Local polynomial estimators}
Given the response data $\bTheta=(\Theta_1, \ldots,  \Theta_n)^\top$ and covariate data $\bX=(X_1, \ldots,$ $ X_n)^\top$, \cite{di2013non} proposed kernel-based estimators for $m(x)$  in the form of 
\begin{equation}\label{eq:mtilde}
    \tilde{m}(x)=\text{atan}2[\tilde{g}_{1}(x),\tilde{g}_{2}(x)], 
\end{equation}
where
\begin{equation}
    \tilde{g}_{1}(x)=\frac{1}{n}\sum_{j=1}^{n}\sin (\Theta_{j})\mathcal{W}(X_{j}-x), \,
    \, \tilde{g}_{2}(x)=\frac{1}{n}\sum_{j=1}^{n}\cos (\Theta_{j})\mathcal{W}(X_{j}-x), 
    \label{eq:g1g2tilde}
\end{equation}
in which $\mathcal{W}(t)$ is a weight function. For example, if $\mathcal{W}(t)=K_h(t)$, where $K_h(t)=K(t/h)/h$, $K(t)$ is a kernel function, and $h$ is the bandwidth, then  
\begin{align*}
\tilde g_1(x)  & = \frac{1}{n}\sum_{j=1}^n \sin (\Theta_j) K_h(X_j-x)\\
 & = \frac{\sum_{j=1}^n \sin (\Theta_j) K_h(X_j-x)}{\sum_{j=1}^n K_h(X_j-x)} \times \frac{1}{n}\sum_{j=1}^n K_h(X_j-x).
\end{align*}
Now one can see that the first factor in the preceding expression is a local constant estimator \citep{Nadaraya} for $m_1(x)$, and the second factor is a traditional kernel density estimator for the probability density function of $X$ evaluated at $x$, denoted by $f_X(x)$. Hence, with $\mathcal{W}(t)=K_h(t)$, $\tilde g_1(x)$ is a sensible estimator for $g_1(x)=m_1(x)f_X(x)$; similarly, $\tilde g_2(x)$ is an estimator for $g_2(x)=m_2(x)f_X(x)$.

The estimator $\tilde m(x)$ in \eqref{eq:mtilde} acknowledges that the traditional local polynomial estimation developed for a linear response is inadequate for a circular response, but remains applicable for estimating $g_1(x)$
 and $g_2(x)$, both functions mapping  from $\mathbb{R}$ onto another linear space. This allows estimation of the ratio $m_1(x)/m_2(x)=g_1(x)/g_2(x)$ based on two rounds of local polynomial estimation. Besides the local constant weight, \cite{di2013non} also considered using the following weight function that depends on all covariate data $\bX$, 
 \begin{equation}
 \label{eq:idealweight}
 \begin{aligned}
  \mathcal{W}(t)= &\ K_h(t)\times \frac{1}{n}  \sum_{k=1}^n \left(\frac{X_k-x}{h}\right)^2 K_h(X_k-x)-\\
  &\ \frac{t}{h} K_h(t) \times  \frac{1}{n}\sum_{k=1}^n \left(\frac{X_k-x}{h}\right) K_h(X_k-x).
  \end{aligned}
 \end{equation}
 This more complicated weight function relates to local linear estimators for the regression function of a linear response $Y$, $n^{-1}\sum_{j=1}^n Y_j \mathcal{W}^*(X_j-x)$, with the weight given by \citep{Fan}
 \begin{equation}
 \label{eq:completeweight}
 \begin{aligned}
    &\ \mathcal{W}^*(X_j-x) \\
    = &\ \mathcal{W}(X_j-x)\left[\left\{\frac{1}{n}\sum_{k=1}^n K_h(X_k-x)\right\}\left\{\frac{1}{n}\sum_{k=1}^n \left(\frac{X_k-x}{h}\right)^2K_h(X_k-x)\right\}-\right.\\
    &\ \left. \left\{\frac{1}{n}\sum_{k=1}^n \left(\frac{X_k-x}{h}\right)K_h(X_k-x)\right\}^2\right]^{-1},
\end{aligned}
\end{equation}
where $\mathcal{W}(X_j-x)$ is equal to \eqref{eq:idealweight} evaluated at $t=X_j-x$. We thus call $\mathcal{W}(X_j-x)$ the local linear weight, and refer to $\mathcal{W}^*(X_j-x)$ as the normalized local linear weight. Plugging in \eqref{eq:g1g2tilde} the local linear weight, one has $\tilde g_\ell(x)=C(\bX)\tilde m_\ell(x)$, where $C(\bX)$ is the term inside the square brackets in  \eqref{eq:completeweight}, and $\tilde m_\ell(x)$ is a local linear estimator of $m_\ell(x)$, for $\ell=1, 2$. Consequently, $\tilde m(x)$ in \eqref{eq:mtilde} is equal to $\text{atan2}[\tilde m_1(x), \tilde m_2(x)]$ since $\tilde g_1(x)/\tilde g_2(x)=\tilde m_1(x)/\tilde m_2(x)$. Under mild conditions {\revise that guarantee the consistency of a local linear estimator for the regression function of a linear response \citep[][Theorem 3.1]{Fan} and the consistency of a kernel density estimator \citep[][Section 6.2]{scott2015multivariate}}, $\tilde g_\ell(x)$ consistently estimates $g_\ell(x)=m_\ell(x)f^2_X(x)\mu_2 $, for $\ell=1, 2$, where $\mu_2=\int t^2 K(t) \, dt$. Unless otherwise stated, all integrals in this article are over the entire real line $\mathbb R$.

Following the above two examples of $\mathcal{W}(t)$, one can easily generalize the construction of $\tilde m(x)$ by adopting weight functions that relate to traditional local polynomial estimators of higher orders. For a more concrete exposition of methodology development and theoretical derivations, we focus on $\tilde m(x)$ with the local linear weight as a benchmark estimator that other estimators brought forth in the next section relate to and compare with. 

\section{Circular-linear regression with error-in-covariate}
\label{sec:regwithme}
Suppose that, instead of $\bX$, error-contaminated covariate data $\bW=(W_{1},\ldots,$ $W_{n})^\top$ are observed, where
\begin{align}
W_{j} & =X_{j}+U_{j}, 
\label{eq:WgivenX}
\end{align}
in which $U_j$ is a mean-zero measurement error following a known distribution with variance $\sigma_u^{2}$, and $U_j$ is independent of $(\Theta_j, X_j, \epsilon_j)$, for $j=1, \ldots, n$. Had one ignored measurement error, one would substitute $X_j$ with $W_j$ in \eqref{eq:g1g2tilde} for $j=1, \ldots, n$, leading to a naive estimator of $m(x)$, denoted by $ \hat m^*(x)$. Because $\hat m^{*}(x)$ is merely a sensible estimator for the circular mean of $\Theta$ given $W=x$, denoted by $m^*(x)$, a function  typically different from $m(x)$, so naive estimation of $m(x)$ using $\hat m^*(x)$ is usually misleading. 

In what follows, we develop three strategies to modify $\hat m^*(x)$ in order to account for covariate measurement error when inferring $m(x)$. To contrast with the naive estimator $\hat m^*(x)$, we call $\tilde m(x)$ in \eqref{eq:mtilde} an ideal estimator of $m(x)$ to stress that $\tilde m(x)$ is only available in the ``ideal'' situation when $\sigma_u^2=0$ and thus $\bW=\bX$. 

\subsection{The deconvoluting kernel estimator}
\label{sec:DKest}
Since the naive substitution of $X_j$'s with $W_j$'s only occurs in the weight, we first propose to correct $\hat m^*(x)$ for measurement error by correcting the naive local linear weight, 
\begin{equation*}
 \begin{aligned}
 \mathcal{W}(W_j-x)= &\ \frac{1}{n} K_h(W_j-x) \sum_{k=1}^n \left(\frac{W_k-x}{h}\right)^2 K_h(W_k-x)-  \\
 &\ \frac{1}{n}\left(\frac{W_j-x}{h}\right) K_h(W_j-x)\sum_{k=1}^n \left(\frac{W_k-x}{h}\right) K_h(W_k-x).
 \end{aligned}
\end{equation*}
More specifically, we replace the above naive weight by 
\begin{equation}
    \label{eq:DKweight}
\begin{aligned}
\mathcal{L}(W_j-x)=&\ \frac{1}{n}L_{0,h}(W_j-x)\sum_{k=1}^n \left(\frac{W_k-x}{h}\right)^2 L_{2,h}(W_k-x)-\\
&\ \frac{1}{n}\left(\frac{W_j-x}{h}\right)L_{1,h}(W_j-x)\sum_{k=1}^{n}\left(\frac{W_k-x}{h}\right)L_{1,h}(W_k-x),
\end{aligned}
\end{equation}
where, for $\ell=0, 1, 2$, $L_{\ell, h}(x)=L_\ell(x/h)/h$, and 
\begin{align}
L_\ell(x)=i^{-\ell}x^{-\ell}\frac{1}{2\pi}\int e^{-itx}\frac{\phi_{K}^{(\ell)}(t)}{\phi_{U}(-t/h)}dt, 
\label{eq:LDK}
\end{align} 
in which $i=\sqrt{-1}$, $\phi_U(t)$ is the characteristic function of the measurement error $U$, $\phi_K(t)$ is the Fourier transform of the kernel, and $\phi_{K}^{(\ell)}(t) = (\partial^{\ell}/ \partial t^\ell)\phi_{K}(t)$. In this paper, we use $\phi_A(t)$ to denote the characteristic function of $A$ if $A$ is a random variable, and to denote the Fourier transform of $A$ if $A(\cdot)$ is a function. Setting $\ell=0$ in \eqref{eq:LDK} gives the deconvoluting kernel used in the deconvoluting kernel density estimator estimator for $f_X(x)$ based on $\bW$ \citep{stefanski1990deconvolving}. \cite{Delaigle} generalized the deconvoluting kernel by introducing \eqref{eq:LDK} for $\ell=0, 1, 2, \ldots$, and used it to construct local polynomial estimators for the regression function of a linear response with error-in-covariate. 

Under suitable conditions (to be stated in Section~\ref{sec:asymp}), \cite{Delaigle} showed that 
\begin{equation}
\IE\{(W-x)^\ell L_{\ell, h}(W-x)|X\}=(X-x)^\ell K_h(X-x), \text{ for } \ell=0,1,2. \label{eq:DKmean}
\end{equation}
Because \{$(X-x)^\ell K_h(X-x),\, \ell=0, 1, 2\}$ are building blocks of the local linear weight $\mathcal{W}(X_j-x)$, \eqref{eq:DKmean} sheds light on the rationale behind the new weight $\mathcal{L}(W_j-x)$ in \eqref{eq:DKweight}: instead of naively using $\{(W-x)^\ell K_h(W-x), \, \ell=0, 1, 2\}$ to substitute these building blocks, one uses $\{(W-x)^\ell L_{\ell,h}(W-x), \, \ell=0, 1, 2\}$ as unbiased estimators of them. In this line of arguments, as in \eqref{eq:DKmean}, we treat $X$, or $X_j$'s, as if they were unknown parameters, and we estimate functions of them using $\bW$. 

Using the new weight $\mathcal{L}(W_j-x)$ in \eqref{eq:DKweight}, we have our first proposed estimator of $m(x)$ referred to as the deconvoluting kernel estimator, $\hat{m}_{\text{DK}}(x)=\text{atan2}[\hat{g}_{1,\text{DK}}(x),\hat{g}_{2,\text{DK}}(x)],$
where $\hat g_{1,\text{DK}}(x) = n^{-1}\sum_{j=1}^n\sin (\Theta_j) \mathcal{L}(W_j-x)$ and $\hat g_{2,\text{DK}}(w)=n^{-1}\sum_{j=1}^n \cos (\Theta_j)\mathcal{L}(W_j-x)$ are estimators for $g_1(x)= \scalebox{.9}{$m_1(x) f^2_X(x)\mu_2$}$ and $g_2(x)= m_2(x)f^2_X(x)\mu_2$, respectively. Instead of correcting the naive local linear weight for measurement error as we do here, the estimator proposed by \citet[][see Section 4.1.2]{di2023kernel} in the same context results from using \eqref{eq:LDK} with $\ell=0$ to correct the naive local constant weight. 

\subsection{The complex error estimator}
\label{sec:CEest}
Following a similar direction leading to $\hat{m}_{\text{DK}}(x)$, we propose a second strategy where one corrects the naive normalized local linear weight, as the naive counterpart of $\mathcal{W}^*(X_j-x)$ in \eqref{eq:completeweight}, for measurement error. This new strategy is motivated by the result in Lemma~\ref{lem:addce2} given next, which we prove in Appendix A. Denote by $\mathbb{C}$ the set of complex numbers, by $\mathbb{N}$ the set of natural numbers, and define $\mathbb{N}_0=\mathbb{N}\cup \{0\}$. 
\begin{lemma}
\label{lem:addce2}
     For an entire function $g: \mathbb{C}^n\to \mathbb{C}$, if $\IE(\bV^{\bell})=0$,  $\forall \bell \in \mathbb{N}_0^n$ such that $|\bell|>0$, then $\IE\{g(\bt+\bV)\}=g(\bt)$, where $\bt=(t_1, \ldots, t_n)^\top\in \mathbb{R}^n$, and $\bV=(V_1, \ldots, V_n)^\top$ is an $n \times 1$ random vector. Here, for $\bell =(\ell_1, \ldots, \ell_n)\in \mathbb{N}_0^n$, $\bV^{\bell}=V_1^{\ell_1}\ldots V_n^{\ell_n}$, and $|\bell|=\ell_1+\ldots+\ell_n$. 
\end{lemma}
This result generalizes a similar result for a univariate entire function $g(\cdot)$ that was exploited in \cite{stefanski1989unbiased}, \cite{novick2002corrected}, and \cite{stefnovick} in measurement error problems. {\revise An entire function is a complex-valued function that is complex differentiable on the whole complex plane \citep{levin1996lectures}.} By Lemma~\ref{lem:addce2}, if measurement errors $\{U_j\}_{j=1}^n$ in \eqref{eq:WgivenX} are independent normal errors, then $\IE\{g(\bX+\bU+i\sigma_u \bZ)|\bX)=g(\bX)$, where $\bU=(U_1, \ldots, U_n)^\top$ and $\bZ=(Z_1, \ldots, Z_n)^\top$, in which $\{Z_j\}_{j=1}^n$ are independent standard normal errors that are also  independent of $\bX$ and $\bU$. This is because, now with $\bV=\bU+i \sigma_u \bZ$, $\IE(\bV^{\bell})=\prod_{j=1}^n\IE\{(U_j+i\sigma_u Z_j)^{\ell_j}\}=0$ for all $\bell\in \mathbb{N}_0^n$ as long as $|\bell|\ne 0$ in this case \citep{stefanski1989unbiased}. In conclusion, with normal measurement errors in \eqref{eq:WgivenX}, $g(\bW+i\sigma_u\bZ)$ is an unbiased estimator of $g(\bX)$. 

Now return to our goal of formulating a new weight using $\bW$ that estimates $\mathcal{W}^*(X_j-x)$ unbiasedly. Since \eqref{eq:completeweight} depends on all true covariate data, we re-write $\mathcal{W}^*(X_j-x)$ as $\mathcal{W}^*(X_j-x; \bX)$, viewed as a function of $\bX$, and we estimate this function in its entirety. {\revise In other words, we have 
\begin{equation*}
 \begin{aligned}
    \mathcal{W}^*(t; \bX) 
    = &\ \mathcal{W}(t)\left[\left\{\frac{1}{n}\sum_{k=1}^n K_h(X_k-x)\right\}\left\{\frac{1}{n}\sum_{k=1}^n \left(\frac{X_k-x}{h}\right)^2K_h(X_k-x)\right\}-\right.\\
    &\ \left. \left\{\frac{1}{n}\sum_{k=1}^n \left(\frac{X_k-x}{h}\right)K_h(X_k-x)\right\}^2\right]^{-1},
\end{aligned}
\end{equation*}
where $\mathcal{W}(t)$ is given in (4), and $t$ depends on $\bX$ via one of its entry.} Assuming normal measurement errors and that $\mathcal{W}^*(X_j-x; \bX)$ is an entire function, by Lemma~\ref{lem:addce2}, 
 $\mathcal{W}^*(W_j^*-x; \bW^*)$ is an unbiased estimator of $\mathcal{W}^*(X_j-x; \bX)$, where $W_j^*=W_j+i\sigma_u Z_j$, for $j=1, \ldots, n$, and $\bW^*=(W_1^*, \ldots, W_n^*)^\top$. Recognizing that both $\bW$ and $\bZ$ contribute to the variability of this unbiased estimator, we employ an average of it across $B$ realizations of $\bZ$ to obtain a less variable unbiased estimator for $\mathcal{W}^*(X_j-x)$,  leading to the new weight,
\begin{align}
\mathcal{L}^*(W_j-x) = &\ \frac{1}{B^*} \sum_{b=1}^{B^*} \mathcal{W}^*(W_{j,b}^*-x; \bW^*_b),
\label{eq:CEweight}
\end{align}
where $W^*_{j,b}=W_j+i \sigma_u Z_{j,b}$, for $j=1, \ldots, n$, $\bW^*_b=(W^*_{1,b}, \ldots, W^*_{n,b})^\top$, for $b=1, \ldots, B$, and $\{Z_{j,b}, \, j=1, \ldots, n\}_{b=1}^B$ are independent standard normal errors.

Using the new weight in \eqref{eq:CEweight} yields our second proposed estimator of $m(x)$ referred to as the complex error estimator, $\hat{m}_{\text{CE}}(x)=\text{atan2}[\hat{g}_{1,\text{CE}}(x),\hat{g}_{2,\text{CE}}(x)],$
where $\hat g_{1,\text{CE}}(x) = n^{-1}\sum_{j=1}^n\sin (\Theta_j) \mathcal{L}^*(W_j-x)$ estimates $m_1(x)$, and $\hat g_{2,\text{CE}}(w)=n^{-1}\sum_{j=1}^n \cos (\Theta_j)\mathcal{L}^*(W_j-x)$ estimate $m_2(x)$. {\revise In a separate study reported in \citet{Woolsey}, we extended this idea to account for non-normal measurement error by introducing hypercomplex quantities.}

\subsection{The one-step correction estimator}
\label{sec:OSest}
The common thread running through the first two strategies of correcting $\hat m^*(x)$ for measurement error is to adopt some weight that considers measurement error. Each new weight leads to non-naive estimators for some functions $g_1(x)$ and $g_2(x)$ such that $g_1(x)/g_2(x)=m_1(x)/m_2(x)$. Deviating from this theme of weight correction, we now propose a third strategy that exploits a one-step correction of a naive estimator for $g_\ell(x)$ as a whole, for $\ell\in \{1, 2\}$. This one-step correction is motivated by Lemma~\ref{lem:1steppf} that we prove in Appendix B. A similar result was also utilized by \cite{huangzhou} for local polynomial regression of a linear response with error-in-covariate. 
\begin{lemma} 
\label{lem:1steppf}
Define $g_\ell(x)=m_\ell(x)f_X(x)$ and  $g^*_\ell(x)=m^*_\ell(x)f_W(x)$, for $\ell=1, 2$, where $m^*_1(x)=\IE(\sin \Theta|W=x)$, $m^*_2(x)=\IE(\cos \Theta|W=x)$, and $f_W(x)$ is the probability density of $W$ evaluated at $x$. Then, assuming all integrals are well-defined, 
\begin{equation*}
g_\ell(x)=\frac{1}{2\pi} \int e^{-itx} \frac{\phi_{g_\ell^*}(t)}{\phi_U(t)}\, dt, \text{ for } \ell=1, 2. 
\end{equation*}
\end{lemma}
Motivated by Lemma~\ref{lem:1steppf}, we propose to estimate $g_\ell(x)$ by first obtaining a naive estimator of it that essentially estimates $g^*_\ell(x)$, then transforming the naive estimator using the integral transform suggested in Lemma~\ref{lem:1steppf}. In this article, we use a local linear estimator denoted by $\hat {m}_\ell^{*}(x)$ to estimate $m_\ell^*(x)$, for $\ell=1, 2$, and use the traditional kernel density estimator to estimate $f_W(x)$, denoted by $\hat f_W(x)$, leading to a naive estimator of $g_\ell(x)$ given by $\hat g_\ell^*(x)=\hat m^*_\ell(x) \hat f_W(x)$. Our third proposed estimator of $m(x)$ referred to as the one-step correction estimator is then given by $\hat m_{\text{OS}}(x)=\mbox{atan2}[\hat g_{1, \text{OS}}(x), \hat g_{2, \text{OS}}(x)]$, where
\begin{align}
    \hat g_{\ell,\text{OS}}(x) & =\frac{1}{2\pi} \int e^{-itx} \frac{\phi_{ \hat g_\ell^*}  (t)}{\phi_U(t)}\, dt 
    \label{eq:1steppf}
\end{align}
is an estimator of $g_\ell(x)=m_\ell(x)f_X(x)$, for $\ell=1, 2$. Henceforth, we use $\mathcal{T}_U(A)(x)$ to refer to the integral transform in Lemma~\ref{lem:1steppf} of a function $A:\mathbb{R}\to \mathbb{R}$, i.e., 
\begin{equation}
  \mathcal{T}_U(A)(x)=\frac{1}{2\pi} \int e^{-itx} \frac{\phi_{A}(t)}{\phi_U(t)} dt.
  \label{eq:Tu}
\end{equation}
Then \eqref{eq:1steppf} is equivalent to $\hat g_{\ell, \text{OS}}(x)=\mathcal{T}_U(\hat g^*_\ell)(x)$, for $\ell=1, 2$. 

\subsection{Comparisons between proposed estimators}
\label{sec:compare}
With three estimators of the circular regression function $m(x)$ accounting for measurement error proposed, some remarks on comparisons between them are in order. Among the three, $\hat m_{\text{CE}}(x)$ stands out as the only one that is developed under the assumption of normal measurement error. Both $\hat m_{\text{DK}}(x)$ and $\hat m_{\text{OS}}(x)$ are practically applicable as long as $\phi_U(t)$ never vanishes since  $\hat m_{\text{DK}}(x)$ depends on the integral  in \eqref{eq:LDK}, and $\hat m_{\text{OS}}(x)$ hinges on the integral in \eqref{eq:1steppf}, both integrals involve $\phi_U(\cdot)$ at the denominator of the integrand. 

By construction, $\hat m_{\text{OS}}(x)$ corrects naive estimation for measurement error at a higher level in the sense that, unlike $\hat m_{\text{DK}}(x)$ that accounts for measurement error by correcting the naive weight, $\hat m_{\text{OS}}(x)$ corrects naive estimation of $g_\ell (x)$ via transforming the naive estimator as a whole. Similar to $\hat m_{\text{DK}}(x)$, $\hat m_{\text{CE}}(x)$ results from correcting the normalized naive weight, but we correct it in its entirety instead of correcting it term-by-term as done in deriving the new weight for $\hat m_{\text{DK}}(x)$. In summary, we correct the naive estimator $\hat m^*(x)$ at an increasingly higher level as we progress from the first non-naive estimator to the third one in Section~\ref{sec:DKest}--\ref{sec:OSest}, i.e., from $\hat m_{\text{DK}}(x)$ to $\hat m_{\text{CE}}(x)$, then to $\hat m_{\text{OS}}(x)$.

On the other hand, $\hat m_{\text{DK}}(x)$ and $\hat m_{\text{OS}} (x)$ exploit integral transforms of a similar form, with the former using \eqref{eq:LDK} to correct a naive weight, while the latter utilizing \eqref{eq:1steppf} to correct a naive estimator of $g_\ell(x)$. Although no integral transform is involved in $\hat m_{\text{CE}}(x)$, it intrinsically relates to a similar integral transform because of Lemma~\ref{lem:addce2} that motivates it. To see this connection, consider the special case with $n=1$ in Lemma~\ref{lem:addce2} and let $V=U+A$, where $U$ is the measurement error. The key there to deriving an unbiased estimator of $g(X)$ based on  $W$ is to find a random variable $A$ such that all moments of $V=U+A$ are equal to zero, because then one has $\IE\{g(W+A)|X\}=\IE\{g(X+U+A)|X\}=\IE\{g(X+V)|X\}=g(X)$ according to Lemma~\ref{lem:addce2}. If one further assumes $A\perp U$, then  $\phi_V(t)=\phi_U(t)\phi_A(t)$, and thus $\phi_A(t)=\phi_V(t)/\phi_U(t)$. By the Fourier inversion theorem, the probability density function of $A$ is 
$f_A(a)=(2\pi)^{-1} \int e^{-ita} \phi_V(t)/\phi_U(t)\, dt$,
which is an integral transform similar to \eqref{eq:LDK} and \eqref{eq:1steppf}. Depending on the distribution of $U$, there may not exist such a $V$ (with all moments equal to zero); consequently, the so-obtained $f_A(a)$ may not be well-defined. But, it has been shown that, if $U\sim N(0, \sigma_u^2)$, then  letting $A=i\sigma_u Z$ gives rise to a desired $V$, with $Z\sim N(0, 1)$ and $Z\perp (U, X)$ \citep{stefanski1989unbiased}. In conclusion, all three proposed estimators have connections with some integral transform relating to deconvolution.  
After comparing the three proposed estimators in regard to assumptions, constructions, and rationales, we present more formally technical conditions imposed on these estimators next, under which we investigate their asymptotic properties.

\section{Asymptotic analysis}
\label{sec:asymp}
\subsection{Conditions and notations}
When studying properties of their proposed local polynomial estimators with error in covariates,  \cite{Delaigle} imposed different sets of conditions on the kernel $K(t)$ depending on the measurement error distribution characterized by $\phi_U(t)$. These conditions are also needed for our proposed estimators depending on the smoothness of $U$ \citep{fan1991asymptotic}. More specifically, if $\lim_{t\to \infty} t^\beta \phi_U(t)=c$ and $\lim_{t\to \infty} t^{\beta+1} \phi'_U(t)=-c\beta$ for some constants $c>0$ and $\beta>1$, then we say that (the distribution of) $U$ is ordinary smooth of order $\beta$. The Laplace distribution is an example of ordinary smooth distribution of order 2. If, for some positive constants $d_{0}$, $d_1$, $\gamma$, and $\beta$, and some real-valued constants $\beta_0$ and $\beta_1$,  
$d_{0}|t|^{\beta_{0}}\exp(-|t|^{\beta}/\gamma)\leq |\phi_U(t)|\leq d_{1}|t|^{\beta_{1}}\exp(-|t|^{\beta}/\gamma)$, as $t\to \infty$, then we say that $U$ is super smooth of order $\beta$. A normal distribution is super smooth of order 2. 

Because all proposed estimators are of the form $\hat{m}_{\cdot}(x)=\scalebox{.98}{$\text{atan2}[\hat{g}_{1,\cdot}(x),\hat{g}_{2,\cdot}(x)]$}$, where ``$\cdot$'' generically refers to an acronym (``DK'',  ``CE'', or ``OS'') relating to a proposed estimator, we first study asymptotic properties of $\hat{g}_{1,\cdot}(x)$ and $\hat{g}_{2,\cdot}(x)$. We then use these findings  to establish the asymptotic moments and distribution of $\hat m_{\cdot}(x)$ through a Taylor expansion of the atan2 function. Conditions frequently referenced in theorems presented next are listed below. Similar conditions are included in \cite{Delaigle} and \cite{huangzhou} in their asymptotic analyses.
\begin{itemize}
\item[]
\hspace{-0.35in}{\textbf{Condition O}}:  For $\ell= 0, 1, 2, 3$, $\|\phi_K^{(\ell)}(t)\|_\infty<\infty$ and $\int (|t|^\beta+|t|^{\beta-1}) |\phi_K^{(\ell)}(t)|dt<\infty$. For  $0\le k, \ell \le 2$, $\int |t|^{2\beta}  |\phi_K^{(k)}(t)||\phi_K^{(\ell)}(t)|dt<\infty$, and $\|\phi_U'(t)\|_\infty<\infty$.
\item[]
\hspace{-0.35in}{\textbf{Condition S}}: For $\ell= 0,1,2$, $\|\phi_K^{(\ell)}(t)\|_\infty<\infty$; $\phi_K(t)$ is supported on $[-1, 1]$.
\end{itemize}
{\revise The first set of conditions is imposed for asymptotic analysis when $U$ is \underline{o}rdinary smooth, and the latter set is imposed when $U$ is \underline{s}uper smooth.   For example, the characteristic function of a mean-zero Laplace distribution with variance $\sigma_u^2$, $\phi_U(t)=1/(1+\sigma_u^2 t^2/2)$, satisfies {\bf Condition O}. As for the kernel $K(\cdot)$, the Gaussian kernel satisfies constraints under {\bf Condition O} because, with $\phi_K(t)=e^{-t^2/2}$, $\phi_K(t)$ is continuously differentiable with bounded derivatives, and both integral constraints relate to integrals that define moments of normal or half normal distributions, which are well-defined at all orders. An example of kernels satisfying {\bf Condition S} that we use in our simulation experiments has its Fourier transform given by $\phi_K(t)=(1-t^2)^3\boldsymbol{1}_{\{-1 \le t\le 1\}}$.} Lastly, for $\ell\in \mathbb{N}_0$, define $\mu_\ell = \int t^\ell K(t) dt$ and $\nu_\ell = \int t^\ell K^2(t) dt$.

\subsection{Asymptotic results}
We derive the asymptotic bias and variance of  $\hat m_{\text{DK}} (x)$,  $\hat m_{\text{CE}} (x)$, and  $\hat m_{\text{OS}} (x)$ in Appendices D, E, and F, respectively. For ease of comparison, we subsume results relating to the asymptotic bias of these estimators in Theorem~\ref{thm:asympbias}, followed by results regarding their asymptotic variance summarized in Theorem~\ref{thm:asympvar}. 
\begin{theorem}
\label{thm:asympbias}
When $U$ is ordinary smooth of order $\beta$, then under Condition O, 
if $nh^{1+2\beta}\to \infty$ as $n\to \infty$ and $h\to 0$, 
\begin{enumerate}[(i)]
\item 
\begin{align*}
  &\ \text{Bias}
  \{\hat{m}_{\text{DK}}(x)|\bX\} \\ 
  = &\ \frac{h^{2}\mu_{2}}{2}\left\{m^{(2)}(x)+2m'(x)\frac{m_{1}(x)m'_{1}(x)+m_{2}(x)m'_{2}(x)}{m_{1}^{2}(x)+m_{2}^{2}(x)}\right\}+o(h^{2})\\
    &\ + \frac{f_W(x)}{nh^{1+2\beta}f^2_X(x)}\bigg[\frac{\{\eta(0,2)-\eta(1,1)\}\{m_1^*(x)m_2(x)-m_1(x)m_2^*(x)\}}{\mu_2 \left\{m_1^2(x)+m_2^2(x)\right\}}\\
    &\ - \left. \eta(0,0) \frac{m_1(x)m_2(x)\left\{\xi^*_1(x)-\xi^*_2(x)\right\}+\psi^*(x) \left\{ m_2^2(x)-m_1^2(x)\right\}}{\{m_1^2(x)+m_2^2(x)\}^2} \right]\\
    &\ +o_p\left(\frac{1}{nh^{1+2\beta}}\right),
\end{align*}
where $\xi_1^*(x)=\IE\{(\sin \Theta)^2|W=x\}$, $\xi_2^*(x)=\IE\{(\cos \Theta)^2|W=x\}$, \\$\psi^*(x)=\IE(\sin \Theta \cos \Theta|W=x)$, and, for $k, \ell\in \{0, 1, 2\}$, 
 $\eta(k, \ell)= \{i^{-(k+\ell)}/(2\pi c^2)\}$ \\$\int t^{2\beta} \phi_K^{(k)}(t) \phi_K^{(\ell)}(-t) dt$, in which $c$ is the positive constant appearing in the definition of ordinary smooth $U$;
 \item 
\begin{align*}
&\ \text{Bias}\{ \hat m_{\text{OS}}(x)|\bX \} \\
= &\ \frac{h^2 \mu_2}{2}\left[\frac{m_2(x)\mathcal{T}_U(M_1)(x)-m_1(x) \mathcal{T}_U(M_2)(x)}{f_X(x) \left\{m^2_1(x)+m^2_2(x) \right\}}\right] +o\left(h^2\right) \\
&\ +\frac{\eta(0,0)f_W(x)}{nh^{1+2\beta}f_X^2(x) \left\{m_1^2(x)+m^2_2(x)\right\}^2} \bigg[m_1(x)m_2(x) \left\{ \sigma^{*2}_1(x)-\sigma^{*2}_2(x)\right\} \\
&\ +\left\{m^2_2(x)-m^2_1(x)\right\}\left\{\psi^*(x)-m_1^*(x)m_2^*(x) \right\}\bigg]+o_p\left(\frac{1}{nh^{1+2\beta}} \right), 
\end{align*}
where $\mathcal{T}_U$ is the integral transform defined in \eqref{eq:Tu}, \\$M_\ell(x) = m^*_\ell(x)f^{(2)}_W(x)+m^{*(2)}_\ell(x)f_W(x)$, for $\ell=1, 2$, $\sigma^{*2}_1(x)=\text{Var}(\sin \Theta|W=x)$, and $\sigma^{*2}_2(x)=\text{Var}(\cos \Theta|W=x)$. 
\end{enumerate}
When $U$ is super smooth of order $\beta$, if  $nh^{1-2\beta_2}\exp(-2h^{-\beta}/\gamma)\to \infty$ as $n\to \infty$ and $h\to 0$, where $\beta_2=\beta_0I(\beta_0<0.5)$, then, under Condition S,
\begin{enumerate}[(i)]
\item 
\begin{align*}
&\ \text{Bias}\left\{\hat{m}_{\text{DK}}(x)|\bX\right\} \\
=&\ \frac{h^2 \mu_2}{2}\left\{m^{(2)}(x)+2m'(x) \frac{m_1(x)m'_1(x)+m_2(x)m_2'(x)}{m_1^2(x)+m_2^2(x)} \ \right\}+o(h^2)\\
  &\ + O_p\left(\frac{\exp(2h^{-\beta}/\gamma)}{nh^{1-2\beta_2}}\right);
\end{align*}
  \item
  \begin{align*}
&\ \text{Bias}\left\{\hat m_{\text{OS}}(x)|\bX\right\} \\
 = &\ \frac{h^2 \mu_2}{2}\left[ \frac{m_2(x)\mathcal{T}_U(M_1)(x)-m_1(x) \mathcal{T}_U(M_2)(x)}{f_X(x) \left\{m^2_1(x)+m^2_2(x) \right\}}\right] +o(h^2)\\
&\  +O_p\left( \frac{\exp(2h^{-\beta}/\gamma)}{nh^{1-2\beta_2}}\right);
 \end{align*}
 \item more specifically assuming $U\sim N(0, \sigma_u^2)$, and thus $\beta=2$, $\beta_2=0$, and $\gamma=2/\sigma_u^2$, then, if $B/n$ tends to a positive constant as $n\to \infty$, 
\begin{align*}
&\ \text{Bias}\left\{\hat{m}_{\text{CE}}(x)|\bX\right\}\\
= &\ \frac{h^2 \mu_2}{2}\left\{m^{(2)}(x)+2m'(x) \frac{m_1(x)m'_1(x)+m_2(x)m_2'(x)}{m_1^2(x)+m_2^2(x)}\right\} \\
&\ +\frac{h^2}{f_X(x)\left\{m_1^2(x)+m_2^2(x)\right\}^2}\Big[ 2 (1-\mu_2) f_X'(x) m_1(x) m_2(x) \left\{m_1(x) m'_1(x)\right.\\
&\ \left. -m_2(x)m'_2(x)\right\}
+ \mu_2 \left\{f_X(x)- f'_X(x)\right\}\left\{m_1^2(x)-m_2^2(x) \right\}\left\{m_1(x)m'_2(x)\right.\\
&\ \left. +m'_1(x)m_2(x)\right\}\Big] +o(h^2)
+O_p\left(\frac{\exp(\sigma_u^2 h^{-2})}{nh}\right).
\end{align*}
\end{enumerate}
\end{theorem}
According to Theorem~\ref{thm:asympbias}, the dominating bias of each proposed estimator consists of a term of order $O(h^2)$ and a term of order depending on the smoothness of $U$. In particular, the term of order $O(h^2)$ in the dominating bias of $\hat m_{\text{DK}}(x)$ coincides with that of $\tilde m(x)$ \citep[see Theorem 4 in][]{di2013non}, which depends on the curvature of $m(x)$. The term that depends on the smoothness of $U$ can involve naive counterparts of $m_\ell(x)$ and other functionals, which is dominated by the first term provided that $n\to \infty$ much faster than $h\to 0$, and the requirement on $n$ to achieve this is more demanding when  $U$ is super smooth than when it is ordinary smooth, or when the order of smoothness $\beta$ is larger. 
\begin{theorem}
    \label{thm:asympvar}
When $U$ is ordinary smooth of order $\beta$, under Condition O, 
if $nh^{1+2\beta}\to \infty$ as $n\to \infty$ and $h\to 0$, then  
\begin{enumerate}[(i)]
    \item 
\begin{align*}
    &\ \text{Var}\left\{\hat{m}_{\text{DK}}(x)|\bX\right\} \\
    & = \frac{f_W(x)\left\{ m_1^2(x)\xi^*_2(x)+m_2^2(x)\xi_1^*(x)-2m_1(x)m_2(x) \psi^*(x)\right\}\eta(0,0)}{nh^{1+2\beta}f_X^2(x)    \left\{m_1^2(x)+m_2^2(x)\right\}^2}\\
    &\ +o_p\left(\frac{1}{nh^{1+2\beta}}\right);
    \end{align*}
    \item 
    \begin{align*}
    &\ \text{Var}\left\{\hat m_{\text{OS}}(x)|\bX\right\} \\
    = &\ \frac{f_W(x) \eta(0,0)}{nh^{1+2\beta} f^2_X(x) \left\{m^2_1(x)+m^2_2(x)\right\}^2}\bigg( m^2_1(x) \xi^*_2(x)+m^2_2(x) \xi^*_1(x)\\ 
    &\ -2m_1(x)m_2(x) \psi^*(x) 
-\left\{m_1(x)m^*_2(x)-m^*_1(x) m_2(x) \right\}^2\\ 
&\  +4m(x)\bigg[m_1(x)m_2(x) \left\{\sigma^{*2}_2(x)-\sigma^{*2}_1(x)\right\}+\phi^*(x) \{m^2_1(x)\\ 
&\ -m^2_1(x)\}\bigg]\bigg) +o_p\left( \frac{1}{nh^{1+2\beta}}\right),
\end{align*}
where $\phi^*(x)=\text{Cov}(\sin \Theta, \cos\Theta|W=x)$.
\end{enumerate}
When $U$ is super smooth of order $\beta$, if $nh^{1-2\beta_2}\exp(-2h^{-\beta}/\gamma)\to \infty$ as $n\to \infty$ and $h\to 0$, then under Condition S, 
\begin{enumerate}[(i)]
    \item $\text{Var}\{\hat m_{\text{DK}}(x)|\bX\}=O_p\left(\displaystyle{\frac{\exp(2h^{-\beta}/\gamma)}{nh^{1-2\beta_2}}}\right)$; 
    \item $\text{Var}\{\hat m_{\text{OS}}(x)|\bX\}=O_p\left(\displaystyle{\frac{\exp(2h^{-\beta}/\gamma)}{nh^{1-2\beta_2}}}\right)$; 
    \item more specifically assuming $U\sim N(0, \sigma_u^2)$, and thus $\beta=2$, $\beta_2=0$, and $\gamma=2/\sigma_u^2$, then, if $B/n$ tends to a positive constant as $n\to \infty$,
\begin{align*}
&\ \text{Var}\{\hat m_{\text{CE}}(x)|\bX\} \\
= &\ \scalebox{.98}{$2h^2\left\{(2\mu_2-1)\frac{f_X'(x)}{f_X(x)}-\mu_2\right\}\frac{m_1(x)m_2(x)\left\{m_1(x)m_2'(x)+m'_1(x)m_2(x)\right\}}{\left\{ m_1^2(x)+m^2_2(x)\right\}^2}$}\\
     &\ +o(h^2)+O_p\left(\frac{\exp\left(\sigma_u^2h^{-2}\right)}{nh}\right).
\end{align*}
\end{enumerate}
\end{theorem}
According to Theorem~\ref{thm:asympvar}, having super smooth $U$ tends to lead to a more variable $\hat m_{\cdot}(x)$. When $U$ is ordinary smooth, the dominating variance of $\hat m_{\text{OS}} (x)$ is equal to that of $\hat m_{\text{DK}}(x)$ plus additional terms depending on $m(x)$, $m_\ell(x)$, and other functions relating to naive inference. Having these additional terms implies that the variability of $\hat m_{\text{OS}}(x)$ is more affected by discrepancies between functionals of actual interest when drawing inference and their naive counterparts, such as how $m_\ell(x)$ differs from $m^*_\ell(x)$. Our asymptotic results regarding $\hat m_{\text{DK}}(x)$ elaborate on the impacts of measurement error when compared with findings in \citet[][see Result 5 in Section 4.1.2]{di2023kernel} for a similar estimator {\revise that uses a local constant weight in place of our local linear weight. Similar to the comparison between the Nadaraya-Watson estimator and a local linear estimator of a regression function for a linear response \citep[see Table 2.1 in][]{Fan}, the dominating variance of the estimator in \citet{di2023kernel} is of the same order as that of $\hat m_{\text{DK}}(x)$; and, although of the same order (of $O(h^2)$), the dominating bias of the former has an extra term given by $2m'(x)f_X'(x)/f_X(x)$ compared with the multiplier of $h^2\mu_2/2$ in the asymototic bias of $\hat m_{\text{DK}}(x)$.} {\revise Besides the results summarized in Theorems~\ref{thm:asympbias} and \ref{thm:asympvar}, we also derived in Appendices E and F more general results for $\hat m_{\text{CE}}(x)$ and $\hat m_{\text{OS}}(x)$ when local polynomial weights of order $p$ are used in general.}

Finally, we establish asymptotic normality of $\hat{m}_{\cdot}(x)$ in Appendix G by exploiting established normality results in existing literature regarding the same type of estimators as $\hat{g}_{1, \cdot}(x)$ and $\hat{g}_{2, \cdot}(x)$ involved in our estimators. 

\section{Bandwidth selection}
\label{sec:band}

\subsection{Bandwidth selection in the absence of measurement error}
\label{sec:bandnome}
The performance of local polynomial estimators is sensitive to the bandwidth $h$. 
Compared to deriving a plug-in type of bandwidth, a more practically feasible and widely applicable bandwidth selection method is cross validation (CV) based a risk function for assessing estimation quality, such as $k$-fold or leave-one-out CV. We focus on 5-fold CV in this article for illustration purposes. 

In the context of estimating a circular regression function $m(x)$ in the absence of measurement error, a sensible risk function associated with the estimator $\tilde m(x)$ is the mean cosine dissimilarity, $\IE\{1-\cos(\Theta-\tilde m (X))\}$. To highlight the dependence of $\tilde m(x)$ on $h$ and the data used to construct the estimator, we re-express the estimator in \eqref{eq:mtilde} as $\tilde m(x; \bTheta, \bX, h)$.
A 5-fold CV entails first randomly splitting the raw data into five data sets of (approximately) equal size, then computing the following loss function formulated using cosine dissimilarity at each candidate $h$,  
\begin{equation}
    D(\tilde m(X); \bX, h) = \sum_{k=1}^5 \frac{1}{|I_k|} \sum_{j\in I_k} \left\{ 1- \cos ( \Theta_j -\tilde m(X_j; \bTheta^{(-k)}, \bX^{(-k)}, h))\right\},    \label{eq:idealCV}
\end{equation}
where $I_k$ is the index set of the $k$-th subsample of size $|I_k|$, and $(\bTheta^{(-k)}, \bX^{(-k)})$ include all data in the original sample except for those in the $k$-th subsample, for $k=1, \ldots, 5$. 
Then $h_{\text{true}}=\argmin_{h\in \mathcal{H}}D(\tilde m(X); \bX, h)$ is the chosen bandwidth based on true covariate data, where $\mathcal{H}$ is the set of candidate values of $h$.

In the presence of covariate measurement error, an immediate hurdle in implementing the traditional CV outlined above is that, for a non-naive estimator of the regression function $\hat m_{\cdot}(x)$, the cosine dissimilarity between $\Theta_j$ and $\hat m_{\cdot}(X_j)$ cannot be computed now that $X_j$ is unobserved but $W_j$ is instead. Naively substituting $X_j$ with $W_j$ again does not solve the problem here because the cosine dissimilarity between $\Theta_j$ and $\hat m_{\cdot}(W_j)$ is a misrepresentation of the actual loss associated with the $j$-th observation. We present two approaches to overcome the hurdle next. 
\subsection{Bandwidth selection using SIMEX}
\label{sec:bandsimex}
The first approach is proposed by \cite{delaiglehall} as a combination of simulation-extrapolation \citep[SIMEX,][Chapter 5]{carroll2006measurement} and CV. Ideally, one would follow the traditional CV to find 
\begin{equation}
\label{eq:idealh}
h_{\text{ideal}} =\argmin_{h \in \mathcal{H}}\sum_{k=1}^5 \frac{1}{|I_k|} \sum_{j\in I_k} \left\{ 1- \cos ( \Theta_j -\hat m_{\cdot}(X_j; \bTheta^{(-k)}, \bW^{(-k)}, h))\right\},
\end{equation}
where we stress that $\hat  m_{\cdot}(x)$ is a non-naive estimator of $m(x)$ constructed using error contaminated data, but we evaluate $x$ at $X_j$'s when computing the loss in \eqref{eq:idealh}. Compared with \eqref{eq:idealCV}, a major distinction in the ideal CV in the presence of measurement error is the mismatch between the training data and the validation data exhibited in \eqref{eq:idealh}: while {\it error-contaminated} data $(\bTheta^{(-k)}, \bW^{(-k)})$ are used to train the estimator $\hat m_{\cdot}(x)$, the {\it error-free} data $(\bTheta^{(k)}, \bX^{(k)})$ are used to validate the estimator. Clearly, the ``ideal'' bandwidth $h_{\text{ideal}}$ is not attainable with $\bX$ now unobserved. Naive CV yields 
$$
h_{\text{naive}} 
 =\argmin_{h \in \mathcal{H}}\sum_{k=1}^5 \frac{1}{|I_k|} \sum_{j\in I_k} \left\{ 1- \cos ( \Theta_j -\hat m_{\cdot}(W_j; \bTheta^{(-k)}, \bW^{(-k)}, h))\right\},
$$
which is expected to distort $h_{\text{ideal}}$ due to evaluating $x$ in $\hat m_{\cdot}(x)$ at $W_j$'s instead of $X_j$'s. To learn how $h_{\text{ideal}}$ takes into consideration that, even though the estimator $\hat m_{\cdot}(x)$ is constructed based on error-contaminated data, $\hat m_{\cdot}(x)$ should be evaluated at error-free data when computing the loss, we ``simulate'' twice the ideal CV mimicking \eqref{eq:idealh} to obtain two bandwidths. 

The first simulated version of ideal bandwidth is  
\begin{equation}
\label{eq:h1}
h_1 =\argmin_{h \in \mathcal{H}} \frac{1}{B}\sum_{b=1}^B \sum_{k=1}^5 \frac{1}{|I_k|} \sum_{j\in I_k} \left\{ 1- \cos ( \Theta_j -\hat m^*_{\cdot}(W_j; \bTheta^{(-k)}, \bW^{*(-k)}_b, h))\right\},
\end{equation}
where, for $b=1, \ldots, B$, $\bW^*_b=(W^*_{1,b}, \ldots, W^*_{n,b})^\top=\bW+\bU^*_b$, {\revise in which $\bU^*_b=(U^*_{1, b}, \ldots,$ $ U^*_{n,b})^\top$ is a random sample generated from the distribution that errors in $\bU=(U_1, \ldots, U_n)^\top$ follow}, and $\hat m^*_{\cdot}(x)$ is the same type of estimator as $\hat m_{\cdot}(x)$ but for estimating $m^*(x)$. For example, when selecting a bandwidth used in the deconvoluting kernel estimator, the estimator in \eqref{eq:idealh} is $\hat m_{\text{DK}}(x)$, which is constructed based on $(\bTheta, \bW)$ to estimate $m(x)$ as the circular mean of $\Theta$ given $X=x$. Accordingly, the estimator in \eqref{eq:h1} is $\hat m^*_{\text{DK}}(x)$ as the deconvoluting kernel estimator constructed based on $(\bTheta, \bW^*_b)$ to estimate $m^*(x)$ as the mean of $\Theta$ given $W=x$. The CV leading to $h_1$ is parallel to that leading to $h_{\text{ideal}}$ in the sense that, like $\bW$ being a noisy surrogate of $\bX$ in \eqref{eq:idealh}, $\bW^*_b$ is also a noisy surrogate of $\bW$ in \eqref{eq:h1} with the same severity of error contamination since $\text{Var}(W^*_{j,b}|W_j)=\text{Var}(W_j|X_j)=\sigma_u^2$; and, more importantly, \eqref{eq:h1} also has a mismatch between training data and validation data, with the former involving the (further) contaminated data $\bW_b^{*(-k)}$, whereas the latter being the data before (further) contamination, $\bW^{(k)}$. Similarly, the second simulated bandwidth results from yet another round of CV parallel to that producing $h_1$, 
$$h_2 =\argmin_{h \in \mathcal{H}} \frac{1}{B}\sum_{b=1}^B \sum_{k=1}^5 \frac{1}{|I_k|} \sum_{j\in I_k} \left\{ 1- \cos ( \Theta_j -\hat m^{**}_{\cdot}(W^*_{j,b}; \bTheta^{(-k)}, \bW^{**(-k)}_b, h))\right\},$$
where $\bW_b^{**}=\bW^*_b+\bU^{**}_b$, $\bU^{**}_b$ is generated in the same way as $\bU^*_b$, with $\bU^{**}_b$, $\bU^*_b$, and $\bU$ all independent of each other, for $b=1, \ldots, B$, and $\hat m^{**}_{\cdot}(\cdot)$ is the same type of estimator as $\hat m_{\cdot}(\cdot)$, but for estimating the regression function when regressing $\Theta_j$'s on $W^*_{j,b}$'s. 

The simulation step of SIMEX ends with outputting $h_1$ and $h_2$. At the extrapolation step of SIMEX, one attempts to recover $h_{\text{ideal}}$ by learning from the connection between $h_1$ and $h_2$ because, with the parallel design in the two rounds of ideal CV leading to $h_1$ and $h_2$, how $h_1$ compares with $h_2$ should be similar to how $h_{\text{ideal}}$ compares with $h_1$. In particular, \cite{delaiglehall} showed that $\log h_{\text{ideal}}-\log h_1 \approx \log h_1 -\log h_2$ when $\sigma_u$ is small, and thus proposed to approximate $h_{\text{ideal}}$ by $h_{\text{SIMEX}}=h_1^2/h_2$. This completes the process of bandwidth selection that we refer to as CV-SIMEX that yields $h_{\text{SIMEX}}$ as the selected bandwidth. An obvious concern of this method is the computational burden. Even the traditional CV without the measurement error complication is computationally demanding in general; and this method requires $2\times B$ rounds of such CV based on data noisier than the original data.  

\subsection{Bandwidth selection using complex error}
\label{sec:bandce}
Acknowledging that $h_{\text{ideal}}$ in \eqref{eq:idealh} is not attainable solely because the loss function there cannot be computed based on observed data $(\bTheta, \bW)$, we propose a second approach for bandwidth selection where we first estimate this loss function, then we choose a bandwidth by minimizing the estimated loss. 

Recall that the loss function in \eqref{eq:idealh} is 
\begin{equation}
D(\hat m_{\cdot}(X); \bW, h)=\sum_{k=1}^5 \frac{1}{|I_k|} \sum_{j\in I_k} \left\{ 1- \cos ( \Theta_j -\hat m_{\cdot}(X_j; \bTheta^{(-k)}, \bW^{(-k)}, h))\right\}.
\label{eq:idealD}
\end{equation}
With observed data $(\bTheta, \bW)$ plugged in, along with a candidate $h$, \eqref{eq:idealD} is unknown purely due to its dependence on $\bX$. If we view this unknown loss as a function of $\bX$, like $g(\bX)$ in Lemma~\ref{lem:addce2}, then we can formulate an unbiased estimator of this loss function when $U\sim N(0, \sigma_u^2)$. This yields an estimated loss given by 
\begin{equation}
\label{eq:Dest}
\begin{aligned}
&\ 
\hat D(\hat m_{\cdot}(W^*); \bW, h) \\
= &\ \frac{1}{B}\sum_{b=1}^B\sum_{k=1}^5 \frac{1}{|I_k|} \sum_{j\in I_k} \left\{ 1- \cos ( \Theta_j -\hat m_{\cdot}(W^*_{j,b}; \bTheta^{(-k)}, \bW^{(-k)}, h))\right\},
\end{aligned}
\end{equation}
where $\{W^{*}_{j,b}=W_j+i\sigma_u Z_{j,b}, \, j=1, \ldots, n\}_{b=1}^B$ are generated in the same way as the complex-valued covariate data in \eqref{eq:CEweight}. If $D(\hat m_{\cdot}(X); \bW, h)$ in \eqref{eq:idealD} is an entire function of $\bX$, and $\bW$ results from additive normal error contamination of $\bX$, then, by Lemma~\ref{lem:addce2}, $\hat D(\hat m_{\cdot}(W^*); \bW, h)$ in \eqref{eq:Dest} is an unbiased estimator of $D(\hat m_{\cdot}(X); \bW, h)$. Minimizing $\hat D(\hat m_{\cdot}(W^*); \bW, h)$ with respect to $h$ gives a bandwidth that is expected to improve over $h_{\text{naive}}$ and be closer to $h_{\text{ideal}}$. Compared with CV-SIMEX, this bandwidth selection procedure is less cumbersome but it still demands heavy computation because, for each of the five validation data sets, one estimates $m(x)$ based on the corresponding training data at $|I_k|\times B$ complex-valued $x$'s, i.e., at $\{W^*_{j,b}, \, j\in I_k\}_{b=1}^B$. Next, we propose a revised procedure to drastically lighten the computational burden. The idea is to move the averaging in \eqref{eq:Dest} ``$B^{-1}\sum_{b=1}^B$'' further inside the summand so that one can use the trick inspired by Lemma~\ref{lem:addce2} to estimate $\hat m_{\cdot}(X_j; \bTheta^{(-k)}, \bW^{(-k)}, h)$ in \eqref{eq:idealD} instead of estimating the loss function as a whole. 

Take the complex error estimator $\hat m_{\text{CE}}(x)$ as an example. In \eqref{eq:Dest}, the estimate of \\ $\hat m_{\text{CE}}(X_j; \bTheta^{(-k)}, \bW^{(-k)}, h)$ is, for each $b\in \{1, \ldots, B\}$,  
\begin{equation}
\label{eq:CEest4CV}
\begin{aligned}
&\ \hat m_{\text{CE}}(W^*_{j,b}; \bTheta^{(-k)}, \bW^{(-k)}, h) \\
= &\ \text{atan2}\left[\hat g_{1,\text{CE}}(W^*_{j,b}; \bTheta^{(-k)}, \bW^{(-k)}, h), \, \hat g_{2,\text{CE}}(W^*_{j,b}; \bTheta^{(-k)}, \bW^{(-k)}, h)\right], 
\end{aligned}
\end{equation}
where $\hat g_{\ell,\text{CE}}(W^*_{j,b}; \bTheta^{(-k)}, \bW^{(-k)}, h)$ is the estimate $\hat g_{\ell,\text{CE}}(W^*_{j,b})$ based on data $(\bTheta^{(-k)}, \bW^{(-k)})$ with the bandwidth set at $h$, for $\ell=1, 2$. By moving the outer averaging in \eqref{eq:Dest} towards $\hat m_{\text{CE}}(\cdot)$ and further passing through the atan2 function in \eqref{eq:CEest4CV}, we have a different estimate of the loss function given by 
\begin{equation}
\label{eq:Dest2}
\begin{aligned}
&\ \hat D^*(\hat m_{\text{CE}}(W^*); \bW, h)\\
=&\ \sum_{k=1}^5 \frac{1}{|I_k|} \sum_{j\in I_k} \left\{ 1- \cos \left( \Theta_j - \text{atan2}\left[\frac{1}{B}\sum_{b=1}^B \hat g_{1,\text{CE}}(W^*_{j,b}; \bTheta^{(-k)}, \bW^{(-k)}, h),  \right.\right.\right. \\
&\ \left. \left. \left. \frac{1}{B}\sum_{b=1}^B \hat g_{2,\text{CE}}(W^*_{j,b}; \bTheta^{(-k)}, \bW^{(-k)}, h)\right]\right) \right\}, 
\end{aligned}
\end{equation}
where, using the weight function $\mathcal{L}^*(\cdot)$ defined in \eqref{eq:CEweight}, the first argument of $\text{atan2}[\cdot, \cdot]$ is
\begin{align}
 &\ \frac{1}{B}\sum_{b=1}^B \hat g_{1,\text{CE}}(W^*_{j,b}; \bTheta^{(-k)}, \bW^{(-k)}, h) \nonumber \\
 = &\ \frac{1}{B}\sum_{b=1}^B \frac{1}{|I^{(-k)}|} \sum_{\ell \in I^{(-k)}} \sin (\Theta_\ell) \mathcal{L}^* (W _\ell-W^*_{j,b}) \nonumber \\
= &\  \frac{1}{|I^{(-k)}|} \sum_{\ell \in I^{(-k)}} \sin (\Theta_\ell)\overline{\mathcal{L}^*}_{\ell,j}, \label{eq:avg1}
\end{align}
in which $I^{(-k)}$ is the index set of the training data $(\bTheta^{(-k)}, \bW^{(-k)})$ of size $|I^{(-k)}|$, and $\overline{\mathcal{L}^*}_{\ell,j}=B^{-1}\sum_{b=1}^B \mathcal{L}^* (W_\ell-W^*_{j,b})$. In other words, \eqref{eq:avg1} is an estimate of the same form as $\hat g_{1,\text{CE}}(\bW^*_{j,b}; \bTheta^{(-k)}, \bW^{(-k)}, h)$ in \eqref{eq:CEest4CV}, but uses a different weight, $\overline{\mathcal{L}^*}_{\ell,j}$. Similarly, the second argument of $\text{atan2}[\cdot, \cdot]$ in \eqref{eq:Dest2} is an estimate of the same form as $\hat g_{2,\text{CE}}(\bW^*_{j,b}; \bTheta^{(-k)}, \bW^{(-k)}, h)$ in \eqref{eq:CEest4CV} but uses the weight $\overline{\mathcal{L}^*}_{\ell,j}$. Define $\hat g_{1,\text{CE}}(\widetilde \bW^*_j; \bTheta^{(-k)}, \bW^{(-k)}, h)$ as the estimate in \eqref{eq:avg1}, and similarly define $\hat g_{2,\text{CE}}(\widetilde \bW^*_j; \bTheta^{(-k)}, \bW^{(-k)}, h)$, where $\widetilde \bW^*_j=(W^*_{j,1}, \ldots, W^*_{j,B})^\top$; and using these two estimates we have the corresponding estimate of the regression function, denoted by $\hat m_{\text{CE}}(\widetilde \bW^*_j; \bTheta^{(-k)}, \bW^{(-k)}, h)$. Now the new estimated loss in \eqref{eq:Dest2} can be re-expressed as
\begin{equation}
\label{eq:Dest2short}
\begin{aligned}
&\ \hat D^*(\hat m_{\text{CE}}(W^*); \bW, h) \\
= &\ \sum_{k=1}^5 \frac{1}{|I_k|} \sum_{j\in I_k} \left\{ 1- \cos ( \Theta_j -\hat m_{\text{CE}}(\widetilde \bW^*_j; \bTheta^{(-k)}, \bW^{(-k)}, h))\right\}.
\end{aligned}
\end{equation}
It is clear from \eqref{eq:Dest2short} that $\hat D^*(\hat m_{\text{CE}}(W^*); \bW, h)$ is similar to $D(\hat m_{\text{CE}}(X); \bW, h)$ in \eqref{eq:idealD} in terms of computation burden, and clearly much less burdensome than the initial estimated loss $\hat D( \hat m_{\text{CE}}(W^*); \bW, h)$ in \eqref{eq:Dest}.   

Admittedly, by passing the averaging operation through nonlinear functions like we did to modify $\hat D( \hat m_{\text{CE}}(W^*); \bW, h)$, we introduce bias as an estimator of $D( \hat m_{\text{CE}}(X); \bW, h)$. This is a small price we pay for a substantial computational gain. Regardless, by construction, $\hat D^*( \hat m_{\text{CE}}(W^*); \bW, h)$ is a sensible estimator that consistently estimates $D( \hat m_{\text{CE}}(X); \bW, h)$ given $(\bTheta, \bW)$ as $B\to \infty$. Consequently, $h_{\text{CE}}=\argmin_{h\in \mathcal{H}}\hat D^*( \hat m_{\text{CE}}(W^*); \bW, h)$ still tends to be closer to $h_{\text{ideal}}$ than  $h_{\text{naive}}$ is. We call this proposed bandwidth selection method that involves complex error CV-CE, henceforth.  

{\revise All bandwidth selection methods considered here require a pre-specified set of candidate values $\mathcal{H}$. A practically effective approach to specify $\mathcal{H}$ is to first find an optimal bandwidth $h_0$ for a simpler (to compute) estimator, such as the naive estimator. Then one creates a range of bandwidths surrounding $h_{0}$, for instance, $[0.8h_{0}, \, 1.3h_{0}]$. If it is always, say, the upper bound that is selected in a cross-validation for a non-naive estimator, which can happen in the presence of severe error contamination in the covariate data, we recommend adjusting the initial search window by enlarging the upper bound slightly. We find this practice of specifying $\mathcal{H}$ via trial-and-error more effective and feasible than alternative approaches, such as those based on minimizing the mean integrated squared error (MISE) of a non-naive estimator. As we show in Appendix H, even when the asymptotic MISE can be derived in some model settings, unknown functionals related to $m(x)$ that the asymptotic MISE depends on make it a practically difficult to use quantity for the search of suitable bandwidths.}   


\section{Simulation study}
\label{sec:simu}
\subsection{Design of simulation}
\label{sec:simudesign}
We are now in the position to inspect the finite sample performance of our proposed non-naive estimators for the circular regression function $m(x)$. Simulation experiments presented in this section are designed to address three issues: (i) {\revise comparisons between four non-naive estimators, $\hat m_{\text{DK}}(x)$, $\hat m_{\text{CE}}(x)$, $\hat m_{\text{OS}}(x)$, and the deconvoluting kernel estimator with a local constant weight proposed by \citet{di2023kernel}}, and the naive estimator $\hat m^*(x)$, using the ideal estimator $\tilde m(x)$ as a benchmark; (ii) comparisons between different bandwidth selection methods, using some optimal bandwidth $h_{\text{opt}}$ to be defined momentarily as a benchmark; (iii) impacts of the measurement error distribution on different estimators. 

{\revise For each Monte Carlo replicate, we first generate true covariate data \scalebox{.99}{$\{X_j\}_{j=1}^n$} of size $n\in\{50, 100, 250\}$ from a covariate distribution; then we generate circular responses $\{\Theta_j\}_{j=1}^n$ according to \eqref{eq:thetagivenX} with a regression function $m(x)$ we design and $\epsilon_j\sim \text{von Mises}(0, 3)$; lastly, error-contaminated covariate data $\{W_j\}_{j=1}^n$ are generated according to \eqref{eq:WgivenX} with $U$ following a mean-zero normal distribution or Laplace distribution, and $\sigma_u^2$ set at a value to achieve a reliability ratio $\lambda\in \{0.8, 0.9\}$, where $\lambda=\text{Var}(X)/\text{Var}(W)$. In particular, two covariate distributions are considered, $\text{uniform}(-5, 5)$ and $N(0, 4)$; two regression functions are used, $m(x)=2\text{atan} (x)$ and $m(x)=2\text{atan}(1/x)$. The first regression function is simpler from the perspective of model fitting because it is smooth and monotone; the second regression function brings forth a more challenging model due to the singularity of $m(x)$ at $x=0$ that creates a jump discontinuity.}

 At each simulation setting specified by {\revise $m(x)$, the covariate distribution,} the value of $n$, and the distribution of $U$, we generate 100 data sets, based on each of which we implement estimation procedures determined by the choice of estimator and the bandwidth selection method. For each estimation procedure, we estimate $m(x)$ at a predetermined grid of values evenly spaced within the support of $X$, $x_1, \ldots, x_n$, over which $m(x)$ exhibits major features such as different degrees of curvature. The metric employed to assess the quality of estimation using $\hat m_{\cdot}(x)$ is the empirical mean cosine dissimilarity as an empirical risk function, $D_n (\hat m_{\cdot}(x); \bW, h)=1-n^{-1}\sum_{j=1}^n \cos (m(x_j)-\hat m_{\cdot}(x_j))$. The aforementioned optimal bandwidth $h_{\text{opt}}$ minimizes $D_n (\hat m_{\cdot}(x); \bW, h)$. 

Finally, we comment on some details imperative for the actual implementation of each estimation procedure. When computing the complex-valued weight $\mathcal{L^*}(W_j-x)$ in \eqref{eq:CEweight} for $\hat m_{\text{CE}}(x)$, we set $B^*=250$, and we extract the real part of $\mathcal{L^*}(W_j-x)$ as the actual weight because, by construction, the imaginary part of $\mathcal{L^*}(W_j-x)$ has a mean of zero. {\revise We choose for all three proposed estimators the kernel  
$K(x) = 96\{x(x^2-15)\cos x+3(5-2x^2)\sin x\}/(2\pi x^{7})$, with 
$\phi_{K}(t) = (1-t^{2})^{3} \boldsymbol{1}_{\{-1\leq t \leq 1\}}$. Although this particular choice of the kernel is made to satisfy the more stringent {\bf Condition S} to safeguard against ill-defined integral transforms in \eqref{eq:LDK} and \eqref{eq:1steppf} when $U$ is super smooth, computing these integrals with this $\phi_{K}(t)$ also tends to be more numerically stable when $U$ is ordinary smooth. More specifically, we employ the fractional fast Fourier transform \citep{bailey} to compute the integral transforms in \eqref{eq:LDK} and \eqref{eq:1steppf} as described in detail in \cite{huangzhou}.} When implementing CV-SIMEX and CV-CE to select a bandwidth, we set $B=30$.

\subsection{Simulation results}

{\revise Figure~\ref{fig:box} provides boxplots of the empirical risk associated with each considered estimation procedure as $n$ varies when $m(x)=2\text{atan}(1/x)$, $X\sim N(0, 4)$, and $U\sim N(0, \sigma^2_u)$. It is evident from Figure~\ref{fig:box} that all non-naive estimation procedures outperform naive estimation. Among the non-naive estimators when $h_{\text{opt}}$ is used, $\hat{m}_{\text{CE}}(x)$ outperforms the rest in this more challenging model fitting exercise. This can be thanks to the (corrected) local linear weight used in $\hat{m}_{\text{CE}}(x)$ that is unbiased, in contrast to the other non-naive estimators with weights that are biased (although consistent) estimators for their error-free counterparts. Between the two types of deconvoluting kernel estimators, our estimator $\hat{m}_{\text{DK}}(x)$ that uses the local linear weight improves over the estimator with the local constant weight. The two bandwidth selection methods, CV-SIMEX and CV-CE, for choosing a bandwidth in $\hat{m}_{\text{CE}}(x)$ produce mostly comparable results. However, CV-CE substantially shortens the computation time.} {\reviseagain This is evidenced in Table \ref{timing} that presents computation times for implementing different bandwidth selection methods, with $\mathcal{H}$ containing 50 candidate bandwidths, on a computer with an Intel Core i5 12600KF processor and a core speed of 4.5 GHz, followed by estimating $m(x)$ via various estimators based on one Monte Carlo replicate data set of size $n \in \{50, 100, 250\}$.} Although time-consuming, CV-SIMEX substantially outperforms the naive CV as evidenced in a simulation study presented in Appendix I where we fitted the same regression model using $\hat m_{\text{DK}}(x)$ with bandwidths chosen by CV-SIMEX and the naive CV.  

{\revise Figure~\ref{fig:box3} provides results under the same simulation setting as that for Figure~\ref{fig:box}, except that $U$ follows a Laplace distribution and we focus on a reliability ratio of $0.9$. Most observed patterns in Figure~\ref{fig:box} in how different estimators compare are similarly observed here. Interestingly, even though the construction of $\hat{m}_{\text{CE}}(x)$ imposes the normality assumption on $U$, violating this assumption does not noticeably compromise its performance, and it still tends to perform better than the other non-naive estimators.}

{\revise Figure~\ref{fig:box4} repeats the simulation experiment summarized in Figure~\ref{fig:box} using a simpler model with $m(x)=2\text{arctan}(x)$ and $X\sim \text{uniform}(-5, 5)$. Now the non-naive estimations are more similar to each other than when fitting a more challenging model as in Figures~\ref{fig:box} and \ref{fig:box3}.}
Figure~\ref{fig:box2} presents {\reviseagain (in the top row)} boxplots  of fitted values for {\revise $m(x)=2\text{atan}(x)$} based on the estimated circular regression functions from various estimation procedures for a grid of $x$'s, along with circular boxplots {\reviseagain (in the bottom row)} of the corresponding errors defined as the fitted value minus the truth of $m(x)$, when $n=100$ and $h_{\text{SIMEX}}$ is used for non-naive estimation. The former collection of boxplots shows that our proposed estimators substantially improve over $\hat m^*(x)$ in capturing the curvature of $m(x)$. The circular boxplots, on the other hand, reveal that many more errors from naive estimation {\reviseagain (see the fourth panel in the bottom row of Figure~\ref{fig:box2})}, $\hat m^*(x)-m(x)$, deviate from zero and approach $\pm \pi$, whereas such drastic deviation from zero is much less observed for our proposed estimates {\reviseagain (see the first three panels in the bottom row of Figure~\ref{fig:box2})}.

\begin{figure}
\centering
    \includegraphics[scale=.6]{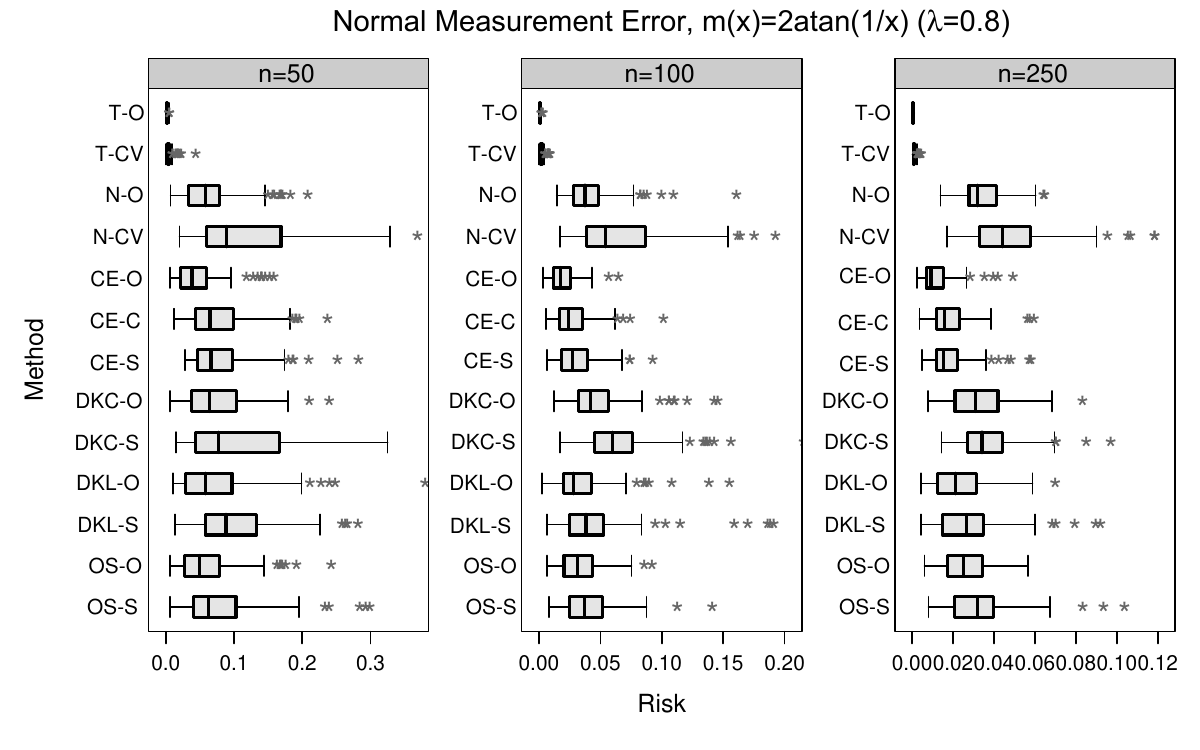}
    \includegraphics[scale=.6]{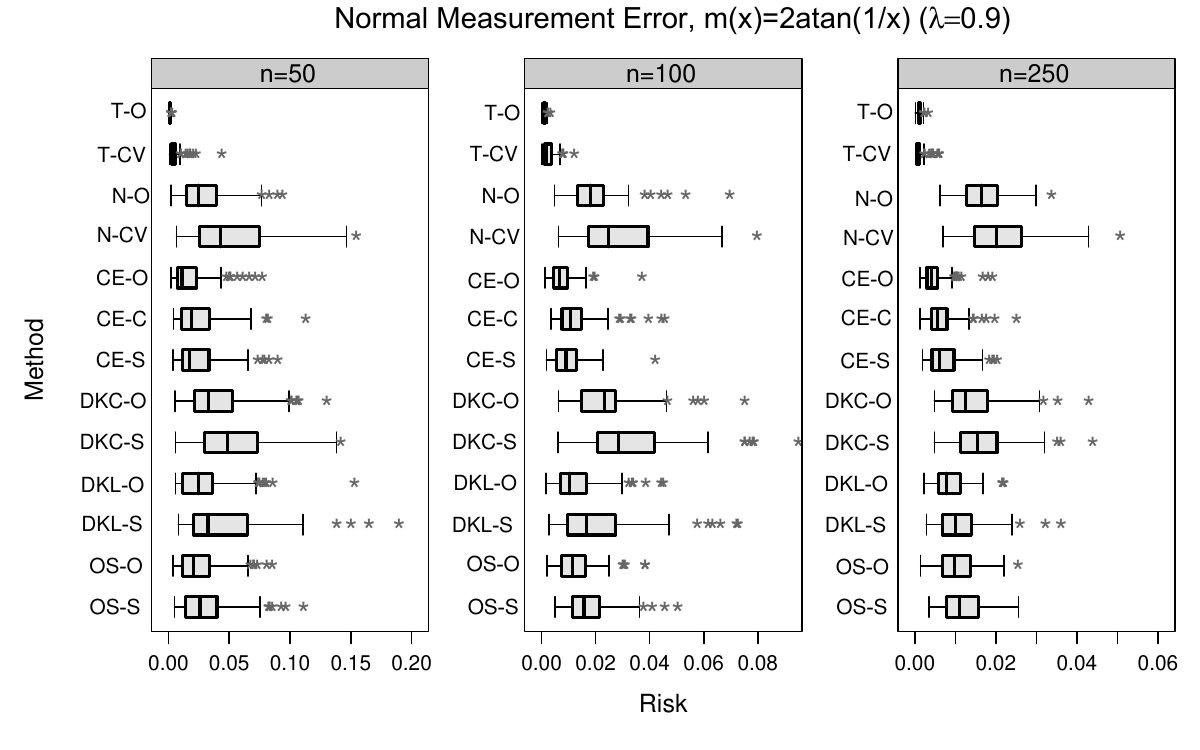}
    \caption{{\revise Boxplots of empirical risk across 100 Monte Carlo replicates at each level of $n\in\{50,100,250\}$ when $X\sim N(0, 4)$,  $m(x)=2\text{atan}(1/x)$, and $U\sim N(0, \sigma^2_u)$ with $\lambda=0.8$ (top panel) and 0.9 (bottom panel) for thirteen estimation procedures, each specified by the estimator and the bandwidth selection method that are linked by a hyphen. Labels for estimators are: $T$ for the ideal estimator $\tilde m(x)$, $N$ for the naive estimator $\hat m^*(x)$, $CE$ for the complex error estimator $\hat m_{\text{CE}}(x)$, $DKC$ for the deconvoluting kernel estimator with the local constant weight, $DKL$ for our deconvoluting kernel estimator $\hat m_{\text{DK}}(x)$ (with the local linear weight), and $OS$ for the one-step correction estimator $\hat m_{\text{OS}}(x)$. Labels for bandwidth selection methods are: $O$ for the optimal bandwidth, 
    $CV$ for cross validation, $C$ for CV-CE, and $S$ for SIMEX.\label{fig:box}}}
\end{figure}

\begin{table}[]
\caption{\label{timing}Computation times (in seconds) different procedures take to estimate $m(x)$ based on a sample of size $n\in \{50, 100, 250\}$, including the time to choose a bandwidth $h$ out of 50 candidates in $\mathcal{H}$. The considered procedures are naive estimation via $\hat{m}^{*}(x)$ paired with cross validation (N-CV), estimation via $\hat{m}_{\text{CE}}(x)$ paired with CV-CE (CE-C), and with CV-SIMEX (CE-S), estimation via $\hat{m}_{\text{DK}}(x)$ using the local constant weight paired with CV-SIMEX (DKC-S), via $\hat{m}_{\text{DK}}(x)$ using the local linear weight paired with CV-SIMEX (DKL-S), and estimation via $\hat{m}_{\text{OS}}(x)$ paired CV-SIMEX (OS-S).}
\begin{tabular}{lllllll}
\hline 
$n$ & N-CV & CE-C  & CE-S & DKC-S & DKL-S & OS-S   \\
\hline 50 & 0.06 & 66.04 & 219 & 34.7 & 113 & 2470 \\
100 & 0.16 & 130 & 625 & 127 & 371 & 4975 \\
250 & 0.21 & 188 & 1250 & 192 & 433 & 12754 \\
  \hline
\end{tabular}
\end{table}
\begin{figure}
\centering
    \includegraphics[scale=.6]{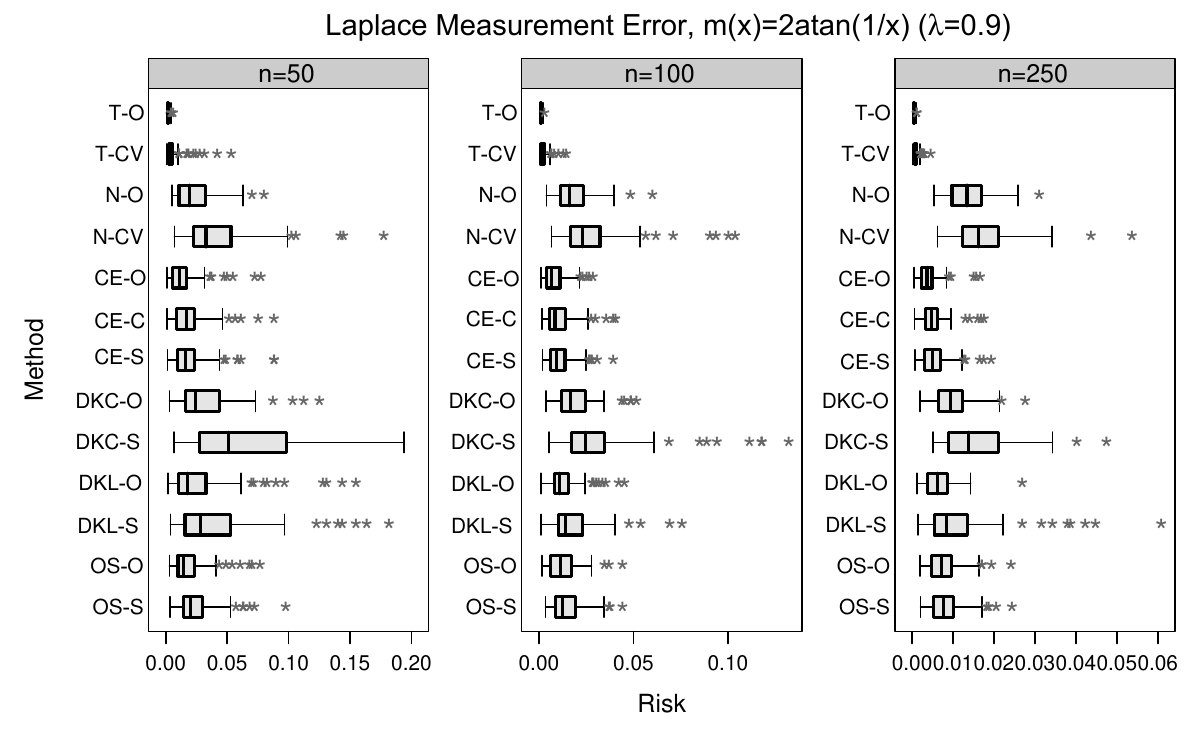}
    \caption{{\revise Boxplots of empirical risk across 100 Monte Carlo replicates at each level of $n\in\{50,100,250\}$ when $X\sim N(0, 4)$,  $m(x)=2\text{atan}(1/x)$, and $U$ follows a mean-zero Laplace distribution with $\lambda=0.9$ for thirteen estimation procedures, each specified by the estimator and the bandwidth selection method that are linked by a hyphen. Labels for estimators are: $T$ for the ideal estimator $\tilde m(x)$, $N$ for the naive estimator $\hat m^*(x)$, $CE$ for the complex error estimator $\hat m_{\text{CE}}(x)$, $DKC$ for the deconvoluting kernel estimator with the local constant weight, $DKL$ for our deconvoluting kernel estimator $\hat m_{\text{DK}}(x)$ (with the local linear weight), and $OS$ for the one-step correction estimator $\hat m_{\text{OS}}(x)$. Labels for bandwidth selection methods are: $O$ for the optimal bandwidth, 
    $CV$ for cross validation, $C$ for CV-CE, and $S$ for SIMEX.}\label{fig:box3}}
\end{figure}
\begin{figure}
\includegraphics[scale=.6]{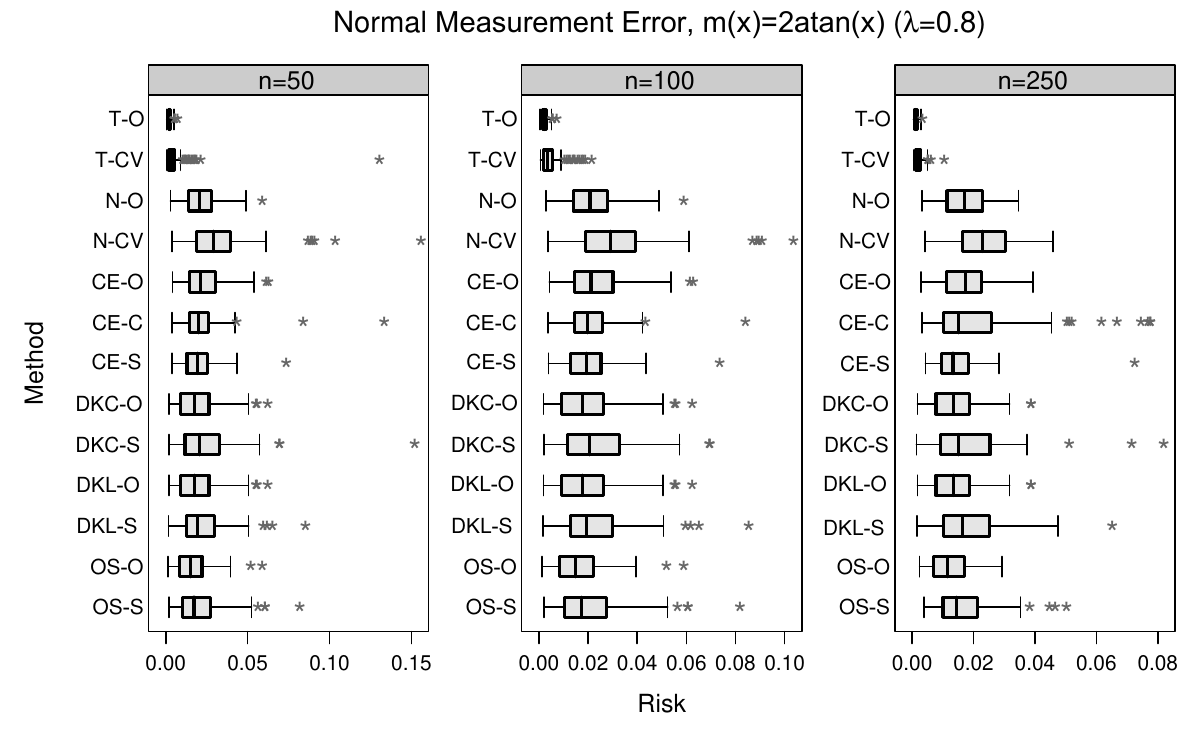}
\includegraphics[scale=.6]{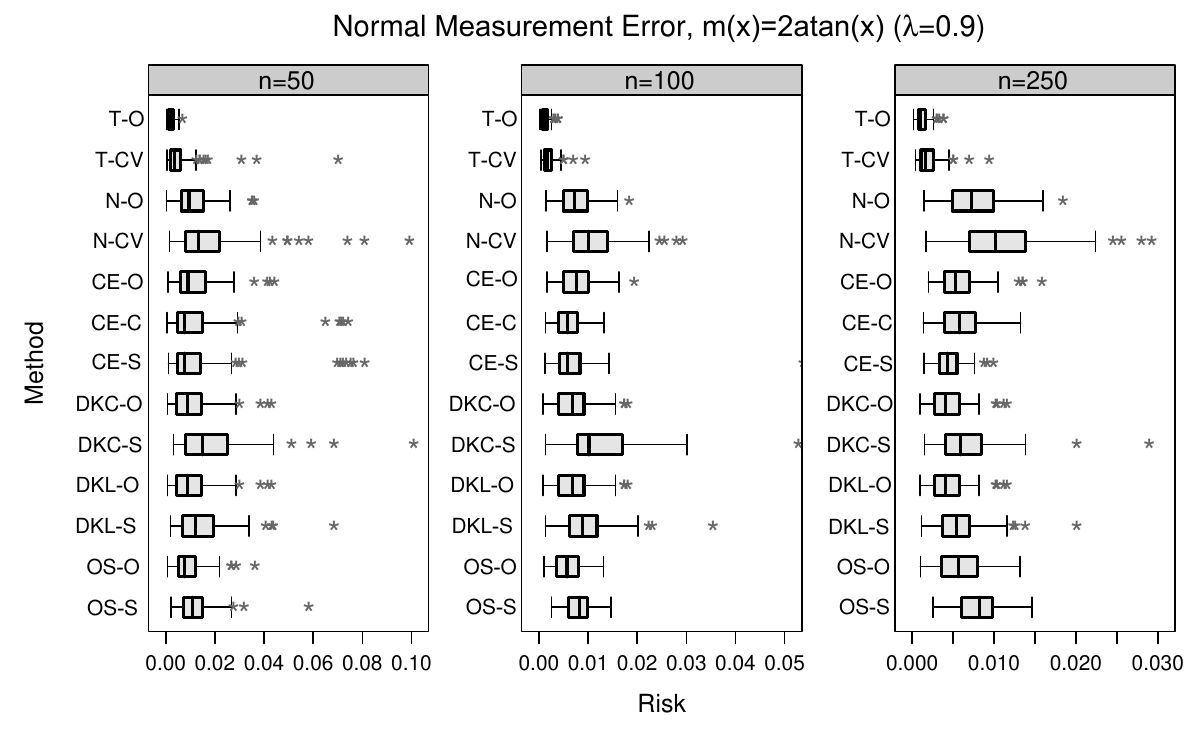}
    \caption{{\revise Boxplots of empirical risk across 100 Monte Carlo replicates at each level of $n\in\{50,100,250\}$ when $X\sim \text{Uniform}(-5, 5)$,  $m(x)=2\text{atan}(x)$, and $U\sim \text{N}(0,100(1-\lambda)/12\lambda)$ with $\lambda=0.8$ (top panel) and 0.9 (bottom panel) for thirteen estimation procedures, each specified by the estimator and the bandwidth selection method that are linked by a hyphen. Labels for estimators are: $T$ for the ideal estimator $\tilde m(x)$, $N$ for the naive estimator $\hat m^*(x)$, $CE$ for the complex error estimator $\hat m_{\text{CE}}(x)$, $DKC$ for the deconvoluting kernel estimator with the local constant weight, $DKL$ for our deconvoluting kernel estimator $\hat m_{\text{DK}}(x)$ (with the local linear weight), and $OS$ for the one-step correction estimator $\hat m_{\text{OS}}(x)$. Labels for bandwidth selection methods are: $O$ for the optimal bandwidth, 
    $CV$ for cross validation, $C$ for CV-CE, and $S$ for SIMEX.\label{fig:box4}}}
\end{figure}
\begin{landscape}
\begin{figure}
    \includegraphics[scale=.6]{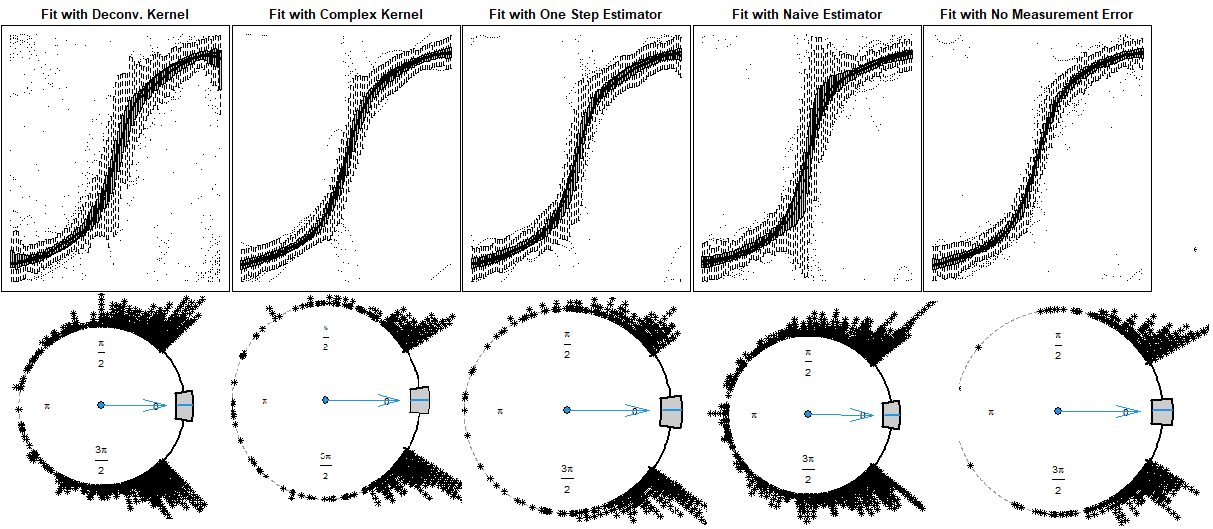}\caption{Boxplots of each fitted point (top panels) and circular boxplots of errors (bottom panels) for each of four methods: $\hat m_{\text{DK}}(x)$, $\hat m_{\text{CE}}(x)$, $\hat m_{\text{OS}}(x)$, $\hat m^*(x)$, and $\tilde m(x)$, based on 100 Monte Carlo replicates when $X\sim \text{uniform}(-5, 5)$, $m(x)=2\text{atan}(x)$. An error here is the fitted value minus the true value of $m(x)$.\label{fig:box2}}
\end{figure}
\end{landscape}
\section{Application to wind data}
\label{sec:real}
\subsection{Texas wind data}
\label{sec:texaswind}
Now we turn to a real-life application and analyze wind direction data from the Texas Natural Resources Conservation Commission \citep{TNRCC}. The data, available at  \url{https://doi.org/10.26023/ZQEZ-ENSF-T09}, contains $n=1756$ records of wind directions $\bTheta$ and times of day $\bX$ collected by a single weather vein from May 20 to July 31 in 2003. The understanding of tendencies for wind direction relative to the time of day is crucial in meteorological studies in order to detect aberrations and predict large-scale weather events. For illustration purposes, we contaminate $\bX$ with additive normal measurement errors to produce surrogate measures, $\bW$, such that an estimated reliability ratio, given by $S_{x}^{2}/S_{w}^{2}$, is equal to 0.9, where $S_x^2$ and $S_w^2$ are the sample variance of $\bX$ and the sample variance of $\bW$, respectively.

We first fit the regression model in \eqref{eq:thetagivenX} using data $(\bTheta, \bX)$ to obtain the ideal estimate $\tilde m(x)$ with the bandwidth selected via the traditional 5-fold CV. Then we use data $(\bTheta, \bW)$ to carry out naive estimation to obtain $\hat m^*(x)$ using $h_{\text{naive}}$. Lastly, we compute three non-naive estimates, $\hat m_{\text{DK}}(x)$, $\hat m_{\text{CE}}(x)$, and $\hat m_{\text{OS}}(x)$, based on data $(\bTheta, \bW)$ using $h_{\text{SIMEX}}$. The resultant five estimated circular regression functions are plotted in Figure~\ref{fig:wind}.

Contrasting the five estimated circular regression functions depicted in Figure \ref{fig:wind}, we see that the naive estimate $\hat m^*(x)$ has a diminished signal strength, flattening the relationship captured by the ideal estimate $\tilde m(x)$. Each of the proposed estimators improves on this by mimicking the pattern of $\tilde m(x)$ more closely, with $\hat m_{\text{CE}}(x)$ being the closest to $\tilde m(x)$ among the four estimates based on error-contaminated data. Like in the simulation study, we set $B=100$ in the complex-valued weight in \eqref{eq:CEweight} when obtaining $\hat m_{\text{CE}}(x)$, which does display some instability, resulting in the estimated regression function less smooth than those from the other two proposed estimators. This may suggest that a larger number of Monte Carlo replicates $B$ in the weight function for $\hat m_{\text{CE}}(x)$ is needed, which is consistent with our asymptotic analysis for this estimator (see Theorems~\ref{thm:asympbias} and \ref{thm:asympvar}), where it is assumed that $B$ grows as $n$ increases.
\begin{figure}[!h]
    \centering
    \includegraphics[scale=.59]{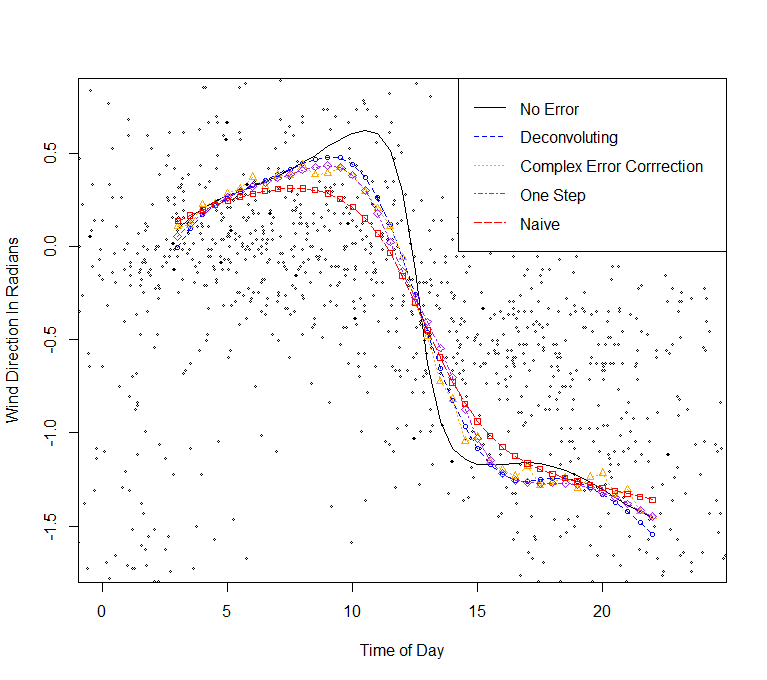}
    \caption{Five estimated circular mean of the wind direction given time of day including the three proposed estimators, $\hat m_{\text{DK}}(x)$ (blue line running through circles $\color{blue}\circ$), $\hat m_{\text{CE}}(x)$ (orange line running through triangles $\color{orange}\triangle$),
    and $\hat m_{\text{OS}}(x)$  (purple line running through  diamonds $\color{purple} \diamond$), the naive estimate $\hat m^*(x)$ (red line passing through squares $\color{red} \square$), and the ideal estimate $\tilde m(x)$ 
    (black line).
    \label{fig:wind}}
\end{figure}

\subsection{Santiago de  Compostela wind data}
The Galician Meteorological Agency of Spain collects climate data, available at the website of the 
Ministry for the Ecological Transition and the Demographic Challenge (\url{https://www.miteco.gob.es/es.html}). To more closely mirror our simulation settings, we downloaded $n=100$ records of wind direction in radians from this website along with measurements of gust wind speed above average in miles per hour from a weather vein in Santiago de Compostela, Spain. The quality of gust wind speed measurements depends on the response characteristics of anemometers and the effects of signal processing, besides other external factors that can introduce errors \citep{pirooz2020effects}. 
The goal of our analysis is to regress the wind direction ($\Theta$) on the gust wind speed above average ($X$), while acknowledging that the measured gust wind speed ($W$) is contaminated by unknown measurement errors. 

The proposed methods for estimating the regression function $m(x)$ remain applicable when the measurement error distribution is unknown but replicate measures of the true covariate are available. In particular, because the complex error estimator, $\hat m_{\text{CE}}(x)$, is formulated under the assumption of normal measurement error, it depends on the error distribution only via $\sigma_u^2$, which can be easily estimated based on replicates measurements \citep[as in equation (4.3),][]{carroll2006measurement}. The deconvoluting kernel estimator, $\hat m_{\text{DK}}(x)$, and the one-step correction estimator, $\hat m_{\text{OS}}(x)$,  involve the measurement error characteristic function $\phi_{U}(\cdot)$. Under the assumption of a symmetric distribution, \cite{delaigle2008deconvolution} provided an estimator for $\phi_U(\cdot)$ based on repeated measurements, whereas \citet{comte2015density} proposed an estimator without assuming symmetry of $U$. These estimated characteristic functions can be used in place of $\phi_U(\cdot)$ in $\hat m_{\text{DK}}(x)$ and $\hat m_{\text{OS}}(x)$. The study of asymptotic properties of the resultant estimators for $m(x)$ is beyond the scope of this manuscript. We conjecture that, with an estimated $\phi_U(\cdot)$ plugged in, finite sample performance of $\hat m_{\text{DK}}(x)$ and $\hat m_{\text{OS}}(x)$ will deteriorate as reported in previous empirical study in a similar context with linear response data \citep{huangzhou}. 

Without repeated measures of the gust wind speed at any given occasion when the wind direction was recorded in this application, we resort to a sensitivity analysis, by fitting the nonparametric regression model under different assumed levels of error contamination corresponding to four levels of reliability ratio, $\lambda\in \{0.75,0.8,0.85,0.9\}$, and two assumed error distributions: normal and Laplace errors. Figure \ref{fig:wind2} shows the estimated regression functions using our proposed methods, contrasting with the naive estimate when normal measurement errors are assumed. Besides using CV-SIMEX to select bandwidths in our estimates, we also repeated the complex error estimation using CV-CE to select the bandwidth in $\hat m_{\text{CE}}(x)$. The two variants of $\hat m_{\text{CE}}(x)$ are similar to each other but are more distinct from the other estimates, especially when a more severe error contamination is assumed. It is at the lower levels of (assumed) $\lambda$ where one witnesses the non-naive estimates differ from the naive one more noticeably. The abrupt jump of $\hat m_{\text{CE}}(x)$  near the lower bound of the covariate when $\lambda=0.75, 0.8$ should not be interpreted as a phenomenon of discontinuity because the response is circular. With this in mind, $\hat m_{\text{CE}}(x)$ appears to be more sensitive to what one assumes for the level of error contamination. The other two estimates, $\hat m_{\text{DK}}(x)$ and $\hat m_{\text{OS}}(x)$, are closer to the naive estimate $\hat m^*(x)$ in comparison, especially at a higher assumed value of $\lambda$, although they both exhibit slightly more curvature than $\hat m^*(x)$.  
\begin{figure}[!h]
\centering
\includegraphics[scale=.59]{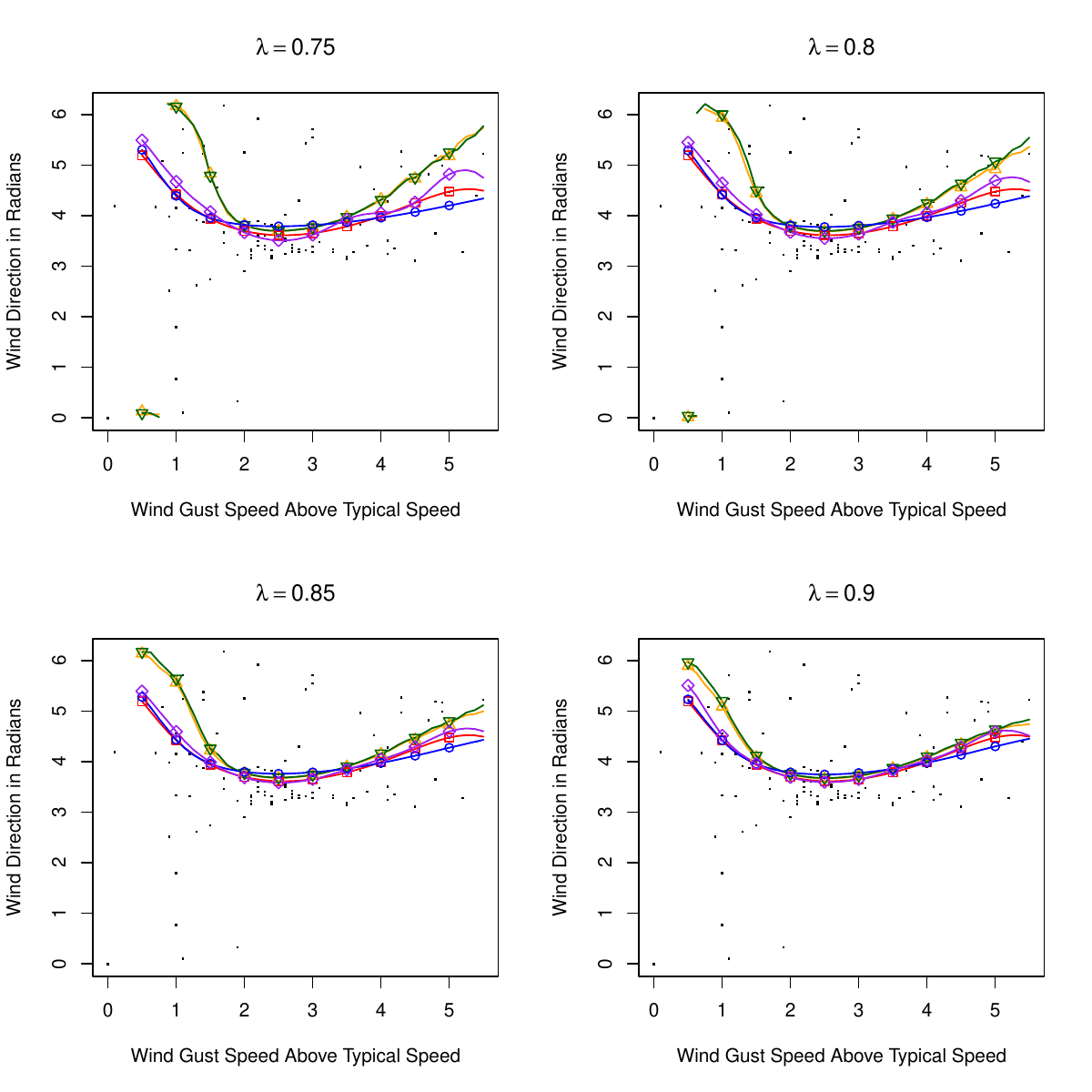}
\caption{Five estimated circular mean of the wind direction given the gust wind speed above average including the three proposed estimators assuming normal measurement error, $\hat m_{\text{DK}}(x)$ (blue line running through circles $\color{blue}\circ$), $\hat m_{\text{CE}}(x)$ using the bandwidth chosen by CV-SIMEX (orange line running through upright triangles $\color{orange}\triangle$), $\hat{m}_{\text{CE}}(x)$ using the bandwidth chosen by CV-CE (green line running through down facing triangles $\color{green}\nabla$) 
    and $\hat m_{\text{OS}}(x)$  (purple line running through  diamonds $\color{purple} \diamond$), and the naive estimate $\hat m^*(x)$ (red line passing through squares $\color{red} \square$).
    \label{fig:wind2}}
\end{figure}

Figure \ref{laplacewind} provides the counterpart estimates when assuming Laplace measurement errors, which are mostly similar to those in Figure~\ref{fig:wind2}. The sensitivity analysis thus suggests that the proposed estimators are more reliant on the assumed level of error contamination and relatively robust to the assumed error distribution. {\reviseagain Although one cannot conclude which estimated regression function is closest to the (nonexistent) ideal estimate $\tilde m(x)$ or the (unknown) truth $m(x)$ in this application, some lessons can be learnt from this exercise. First, unless one has strong data or scientific evidence indicating non-Gaussian measurement error, we recommend using the complex error estimator $\hat m_{\text{CE}}(x)$ that is most computationally convenient to obtain, especially when paired with CV-CE bandwidth selection. Second, when confronted with serious concern or doubt regarding the normality assumption for measurement error, the one-step correction estimator $\hat m_{\text{OS}}(x)$ is the next in line as a contender that tends to be numerically more stable than the deconvoluting kernel estimator $\hat m_{\text{DK}}(x)$. Third, even when using $\hat m_{\text{OS}}(x)$ or $\hat m_{\text{DK}}(x)$, one may consider using the Laplace characteristic function for $\phi_U(\cdot)$ when computing the estimator for a smoother implementation of the deconvoluting operator.}
\begin{figure}[!h]
    \centering
    \includegraphics[scale=.59]{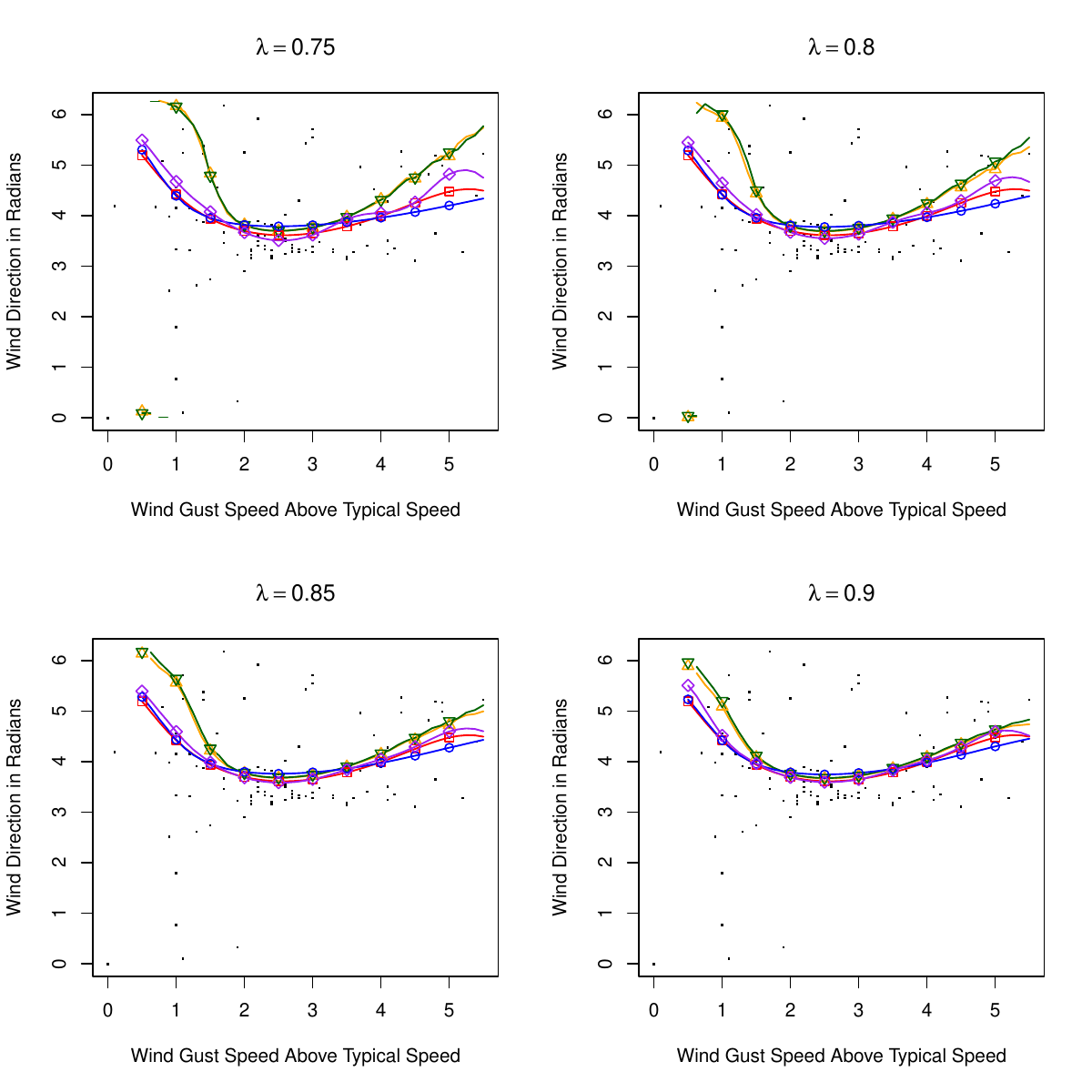}
    \caption{Five estimated circular mean of the wind direction given the gust wind speed above average including the three proposed estimators assuming Laplace measurement error, $\hat m_{\text{DK}}(x)$ (blue line running through circles $\color{blue}\circ$), $\hat m_{\text{CE}}(x)$ using the bandwidth chose by CV-SIMEX (orange line running through upright triangles $\color{orange}\triangle$), $\hat{m}_{\text{CE}}(x)$ using the bandwidth chosen by CV-CE (green line running through down facing triangles $\color{green}\nabla$) 
    and $\hat m_{\text{OS}}(x)$  (purple line running through  diamonds $\color{purple} \diamond$), and the naive estimate $\hat m^*(x)$ (red line passing through squares $\color{red} \square$).
    \label{laplacewind}}
\end{figure}

\section{Discussion}
\label{sec:disc1}
In this study, we develop three methods to extend local polynomial regression of a linear response to local polynomial regression of a circular response while accounting for measurement error in a linear covariate. An \texttt{R} package  \texttt{NPMEDD} for implementing these methods is available to download at \url{https://github.com/nwwoolsey/NPMEDD}, where the data and the \texttt{R} code for reproducing results in Section \ref{sec:real} are also available.
 The proposed methods lead to novel estimators for the circular regression function that use three different strategies to achieve the common goal of correcting naive estimators for covariate measurement error. A common thread running through the novel estimators is a certain form of integral transforms relating to a deconvoluting operation. Besides adapting CV-SIMEX to the circular regression model to select bandwidths in the novel estimators, we propose the CV-CE procedure to enhance the computational efficiency and feasibility of bandwidth selection in the presence of measurement error. Moreover, we thoroughly study the asymptotic properties of the proposed estimators, and establish consistency and asymptotic normality under suitable conditions. Our theoretical investigation provides a generic recipe for studying nonparametric estimators of a circular regression function with error-in-covariate. 

Generalization of the proposed methods to incorporate multiple covariates with some prone to measurement error can be realized by using the multivariate generalization for deconvoluting kernels as outlined in \citet[][Section 10.7.4]{Yi}. In the real-life application in Section~\ref{sec:texaswind}, the time of day can be transformed into a directional covariate as the time on the 24-hour clock, and the proposed methods can be revised by the simple adaptation of a directional kernel as opposed to a linear kernel as in \cite{di2013non} and \cite{di2023kernel}. Bandwidth selection in the presence of covariate measurement error remains a challenging task. Although CV-CE saves a tremendous amount of computation time compared to CV-SIMEX, we acknowledge room for improvement for CV-CE by improving the estimated loss function. Recently, in the context of local polynomial regression of a linear response with error-in-covariate, \cite{dong2022bandwidth} proposed bandwidth selection procedures that also strive for estimating a loss function, which is the mean squared prediction error, by accounting for measurement error in the estimated regression function and also the estimated density of error-prone covariates. Our preliminary experiments on their methods suggest an increased computational intensity and instability, without noticeable gain in the final estimation quality. Lastly, we assume fully known measurement error distribution throughout the study to avoid identifiability issues for measurement error models. In practice, one would have to count on replicate data or external/validation data to estimate the error distribution, or to conduct sensitivity analysis if such data are unavailable to estimate $\sigma_u^2$ or $\phi_U(t)$. 

\begin{appendix}
\renewcommand*{\theequation}{A.\arabic{equation}}
\setcounter{equation}{0}
\section*{Appendix A: Proof of Lemma~\ref{lem:addce2}}
\begin{proof}
    By the multivariate version of Taylor's theorem, one has a Taylor expansion of $g(\bt+\bV)$ around $\bV=\boldsymbol{0}$ given by   
    \begin{equation*}
        g(\bt+\bV)=\sum_{|\bell|=0}^\infty \frac{D^{\bell} g (\bt)}{\bell !} \bV^{\bell}, 
    \end{equation*}
    where $\bell !=\ell_1 !\ldots \ell_n !$ and $D^{\bell} g(\bt)=(\partial ^{\ell_1} \ldots \partial^{\ell_n} / \partial t_1^{\ell_1} \ldots \partial t_n^{\ell_n })g(\bt)$.
    It follows that 
    \begin{equation*}        
    \IE \{g(\bt+\bV)\}=g(\bt)+\sum_{|\bell|>0}^
    \infty \frac{D^{\bell} g (\bt)}{\bell !} \IE(\bV^{\bell}),
    \end{equation*}
    which reduces to $g(\bt)$ if $\IE(\bV^{\bell})=0$, $\forall \bell$ with $|\bell|>0$. This completes the proof.
\end{proof}

\renewcommand*{\theequation}{B.\arabic{equation}}
\renewcommand*{\thelemma}{B.\arabic{lemma}}
\setcounter{equation}{0}
\section*{Appendix B: Proof of Lemma \ref{lem:1steppf}}
The proof is similar to that given in \citet[][Appendix A.2]{delaigle2014nonparametric}.
\begin{proof}
Define $Y=\sin \Theta$. Because $U\perp X$, we have $Y\perp W|X$, thus
\begin{align*}
E(Y|W=w) & = \int_{-1}^1 y f_{Y|W}(y|w)\, dy \\   
& =  \int_{-1}^1 y \cdot \frac{f_{Y,W}(y,w)}{f_W(w)}\, dy\\
& =  \frac{1}{f_W(w)}\int_{-1}^1 y \int_{-\infty}^\infty f_{Y,W|X}(y,w|x)f_X(x)\, dx\, dy \\
& =  \frac{1}{f_W(w)} \int_{-\infty}^\infty \int_{-1}^1 y  f_{Y|X}(y|x)f_{W|X}(w|x)\, dy f_X(x) \, dx \\
& =  \frac{1}{f_W(w)} \int_{-\infty}^\infty \int_{-1}^1 y  f_{Y|X}(y|x)\, dy f_X(x)f_U(w-x) \, dx \\
& = \frac{1}{f_W(w)} \int_{-\infty}^\infty m_1(x) f_X(x)f_U(w-x) \, dx. 
\end{align*}
Hence, 
\begin{align}
m^*_1(w)f_W(w)= \int_{-\infty}^\infty m_1(x) f_X(x)f_U(w-x) \, dx. \label{eq:naiveprod*}
\end{align}
Applying the Fourier transform on both sides of \eqref{eq:naiveprod*} yields
\begin{align*}
  \phi_{m^*_1f_W}(t) & = \int_{-\infty}^\infty e^{itw}  \int_{-\infty}^\infty m_1(x) f_X(x)f_U(w-x)  \,dx \,dw \\
  & = \int_{-\infty}^\infty m_1(x)  f_X(x) e^{itx}  \int_{-\infty}^\infty e^{it(w-x)}f_U(w-x)  \,dw \,dx \\
  & = \int_{-\infty}^\infty m_1(x)  f_X(x) e^{itx}  \int_{-\infty}^\infty e^{itu}f_U(u)  \,du \,dx \\ 
  & = \phi_U(t)\int_{-\infty}^\infty m_k(x)  f_X(x) e^{itx}   \,dx \\
  & = \phi_{m_1f_X}(t)\phi_U(t),
\end{align*}
where $\phi_{m^*_1f_W}(t)$ is the Fourier transform of $m^*_1(w)f_W(w)$, and $\phi_{m_1f_X}(t)$ is the Fourier transform of $m_1(x)f_X(x)$. Finally, applying the inverse Fourier transform on both sides of $\phi_{m_1f_X}(t)=\phi_{m^*_1f_W}(t)/\phi_U(t)$ gives $$m_1(x)f_X(x)=\frac{1}{2\pi} \int e^{-itx} \frac{\phi_{m_1^* f_W}(t)}{\phi_U(t)}\, dt.$$
A similar result can be derived for $m^*_2(w)=\IE(\cos \Theta|W=w)$. This completes the proof of Lemma~\ref{lem:1steppf}.
\end{proof}
Even though we consider $Y=\sin \Theta$ and $Y=\cos \Theta$ in this proof, the arguments leading up to \eqref{eq:naiveprod*} hold for any linear random variable $Y$ as long as all relevant integrals are well-defined. This point is important for later development, e.g, in Appendix D.4, where we consider $Y$ to be other linear random variables. To recap, provided that all relevant integrals (including expectations) are well-defined and assuming a classical additive measurement error model relating $W$ and $X$, then we have, for any linear random variable $Y$, that 
$$\IE(Y|W=w)f_W(w)=\int \IE(Y|X=x)f_X(x)f_U(w-x)\, dx,$$ 
or, equivalently, 
\begin{equation}
    \IE\{\IE(Y|X)f_U(w-X)\}=\IE(Y|W=w)f_W(w).
    \label{eq:anyYprod}
\end{equation}

\renewcommand*{\theequation}{C.\arabic{equation}}
\setcounter{equation}{0}
\renewcommand*{\thesubsection}{C.\arabic{subsection}}
\renewcommand*{\thelemma}{C.\arabic{lemma}}
\setcounter{subsection}{0}
\section*{Appendix C: Notations and important results}
\subsection{Notations}
For easy reference, we provide a list of notations next. 
\begin{itemize}
\item For $\ell=1, 2$,  $\tau_\ell(x)=m_\ell(x)f_X(x)$, $\tau^*_\ell(x)=m^*_\ell(x)f_W(x)$, where $m_1(x)=\IE(\sin \Theta |X=x)$, $m_2(x)=\IE(\cos \Theta |X=x)$, $m^*_1(x)=\IE(\sin \Theta |W=x)$, and $m^*_2(x)=\IE(\cos \Theta |W=x)$.
\item $\xi_1(x)=\IE\{(\sin \Theta)^2|X=x\}$, $\xi_2(x)=\IE\{(\cos \Theta)^2|X=x\}$, \item $\psi(x)=\IE\{(\sin \Theta)( \cos \Theta)|X=x\}$.
\item For $\ell \in \mathbb{N}_0$,  $\mu_\ell = \int t^\ell K(t) dt$, $\nu_\ell = \int t^\ell K^2(t) dt$.
\item  For $\ell \in \mathbb{N}_0$, $K_{U, \ell, h}(x)=K_{U \text{\revise ,}\ell}(x/h)/h$, where 
\begin{equation}
    K_{U,\ell}(x)=x^\ell L_{\ell}(x)=i^{-\ell}\frac{1}{2\pi}\int e^{-itx}\frac{\phi_K^{(\ell)}(t)}{\phi_U(-t/h)}dt,
    \label{eq:KUlh}
\end{equation} 
in which $L_{\ell}(x)$ is given in Equation \eqref{eq:LDK}. 
\end{itemize}

\subsection{Important lemmas}
Several results that we heavily rely on in our asymptotic analyses are listed below. Some of these results are established in existing literature, while others generalize existing results.  
\begin{lemma}
\label{lem:Fan1991}
If a sequence of Borel functions $\{G_n(t)\}$ which satisfy the limit $\lim_{n\to \infty} G_n(t)=G(t)$ and $\sup_n |G_n(t)|\le G^*(t)$, where $G^*(t)$ is an integrable function satisfying \\ $\lim_{t\to \infty} | G^*(t)|=0$, then $$\lim_{n\to \infty} \int G_n(t)f(x-ht)dt=f(x)\int G(t)dt,$$ where $h\to 0$ as $n\to \infty$, and $x$ is a continuity point of function $f(\cdot)$.
\end{lemma}
This is Lemma 2.1 in \cite{fan1991asymptotic}. 

\begin{lemma}
\label{lem:DFCB3}
If $U$ is ordinary smooth of order $\beta$ with a non-vanishing characteristic function $\phi_U(t)$,  $\|\phi^{(\ell)}_K(t)\|_\infty <\infty$, and $\int |t|^\beta \phi_K^{(\ell)}(t)dt<\infty$, then 
$$\lim_{n\to \infty} h^\beta K_{U, \ell}(x) = \frac{i^{-\ell}}{2\pi c} \int e^{-itx} |t|^\beta \phi^{(\ell)}_K(t)dt,$$
where $c=\lim_{t\to \infty} t^\beta \phi_U(t)$.
\end{lemma}
This is Lemma B.3 in \cite{Delaigle}.

\begin{lemma}
\label{lem:DFCB4}
For $\bell=(\ell_1, \ldots, \ell_Q)^\top\in \mathbb{N}^Q_0$ with $Q\ge 2$, suppose the kernel $K(\cdot)$ satisfies $\|\phi_K^{(\ell_q)}(t)\|_\infty <\infty$, $\|\phi_K^{(\ell_q+1)}(t)\|_\infty <\infty$, $\int |t|^\beta \phi_K^{(\ell_q)}(t) dt<\infty$, and $\int (|t|^\beta+|t|^{\beta-1}) \{|\phi_K^{(\ell_q)}(t)|+|\phi_K^{(\ell_q+1)}(t)|\} dt<\infty$, for $q=1, \ldots, Q$, also, $U$ is ordinary smooth of order $\beta$ with $\phi_U(t)\ne 0$, for all $t$, and $\|\phi'_U(t)\|_\infty< \infty$, then 
$$\lim_{n\to \infty} h^{Q\beta} \int \left\{ \prod_{q=1}^Q K_{U, \ell_q}(v) \right\}  f(x-hv)  dv=f(x)\eta(\ell_1, \ldots, \ell_Q),$$
where $h\to 0$ as $n\to \infty$, $x$ is a continuity point of a bounded function $f(\cdot)$, and 
\begin{align*}
 &\ \eta(\ell_1, \ldots, \ell_Q) \\
 = &\ \frac{i^{-|\ell|}}{c^Q(2\pi)^{Q-1}}\int\ldots \int \left\vert \left(\prod_{q=1}^{Q-1} t_q\right)\left(\sum_{q=1}^{Q-1} t_q\right) \right\vert^\beta \left\{ \prod_{q=1}^{Q-1} \phi_K^{(\ell_q)}(t_q)\right\}\times \\
 &\ \phi_K^{(\ell_Q)}\left(-\sum_{q=1}^{Q-1} t_q\right) dt_1\ldots dt_{Q-1},
\end{align*}
where $c=\lim_{t\to \infty} t^\beta \phi_U(t)$.
\end{lemma}
This is a generalization of Lemma B.4 in \cite{Delaigle} that we prove next. 

\begin{proof}
Firstly, by Lemma~\ref{lem:DFCB3},
$$\lim_{n \to \infty} h^{Q\beta} \prod_{q=1}^Q K_{U, \ell_q}(v)=i^{-\sum_{q=1}^Q \ell_q} \frac{1}{(2\pi c)^Q}\prod_{q=1}^Q \int e^{-itv}|t|^\beta \phi_K^{(\ell_q)}(t) dt.$$
Secondly, by Lemma~\ref{lem:Fan1991}, 
\begin{align}
   &\ \lim_{n\to \infty} h^{Q\beta} \int \left\{ \prod_{q=1}^Q K_{U, \ell_q}(v) \right\}  f(x-hv)  \, dv \nonumber \\
   = &\ f(x) \times i^{-\sum_{q=1}^Q \ell_q} \frac{1}{(2\pi c)^Q}\int \prod_{q=1}^Q \int e^{-its}|t|^\beta \phi_K^{(\ell_q)}(t) \, dt \, ds \nonumber \\
  = &\ i^{-\sum_{q=1}^Q \ell_q} \frac{f(x)}{(2 \pi c)^Q} \int \int \ldots \int e^{-it_1s} |t_1|^\beta \phi_K^{(\ell_1)}(t_1) \ldots e^{-it_Qs} |t_Q|^\beta \times \nonumber \\
  &\ \phi_K^{(\ell_Q)}(t_Q)\, dt_1\ldots dt_Q \, ds \nonumber \\
   = &\ \frac{f(x)}{c^Q}\frac{i^{-\sum_{q=1}^Q \ell_q}}{(2 \pi)^{Q-1}} 
   \int\ldots \int \prod_{q=1}^{Q-1} \left\{| t_q|^\beta  \phi_K^{(\ell_q)}(t_q) \right\} \times \nonumber \\
   &\ \frac{1}{2\pi} \int \int e^{-is(t_Q+\sum_{q=1}^{Q-1}t_q)}|t_Q|^\beta \phi^{(\ell_Q)}_K(t_Q) \, dt_Q  \, ds\, dt_1 \ldots d_{t_{Q-1}}.\label{eq:lasttwo}
\end{align}
Thirdly, the inner double integral in (\ref{eq:lasttwo}) can be simplified using the Fourier integral theorem. Recall that the Fourier integral theorem \citep[][Section 3.2.1]{lewis2021advanced}
states that, if a function $g(v)$ satisfies the Dirichlet conditions in every finite interval centering at zero, and $\int |g(v)|dv<\infty$, then 
$$g(v)=\frac{1}{2\pi} \int \int e^{-is(t-v)} g(t)  dt ds.$$ 
By the Fourier integral theorem,  
\begin{equation*}
    \frac{1}{2\pi} \int \int e^{-is(t_Q+\sum_{q=1}^{Q-1}t_q)}|t_Q|^\beta \phi^{(\ell_Q)}_K(t_Q) \, dt_Q  \, ds = \left|\sum_{q=1}^{Q-1} t_q\right|^\beta \phi_K^{(\ell_Q)}\left(-\sum_{q=1}^{Q-1} t_q\right).
\end{equation*}
Using this result in (\ref{eq:lasttwo}) gives the result stated in Lemma~\ref{lem:DFCB4}.
\end{proof}

\begin{lemma}
\label{lem:DFCB9}
For $(\ell_1, \ldots, \ell_Q)\in \mathbb{N}^Q_0$ with $Q\ge 2$, suppose that $\phi_K(t)$ is supported on $[-1, 1]$,  $\|\phi_K^{(\ell_q)}(t)\|_\infty <\infty$, for $q=1, \ldots, Q$, and $U$ is super smooth of order $\beta$ with $\phi_U(t)\ne 0$, for all $t$, then 
$$\left\vert \int \prod_{q=1}^Q K_{U, \ell_q}(v) dv \right\vert \le C h^{Q\beta_2}\exp\left(Qh^{-\beta}/\gamma\right),$$
where $C$ is some positive constant and $\beta_2=\beta_0I(\beta_0<0.5)$.
\end{lemma}
This is a generalization of Lemma B.9 in \cite{Delaigle}, and can be proved using Lemma B.8 in \cite{Delaigle}. 

\renewcommand*{\theequation}{D.\arabic{equation}}
\setcounter{equation}{0}
\renewcommand*{\thesubsection}{D.\arabic{subsection}}
\renewcommand*{\thelemma}{D.\arabic{lemma}}
\setcounter{subsection}{0}
\section*{Appendix D: Asymptotic properties of the deconvoluting kernel estimator $\hat m_{\text{DK}}(x)$}
    \subsection{The mean of the local linear weight $\mathcal{L}(W_j-x)$}\label{sec:loclin}
    Using \eqref{eq:KUlh}, we re-express the new weight $\mathcal{L}(W_j-x)$ in Equation \eqref{eq:DKweight} as 
    \begin{align*}
    &\ \mathcal{L}(W_j-x)\\
    = &\ \frac{1}{n}K_{U,0,h}(W_j-x)\sum_{k=1}^{n}
        K_{U,2,h}(W_k-x)-\frac{1}{n}K_{U,1,h}(W_j-x)\sum_{k=1}^{n}K_{U,1,h}(W_k-x).
    \end{align*}
    The result in Equation \eqref{eq:DKmean} is equivalent to 
    \begin{equation}
        E\{K_{U,\ell, h}(W-x)|X\}=\left(\frac{X-x}{h}\right)^\ell K_h(X_j-x), \text{ for $\ell=0, 1, 2$}.
        \label{eq:EKUlh}
    \end{equation} Using this result and the new expression of the weight given above, we have
    \begin{align}
        &\ \IE\{\mathcal{L}(W_{j}-x)|\bX\} \nonumber \\
        = &\ \frac{1}{n}\IE\left\{K_{U,0,h}(W_j-x)K_{U,2,h}(W_j-x)+K_{U,0,h}(W_j-x)\sum_{k \ne j}K_{U,2,h}(W_k-x)\right.\nonumber\\
        &\ \left.\left.-K^2_{U,1,h}(W_j-x)-K_{U,1,h}(W_j-x)\sum_{k\ne j}K_{U,1,h}(W_k-x)\right\vert\bX\right\} \nonumber\\
       = &\ \frac{1}{n}\IE\left\{\left. K_{U,0,h}(W_j-x)K_{U,2,h}(W_j-x)-K^2_{U,1,h}(W_j-x)\right \vert X_j \right\} + \nonumber\\
        &\ \frac{1}{n} \left\{K_{h}(X_{j}-x)\sum_{k\ne j} \bigg(\frac{X_k-x}{h}\bigg)^{2} K_{h}(X_k-x)- \right. \nonumber\\
        &\ \left.\bigg(\frac{X_{j}-x}{h}\bigg)K_{h}(X_{j}-x)\sum_{k \ne j} \bigg(\frac{X_k-x}{h}\bigg) K_{h}(X_k-x)\right\} \nonumber\\
       = &\ \frac{1}{n}\IE\left\{\left.K_{U,0,h}(W_j-x)K_{U,2,h}(W_j-x)\right\vert X_{j}\right\}-\frac{1}{n}\IE\left\{\left.K^2_{U,1,h}(W_j-x)\right \vert X_j\right\} +\label{eq:DKweight2exp}\\
        &\ \frac{1}{n} \left\{K_{h}(X_{j}-x)\sum_{k=1}^n \bigg(\frac{X_k-x}{h}\bigg)^{2} K_{h}(X_k-x)- \bigg(\frac{X_{j}-x}{h}\bigg)K_{h}(X_{j}-x)\times \right.\label{eq:back2original}\\
        &\ \left.\sum_{k=1}^n \bigg(\frac{X_k-x}{h}\bigg) K_{h}(X_k-x)\right\},\nonumber
    \end{align}
where \eqref{eq:back2original} is the original local linear weight $\mathcal{W}(X_j-x)$. To derive the two expectations in \eqref{eq:DKweight2exp} and other similar expectations arising in later derivations in this appendix, we derive a more general result next. 

For $\bell=(\ell_1, \ldots, \ell_Q)^\top\in \mathbb{N}_0^Q$ with $Q\ge 2$, by Lemmas~\ref{lem:DFCB4} and~\ref{lem:DFCB9},
\begin{align}
&\ \IE\left\{ \left. \prod_{q=1}^Q K_{U, \ell_q, h}(W-x) \right \vert X\right\} \nonumber \\
= &\ \int \left\{\prod_{q=1}^Q K_{U, \ell_q, h}(w-x)\right\} f_U(w-X) \, dw , \text{ next let $v=(w-x)/h$,} \nonumber\\
= &\ \frac{1}{h^{Q-1}}\int \left\{\prod_{q=1}^Q K_{U, \ell_q}(v)\right\} f_U(x+hv-X)\, dv\nonumber \\
= &\ \frac{1}{h^{Q-1}}\int \left\{\prod_{q=1}^Q K_{U, \ell_q}(v)\right\} f_U(X-x-hv)\, dv, \text{ assuming $f_U(-u)=f_U(u)$,} \nonumber \\
= &\ 
\begin{cases}
   \displaystyle{\frac{\eta(\ell_1, \ldots, \ell_Q)}{h^{Q-1+Q\beta}}f_U(X-x)+o_p\left( \frac{1}{h^{Q-1+Q\beta}} \right)}, & \text{if $U$ is ordinary smooth,}\\ 
    \displaystyle{O_p\left( \frac{\exp(Qh^{-\beta}/\gamma)}{h^{Q-1-Q\beta_2}}\right)},  & \text{if $U$ is super smooth.}
\end{cases}
\label{eq:EprodKs}
\end{align}
Using \eqref{eq:EprodKs} with $\bell=(0, 2)^\top$, we have, for the first expectation in \eqref{eq:DKweight2exp},
\begin{align*}
&\  \IE\left\{ \left. K_{U,0,h}(W_j-x)K_{U,2,h}(W_j-x)\right \vert X_j\right\} \\
= &\ 
\begin{cases}
    \displaystyle{\frac{\eta(0,2)}{h^{1+2\beta}}f_{U}(X_{j}-x)+o_p\left(\frac{1}{h^{1+2\beta}}\right)}, & \text{if $U$ is ordinary smooth of order $\beta$,}\\
   \displaystyle{O_p\left(\frac{\exp(2h^{-\beta}/\gamma)}{h^{1-2\beta_2}} \right),} & \text{if $U$ is super smooth of order $\beta$;}
\end{cases}
\end{align*}
similarly, the second expectation in \eqref{eq:DKweight2exp} is, with $\bell=(1, 1)^\top$ in \eqref{eq:EprodKs},\begin{align*}
       &\IE\left\{ \left. K^2_{U,1,h}(W_{j}-x)\right \vert X_j\right\} 
      \\&= 
    \begin{cases}
    \displaystyle{\frac{\eta(1,1)}{h^{1+2\beta}}f_{U}(X_{j}-x)+o_p\left(\frac{1}{h^{1+2\beta}}\right)}, & \text{if $U$ is ordinary smooth of order $\beta$,}\\
   \displaystyle{O_p\left(\frac{\exp(2h^{-\beta}/\gamma)}{h^{1-2\beta_2}} \right),} & \text{if $U$ is super smooth of order $\beta$.}
    \end{cases}
    \end{align*}

It follows that 
\begin{equation}
\label{eq:locallinsingle}
\begin{aligned}
 &\  \IE\{\mathcal{L}(W_j-x) \vert \bX\} \\
      = &\ 
    \begin{cases}
    \mathcal{W}(X_j-x)+ \displaystyle{\frac{\eta(0,2)-\eta(1,1)}{nh^{1+2\beta}}f_{U}(X_{j}-x)}+ & \\
    \displaystyle{o_p\left(\frac{1}{nh^{1+2\beta}}\right)}, & \text{ if $U$ is ordinary smooth,} \\
   \mathcal{W}(X_j-x)+ \displaystyle{O_p\left(\frac{\exp(2h^{-\beta}/\gamma)}{nh^{1-2\beta_2}} \right),} & \text{ if $U$ is super smooth.}
    \end{cases}
\end{aligned}
\end{equation}
Hence, the new weight based on $\bW$, $\mathcal{L}(W_j-x)$, is a consistent estimator of the local linear weight based on $\bX$, $\mathcal{W}(X_j-x)$, under suitable conditions. In particular, if $U$ is ordinary smooth of order $\beta$, then the condition is $nh^{1+2\beta}\to \infty$, as $n\to \infty$ and $h\to 0$; if $U$ is super smooth of order $\beta$, then the condition is $nh^{1-2\beta_2}\exp(-2h^{-\beta}/\gamma)\to \infty$, as $n\to \infty$ and $h\to 0$. We  impose these conditions throughout this appendix. 

\subsection{The mean of $\mathcal{L}^2(W_j-x)$}\label{sec:loclin2}
To derive the asymptotic variance of $\hat g_{\ell, \text{DK}}(x)$, for $\ell=1, 2$, we derive the mean of $\mathcal{L}^2(W_j-x)$ given $\bX$, 
    $$
    \IE\left\{\left.\mathcal{L}^2(W_j-x)\right \vert \bX\right\} 
    = \frac{1}{n^2}(I-2II+III),
    $$
    where 
\begin{align*}
 I = &\  \left\{K_{U, 0, h}(W_j-x)\sum_{k=1}^n K_{U, 2, h}(W_k-x)\right\}^2, \\
 II = &\ \displaystyle{\left\{K_{U,0,h}(W_j-x)\sum_{k=1}^n K_{U,2,h}(W_k-x)\right\}}\times\\
 &\ \displaystyle{\left\{K_{U,1,h}(W_j-x)\sum_{\ell=1}^n K_{U,1,h}(W_\ell-x)\right\}},\\
 III = &\ \left\{K_{U,1,h}(W_j-x) \sum_{k=1}^n K_{U,1,h}(W_k-x)\right\}^2.
\end{align*}

First,
    \begin{align*}
          &\ \IE(I|\bX)\\
        = &\ \IE\left\{\left. K^2_{U,0,h}(W_j-x) K^2_{U,2,h}(W_j-x)\right \vert X_j\right\}\\&+\IE\left\{\left. K^2_{U,0,h}(W_j-x)\right\vert X_j\right\}\sum_{k\ne j}\IE\left\{\left. K^2_{U,2,h}(W_k-x)\right\vert X_k\right\}+\\
        &\ 2\IE\left\{\left. K^2_{U,0,h}(W_j-x)K_{U,2,h}(W_j-x)\right\vert X_j\right\}\sum_{k\ne j}\IE\left\{\left. K_{U,2,h}(W_k-x)\right \vert X_k\right\}+\\
        &\ \IE\left\{\left. K^2_{U,0,h}(W_j-x)\right \vert X_j\right\}\sum_{\ell \ne k \ne j}\IE\left\{\left. K_{U,2,h}(W_k-x)\right \vert X_k\right\}\IE
        \left\{\left. K_{U,2,h}(W_\ell-x)\right \vert X_\ell\right\}.
    \end{align*}
Applying \eqref{eq:EKUlh} and \eqref{eq:EprodKs} to expectations in the above elaboration of $E(I|\bX)$, we have, if $U$ is ordinary smooth of order $\beta$,       \begin{align}   
        &\ \IE(I|\bX) \nonumber\\
       = &\ \frac{\eta(0,0,2,2)}{h^{3+4\beta}}f_U(X_j-x)+o_p\left(\frac{1}{h^{3+4\beta}}\right)+\frac{\eta(0,0)\eta(2,2)}{h^{2+4\beta}}f_U(X_j-x)\times \nonumber \\
       &\ \sum_{k\ne j}f_U(X_k-x)+o_p\left(\frac{n}{h^{2+4\beta}}\right)+ \label{eq:EIrow1}\\
       &\ \left\{\frac{\eta(0,0,2)}{h^{2+3\beta}}f_U(X_j-x)+o_p\left(\frac{1}{h^{2+3\beta}}\right)\right\}\times 2\sum_{k\ne j}\left(\frac{X_k-x}{h}\right)^2 K_h(X_k-x)+\label{eq:EIrow2}\\
       &\ \left\{\frac{\eta(0,0)}{h^{1+2\beta}}f_U(X_j-x)+o_p\left(\frac{1}{h^{1+2\beta}}\right)\right\}\sum_{\ell\ne k \ne j}\left(\frac{X_k-x}{h}\right)^2K_{h}(X_k-x)\times \nonumber \\
        &\ \left(\frac{X_\ell-x}{h}\right)^2 K_{h}(X_\ell-x), \label{eq:EIrow3}
    \end{align}
which can be further simplified by keeping only the dominating term as we explain next. Using the approximation for a random variable $A$ with a finite variance, $A=\IE(A)+O_p(\sqrt{\text{Var}(A)})$, one can show that 
\begin{equation}
\label{eq:apppolyK}
\left\{
\begin{aligned}
\left( \frac{X-x}{h}\right)K_h(X-x) & = h f'_X(x)\mu_2+O(h^3)+O_p(1/\sqrt{h}), \\
\left( \frac{X-x}{h}\right)^2K_h(X-x) & = f_X(x)\mu_2+O(h^2)+O_p(1/\sqrt{h}), 
\end{aligned}
\right.
\end{equation}
as $n\to \infty$ and $h\to 0$, and thus 
\begin{equation}
\label{eq:sumellkj}
\left\{
\begin{aligned}
&\ \sum_{\ell\ne k \ne j}\left(\frac{X_k-x}{h}\right)K_{h}(X_k-x)\left(\frac{X_\ell-x}{h}\right) K_{h}(X_\ell-x)\\
= &\ (n-1)(n-2) \left[ h^2 \left\{f'_X(x)\right\}^2\mu_2^2 +O_p\left(\frac{1}{h}\right)\right],\\
&\ \sum_{\ell\ne k \ne j}\left(\frac{X_k-x}{h}\right)^2K_{h}(X_k-x)\left(\frac{X_\ell-x}{h}\right) K_{h}(X_\ell-x) \\
= &\ (n-1)(n-2) \left\{ hf_X(x)f_X'(x)\mu_2^2 +O_p\left(\frac{1}{h}\right)\right\},\\
&\  \sum_{\ell\ne k \ne j}\left(\frac{X_k-x}{h}\right)^2K_{h}(X_k-x)\left(\frac{X_\ell-x}{h}\right)^2 K_{h}(X_\ell-x)\\
= &\ (n-1)(n-2) \left\{ f^2_X(x)\mu_2^2 +O_p\left(\frac{1}{h}\right)\right\},
\end{aligned}
\right.
\end{equation}
which are of order $n^2/h$, and hence dominate $\sum_{k\ne j}\{(X_k-x)/h\}K_h(X-k-x)$ and $\sum_{k\ne j}\{(X_k-x)/h\}^2K_h(X-k-x)$ that are of order $n/\sqrt{h}$. 
Now one can see that, with $nh^{1+2\beta}\to \infty$, \eqref{eq:EIrow1} is of order $n/h^{2+4\beta}$, which is dominated by \eqref{eq:EIrow3} that is of order $n^2/h^{2+2\beta}$; and, if $\beta>0.5$ as we assume henceforth, \eqref{eq:EIrow1} dominates \eqref{eq:EIrow2} that is of order $n/h^{2.5+3\beta}$. We thus conclude that 
\begin{equation}
\label{eq:EIsimpord}
\begin{aligned}   
    &\ \IE(I|\bX) \\
 =  &\ \left\{\frac{\eta(0,0)}{h^{1+2\beta}}f_U(X_j-x)+o_p\left(\frac{1}{h^{1+2\beta}}\right)\right\}\times \\
 &\ \sum_{\ell\ne k \ne j}\left(\frac{X_k-x}{h}\right)^2K_{h}(X_k-x)\left(\frac{X_\ell-x}{h}\right)^2 K_{h}(X_\ell-x). 
 \end{aligned}
\end{equation}
If $U$ is super smooth of order $\beta$, assuming $nh^{1-2\beta_2}\exp(-2h^{-\beta}/\gamma)\to \infty$, we have
\begin{align}   
     &\  \IE(I|\bX) \\
     = &\ O_p\left(\frac{\exp(4h^{-\beta}/\gamma)}{h^{3-4\beta_2}}\right) +O_p\left(\frac{n\exp(4h^{-\beta}/\gamma)}{h^{2-4\beta_2}}\right)+O_p\left(\frac{\exp(3h^{-\beta}/\gamma)}{h^{2-3\beta_2}}\right) \times \nonumber \\
     &\ 2\sum_{k\ne j}\left(\frac{X_k-x}{h}\right)^2 K_h(X_k-x)+O_p\left(\frac{\exp(2h^{-\beta}/\gamma)}{h^{1-2\beta_2}}\right) \times \nonumber \\
     &\ \sum_{\ell\ne k \ne j}\left(\frac{X_k-x}{h}\right)^2K_{h}(X_k-x)\left(\frac{X_\ell-x}{h}\right)^2 K_{h}(X_\ell-x) \nonumber\\
     = &\ O_p\left(\frac{\exp(2h^{-\beta}/\gamma)}{h^{1-2\beta_2}}\right)  \sum_{\ell\ne k \ne j}\left(\frac{X_k-x}{h}\right)^2K_{h}(X_k-x)\left(\frac{X_\ell-x}{h}\right)^2 K_{h}(X_\ell-x). \label{eq:EIsimpsup}
     \end{align} 
    Second, 
    \begin{align*}
         &\ \IE(II|\bX) \\
        = &\  \IE\left\{\left. K_{U,0,h}(W_j-x)K^2_{U,1,h}(W_j-x)K_{U,2,h}(W_j-x) \right \vert X_j\right\}+\\
        &\ \IE\left\{\left. K_{U,0,h}(W_j-x)K_{U,1,h}(W_j-x) \right\vert X_j\right\}\times \\
        &\ \sum_{k\ne j}\IE\left\{\left. K_{U,1,h}(W_k-x)K_{U,2,h}(W_k-x)\right \vert X_k\right\}+\\
        &\  \IE\left\{\left. K_{U,0,h}(W_j-x)K^2_{U,1,h}(W_j-x) \right \vert X_j\right\} \sum_{k\ne j} \IE\left\{\left. K_{U,2,h}(W_k-x)\right \vert X_k \right\}+\\
        &  \IE\left\{\left. K_{U,0,h}(W_j-x)K_{U,1,h}(W_j-x)K_{U,2,h}(W_j-x) \right \vert X_j\right\}\times \\
        &\ \sum_{\ell\ne j } \IE\left\{\left. K_{U,1,h}(W_\ell-x)\right \vert X_\ell\right\}+ \IE\left\{\left. K_{U,0,h}(W_j-x)K_{U,1,h}(W_j-x)\right \vert X_j \right\}\times \\
        &\ \sum_{\ell \ne k \ne j} 
        \IE\left\{\left. K_{U,1,h}(W_\ell-x)\right\vert X_\ell\right\}\IE\left\{\left. K_{U,2,h}(W_k-x)\right \vert X_k\right\}.
    \end{align*}
    By \eqref{eq:EKUlh} and \eqref{eq:EprodKs}, if $U$ is ordinary smooth of order $\beta$, then 
    \begin{align}
      &\ \IE(II|\bX) \nonumber \\
        = &\  \frac{\eta(0,1,1,2)}{h^{3+4\beta}}f_U(X_j-x) + o_p\left( \frac{1}{h^{3+4\beta}} \right) + \frac{\eta(0,1)\eta(1,2)}{h^{2+4\beta}}f_U(X_j-x)\times \nonumber \\
        &\ \sum_{k \ne j } f_U(X_k-x)+o_p\left( \frac{n}{h^{2+4\beta}} \right)+ \left\{\frac{\eta(0,1,1)}{h^{2+3\beta}}f_U(X_j-x)+o_p\left(\frac{1}{h^{2+3\beta}}\right)\right\}\times\nonumber \\
        &\ \sum_{k\ne j} \left(\frac{X_k-x}{h}\right)^2K_h(X_k-x)+ \left\{\frac{\eta(0,1,2)}{h^{2+3\beta}}f_U(X_j-x)+o_p\left( \frac{1}{h^{2+3\beta}} \right)\right\}\times \nonumber \\
        &\ \sum_{\ell \ne j} \left(\frac{X_\ell-x}{h}\right) K_h(X_\ell-x) +\left\{\frac{\eta(0,1)}{h^{1+2\beta}}f_U(X_j-x)+o_p\left(\frac{1}{h^{1+2\beta}} \right)\right\}\times \nonumber \\
        &\ \sum_{\ell\ne k \ne j} \left(\frac{X_\ell-x}{h}\right) K_h(X_\ell-x)\left(\frac{X_k-x}{h}\right)^2K_{h}(X_k-x) \nonumber \\
       = &\   \left\{\frac{\eta(0,1)}{h^{1+2\beta}}f_U(X_j-x)+o_p\left(\frac{1}{h^{1+2\beta}} \right)\right\}\times\nonumber \\
       &\ \sum_{\ell\ne k \ne j} \left(\frac{X_\ell-x}{h}\right) K_h(X_\ell-x)\left(\frac{X_k-x}{h}\right)^2K_{h}(X_k-x), \label{eq:EIIsimpord} 
    \end{align}
which is of the same order as that of \eqref{eq:EIsimpord}. 
If $U$ is super smooth of order $\beta$, 
\begin{align}
    &\ \IE(II|\bX) \nonumber \\
    = &\ O_p\left( \frac{\exp(4h^{-\beta}/\gamma)}{h^{3-4\beta_2}} \right) + O_p\left( \frac{n\exp(4h^{-\beta}/\gamma)}{h^{2-4\beta_2}} \right) \nonumber\\
        &\ O_p\left( \frac{\exp(3h^{-\beta}/\gamma)}{h^{2-3\beta_2}} \right)\sum_{k\ne j} \left(\frac{X_k-x}{h}\right)^2K_h(X_k-x)+\nonumber\\
        &\ O_p\left( \frac{\exp(3h^{-\beta}/\gamma)}{h^{2-3\beta_2}} \right)\sum_{\ell \ne j} \left(\frac{X_\ell-x}{h}\right) K_h(X_\ell-x) + \nonumber \\
        &\ O_p\left( \frac{\exp(2h^{-\beta}/\gamma)}{h^{1-2\beta_2}} \right)\sum_{\ell\ne k \ne j} \left(\frac{X_\ell-x}{h}\right) K_h(X_\ell-x)\left(\frac{X_k-x}{h}\right)^2K_{h}(X_k-x) \nonumber \\
        = &\ O_p\left( \frac{\exp(2h^{-\beta}/\gamma)}{h^{1-2\beta_2}} \right)\sum_{\ell\ne k \ne j} \left(\frac{X_\ell-x}{h}\right) K_h(X_\ell-x)\left(\frac{X_k-x}{h}\right)^2K_{h}(X_k-x), \label{eq:EIIsimpsup}
\end{align}
which is of the same order as that of \eqref{eq:EIsimpsup}. 
    Third, 
    \begin{align*}
       &\  \IE(III|\bX) \\
     =  &\ \IE\left\{\left. K^4_{U,1,h}(W_j-x) \right \vert X_j\right\} + \\
     &\ 2\IE\left\{\left. K^3_{U,1,h}(W_j-x)\right \vert X_j\right\}  \sum_{k\ne j} \IE\left\{\left. K_{U,1,h}(W_k-x)\right \vert X_k\right\} + \\
     &\ \IE\left\{\left. K^2_{U,1,h}(W_j-x)\right\vert X_j\right\}\sum_{k\ne j} \IE\left\{\left. K^2_{U,1,h}(W_k-x)\right \vert X_k\right\} +\\
     &\ \IE\left\{\left. K^2_{U,1,h}(W_j-x)\right\vert X_j\right\}\sum_{\ell\ne k \ne j} \IE\{K_{U,1,h}(W_\ell-x)|X_\ell\}\IE\{K_{U,1,h}(W_k-x)|X_k\}.
    \end{align*}
By \eqref{eq:EKUlh} and \eqref{eq:EprodKs}, if $U$ is ordinary smooth of order $\beta$, then 
    \begin{align}
     &\ E(III|\bX) \nonumber \\
     = &\ \frac{\eta(1,1,1,1)}{h^{3+4\beta}}f_U(X_j-x)+o_p\left(\frac{1}{h^{3+4\beta}} \right)+\left\{\frac{\eta(1,1,1)}{h^{2+3\beta}}f_U(X_j-x)+\right. \nonumber \\
     &\ \left. o_p\left(\frac{1}{h^{2+3\beta}}\right)\right\}\times 
     2\sum_{k\ne j} \left(\frac{X_k-x}{h}\right)K_h(X_k-x)+\left\{\frac{\eta(1,1)}{h^{1+2\beta}}f_U(X_j-x)+\right. \nonumber \\
     &\ \left. o_p\left(\frac{1}{h^{1+2\beta}} \right)\right\}  \sum_{k\ne j} \left\{\frac{\eta(1,1)}{h^{1+2\beta}}f_U(X_k-x)+o_p\left(\frac{1}{h^{1+2\beta}} \right)\right\}+ \nonumber\\
     &\ \left\{\frac{\eta(1,1)}{h^{1+2\beta}} f_U(X_j-x)+
      o_p\left(\frac{1}{h^{1+2\beta}} \right)\right\}\sum_{\ell \ne k\ne j} \left(\frac{X_\ell-x}{h}\right) K_h(X_\ell-x) \times \nonumber \\
      &\ \left(\frac{X_k-x}{h}\right) K_{h}(X_k-x) \nonumber\\
     = &\ \left\{\frac{\eta(1,1)}{h^{1+2\beta}}f_U(X_j-x)+o_p\left(\frac{1}{h^{1+2\beta}} \right)\right\}\sum_{\ell \ne k\ne j}\left(\frac{X_\ell-x}{h}\right) K_h(X_\ell-x)\times \nonumber \\
     &\ \left(\frac{X_k-x}{h}\right) K_{h}(X_k-x), \label{eq:EIIIsimpord}
    \end{align} 
    which is of the same order as that of \eqref{eq:EIsimpord}. If $U$ is super smooth of order $\beta$, then \ \begin{align}
    &\ E(III|\bX)\nonumber \\
     = &\ O_p\left(\frac{\exp(4h^{-\beta}/\gamma)}{h^{3-4\beta_2}} \right)+O_p\left(\frac{\exp(3h^{-\beta}/\gamma)}{h^{2-3\beta_2}} \right) \times 2\sum_{k\ne j} \left(\frac{X_k-x}{h}\right)K_h(X_k-x)+ \nonumber\\
      &\ O_p\left(\frac{\exp(2h^{-\beta}/\gamma)}{h^{1-2\beta_2}} \right)\times\sum_{\ell \ne k\ne j} \left(\frac{X_\ell-x}{h}\right) K_h(X_\ell-x)\left(\frac{X_k-x}{h}\right)K_{h}(X_k-x)+ \nonumber \\
      &\ 
      O_p\left(\frac{n\exp(4h^{-\beta}/\gamma)}{h^{2-4\beta_2}} \right) \nonumber\\
    = &\ O_p\left(\frac{\exp(2h^{-\beta}/\gamma)}{h^{1-2\beta_2}} \right)\sum_{\ell \ne k\ne j} \left(\frac{X_\ell-x}{h}\right) K_h(X_\ell-x)\left(\frac{X_k-x}{h}\right)K_{h}(X_k-x), \label{eq:EIIIsimpsup}
    \end{align}
which is of the same order as that of \eqref{eq:EIsimpsup}. 
Combining \eqref{eq:EIsimpord}, \eqref{eq:EIIsimpord}, and \eqref{eq:EIIIsimpord} for ordinary smooth $U$,  we have 
\begin{equation}
\label{eq:ELsqord}
    \begin{aligned}
     &\   \IE\left\{\left. \mathcal{L}^2(W_j-x)\right \vert \bX\right\} \\
       = &\ \frac{f_U(X_j-x)}{n^2h^{1+2\beta}}\left\{ \eta(0,0) \sum_{\ell\ne k \ne j}\left(\frac{X_k-x}{h}\right)^2K_{h}(X_k-x)\left(\frac{X_\ell-x}{h}\right)^2 K_{h}(X_\ell-x)
       \right.\\
       &\ -2 \eta(0,1) \sum_{\ell\ne k \ne j}\left(\frac{X_k-x}{h}\right)^2K_{h}(X_k-x)\left(\frac{X_\ell-x}{h}\right) K_{h}(X_\ell-x)\\
       &\ \left. +\eta(1,1) \sum_{\ell\ne k \ne j}\left(\frac{X_k-x}{h}\right)K_{h}(X_k-x)\left(\frac{X_\ell-x}{h}\right) K_{h}(X_\ell-x)\right\}\\
       &\ +o_p\left(\frac{1}{h^{1+2\beta}}\right),
    \end{aligned}
    \end{equation}
    and, using \eqref{eq:EIsimpsup}, \eqref{eq:EIIsimpsup}, and \eqref{eq:EIIIsimpsup} for super smooth $U$ gives 
\begin{equation}\label{eq:ELsqsup}
       \IE\left\{\left. \mathcal{L}^2(W_j-x)\right \vert \bX\right\}  = O_p\left(\frac{\exp(2h^{-\beta}/\gamma)}{h^{1-2\beta_2}}\right).         
     \end{equation}
    \subsection{Expectations of \texorpdfstring{$\hat{g}_{1,\text{DK}}(x)$}{Lg} and \texorpdfstring{$\hat{g}_{2,\text{DK}}(x)$}{Lg}} \label{sec:Eg1g2} 
    For $\hat{g}_{1,\text{DK}}(x) = 
    n^{-1}\sum_{j=1}^n \sin \Theta_j \mathcal{L}(W_j-x)$ and $\hat{g}_{2,\text{DK}}(x) = 
    n^{-1}\sum_{j=1}^n \cos \Theta_j $ \\ $\mathcal{L}(W_j-x)$, by \eqref{eq:locallinsingle},
    \begin{align}
        &\ \IE\left\{\left.\hat{g}_{\ell,\text{DK}}(x)\right \vert \bX\right\} \nonumber \\
       = &\ \frac{1}{n}\sum_{j=1}^n m_\ell(X_j)\IE\left\{\left. \mathcal{L}(W_j-x)\right \vert \bX\right\}  \nonumber\\ 
       = &\ \frac{1}{n}\sum_{j=1}^n m_\ell(X_j)\mathcal{W}(X_j-x) +\nonumber   \\
       &\ 
        \begin{cases}
    \displaystyle{\frac{\eta(0,2)-\eta(1,1)}{n^2h^{1+2\beta}}\sum_{j=1}^n m_\ell(X_j)f_{U}(X_{j}-x)}+ & \\
    \displaystyle{o_p\left(\frac{1}{nh^{1+2\beta}}\right)}, & \text{ if $U$ is ordinary smooth,}\\
   \displaystyle{O_p\left(\frac{\exp(2h^{-\beta}/\gamma)}{nh^{1-2\beta_2}} \right)}, & \text{ if $U$ is super smooth.}
        \end{cases} \label{eq:idealest}
     \end{align}
Using existing results (e.g., those in \eqref{eq:apppolyK}), one can show that 
    \begin{align*}
    \IE\left\{m_\ell(X)\mathcal{W}(X-x)\right\} = &\ \tau_\ell(x)f_X(x)\mu_2+
         \frac{h^2}{2}\{\tau^{(2)}_\ell(x)f_X(x)\mu^2_2+\\
         &\ \tau_\ell(x)f^{(2)}_X(x)\mu_4-2\tau'_\ell(x)f'_X(x)\mu^2_2\}+O(h^4),
   \end{align*}
where $\tau_\ell(x)=m_\ell(x)f_X(x)$.
In addition, by \eqref{eq:naiveprod*}, $\IE\{m_\ell(X)f_U(X-x)\} = \tau_\ell^*(x)$, where $\tau^*_\ell(x)=m^*_\ell(x)f_W(x)$. It follows that 
\begin{equation}
\label{eq:EgDK}
    \begin{aligned}
       &\ \IE\left\{\hat g_{\ell, \text{DK}}(x)|\bX\right\} \\
       = &\ \tau_\ell(x)f_X(x)\mu_2+
         \frac{h^2}{2}\left\{\tau^{(2)}_\ell(x)f_X(x)\mu^2_2+\tau_\ell(x)f^{(2)}_X(x)\mu_4 - \right. \\
         &\ 2\tau'_\ell(x)f'_X(x)\mu^2_2\Big\}+o(h^2)+\\
        &\ 
        \begin{cases}
    \displaystyle{\frac{\eta(0,2)-\eta(1,1)}{nh^{1+2\beta}} \tau_\ell^*(x) +o_p\left(\frac{1}{nh^{1+2\beta}}\right)}, & \text{ if $U$ is ordinary smooth,}\\
   \displaystyle{O_p\left(\frac{\exp(2h^{-\beta}/\gamma)}{nh^{1-2\beta_2}} \right)}, & \text{ if $U$ is super smooth.}
        \end{cases}
        \end{aligned}
\end{equation}    
Hence, $\hat g_{\ell, \text{DK}}(x)$ is an asymptotically unbiased estimator of $m_\ell(x)f^2_X(x)\mu_2$ if  $nh^{1+2\beta}\to \infty$ for ordinary smooth $U$, or $nh^{1-2\beta_2} \exp(-2h^{-\beta}/\gamma)\to \infty$ for super smooth $U$. If we further assume that $nh^{3+2\beta}\to \infty$ for ordinary smooth $U$, or $nh^{3-2\beta_2} \exp(-2h^{-\beta}/\gamma)\to \infty$ for super smooth $U$, then the dominating bias of $\hat g_{\ell, \text{DK}}(x)$ (as an estimator for $m_\ell(x)f^2_X(x)\mu_2$) is equal to that of $\tilde g_\ell(x)$, i.e., the counterpart estimator based on error-free data.

\subsection{Variances and covariance of $\hat{g}_{1,\text{DK}}(x)$ and $\hat{g}_{2,\text{DK}}(x)$}
\label{sec:varg}
Because $\text{Var}\{\hat g_{\ell, \text{DK}}(x)| \bX\}$ is dominated by $\IE\{\hat g^2_{\ell, \text{DK}}(x)| \bX\}$, we focus on this conditional expectation next.

Define $\xi_1(x)=\IE\{(\sin \Theta)^2|X=x\}$ and $\xi_2(x)=\IE\{(\cos \Theta)^2|X=x\}$. Results in Sections~\ref{sec:loclin} and~\ref{sec:loclin2} regarding moments of $\mathcal{L}(W_j-x)$ are now helpful in deriving the following expectation,  
    \begin{align*}
        \IE\left\{\left. \hat{g}^2_{1,\text{DK}}(x) \right\vert \bX\right\}
    =&\ \frac{1}{n^2}\IE\left[\left.\left\{\sum_{j=1}^{n} \sin \Theta_j \mathcal{L}(W_j-x)\right\}^2\right \vert \bX\right]\\
    =&\ \frac{1}{n^2}\IE\left\{\left.\sum_{j=1}^n (\sin \Theta_j)^2\mathcal{L}^2(W_j-x)\right \vert \bX \right\} \\
    &\ + \frac{1}{n^2}\IE\left\{\left.\sum_{j \ne k }\sin \Theta_j\mathcal{L}(W_j-x)\sin \Theta_k\mathcal{L}(W_k-x)\right \vert \bX\right\}\\
    = &\  \frac{1}{n^2}\sum_{j=1}^{n}\xi_1(X_{j})\IE \left\{\left.\mathcal{L}^2(W_j-x)\right \vert \bX\right\}+\\
    &\ \frac{1}{n^2}\sum_{j\ne k} m_1(X_j)m_1(X_k)\IE\{\mathcal{L}(W_j-x)\mathcal{L}(W_k-x)|\bX\},
    \end{align*}
where {\reviseagain the second summation is dominated by the first summation and thus is absorbed in the non-dominating (lower-order) terms of the first summation that is elaborated next}.

When $U$ is ordinary smooth, by \eqref{eq:ELsqord}, we have\vspace{-2mm}
\begin{align*}
        &\ \IE\left\{\left. \hat{g}^2_{1,\text{DK}}(x) \right\vert \bX\right\}\\
=&\ \frac{1}{n^4h^{1+2\beta}}\sum_{j=1}^n \xi_1(X_j) f_U(X_j-x)  \left\{ \eta(0,0) \sum_{\ell\ne k \ne j}\left(\frac{X_k-x}{h}\right)^2K_{h}(X_k-x)\times \right.\\
&\ \left(\frac{X_\ell-x}{h}\right)^2K_{h}(X_\ell-x)
        -2 \eta(0,1) \sum_{\ell\ne k \ne j}\left(\frac{X_k-x}{h}\right)^2K_{h}(X_k-x)\left(\frac{X_\ell-x}{h}\right) \\
        &\ K_{h}(X_\ell-x)
    \left. +\eta(1,1) \sum_{\ell\ne k \ne j}\left(\frac{X_k-x}{h}\right)K_{h}(X_k-x)\left(\frac{X_\ell-x}{h}\right) K_{h}(X_\ell-x)\right\}\\
    &\ +o_p\left(\frac{1}{nh^{2+2\beta}}\right).
\end{align*}    
Finally, using
\eqref{eq:anyYprod} and \eqref{eq:sumellkj}, one can show that 
$$\IE\{\hat g^2_{1,\text{DK}}(x)|\bX\}=\frac{\eta(0,0)}{nh^{1+2\beta}}\xi_1^*(x)f_W(x)f^2_X(x) \mu_2^2+o_p\left(\frac{1}{nh^{1+2\beta}}\right),$$ 
where $\xi^*_1(x)=\IE\{(\sin \Theta)^2|W=x\}$. Similarly, 
$$\IE\{\hat g^2_{2,\text{DK}}(x)|\bX\}=\frac{\eta(0,0)}{nh^{1+2\beta}}\xi_2^*(x)f_W(x)f^2_X(x) \mu_2^2+o_p\left(\frac{1}{nh^{1+2\beta}}\right),$$ 
where $\xi^*_2(x)=\IE\{(\cos \Theta)^2|W=x\}$. Hence, when $U$ is ordinary smooth, 
\begin{equation}
\text{Var}\{\hat g_{\ell, \text{DK}}(x)|\bX\}=\frac{\eta(0,0)}{nh^{1+2\beta}}\xi_\ell^*(x)f_W(x)f^2_X(x) \mu_2^2+o_p\left(\frac{1}{nh^{1+2\beta}}\right), \text{ for $\ell=1, 2$.}
\label{eq:vargDKord}
\end{equation}

When $U$ is super smooth, by \eqref{eq:ELsqsup}, one has 
\begin{equation}
\text{Var}\{\hat g_{\ell, \text{DK}}(x)|\bX\}=O_p\left(\frac{\exp(2h^{-\beta}/\gamma)}{nh^{1-2\beta_2}}\right), \text{ for $\ell=1, 2$.}
\label{eq:vargDKsup}
\end{equation}

Because
\scalebox{.99}{$\text{Cov}(\hat g_{1, \text{DK}}(x), \hat g_{2, \text{DK}}(x)|\bX)$} is dominated by $\IE\{\hat g_{1, \text{DK}}(x)\hat g_{2, \text{DK}}(x)|\bX\}$, we focus on deriving this conditional expectation next. 

Define $\psi(x)=\IE(\sin \Theta \cos \Theta|X=x)$, then we have
\begin{align*}
    &\ \IE\left
   \{\left.\hat{g}_{1,\text{DK}}(x)\hat{g}_{2,\text{DK}}(x)\right \vert \bX\right\} \\
   = &\  
   \frac{1}{n^2}\sum_{j=1}^{n}\IE\left\{\left. \sin \Theta_{j} \cos \Theta_{j}\mathcal{L}^2(W_{j}-x)\right \vert \bX\right\}+
   \\ &\
   \frac{1}{n^2}\sum_{j\ne k}\IE\left\{\left.\sin \Theta_{j} \mathcal{L}(W_{j}-x)\cos \Theta_k \mathcal{L}(W_k-x)\right \vert \bX\right\}\\ 
   = &\ \frac{1}{n^2}\sum_{j=1}^{n}\psi(X_{j})\IE\left\{\left.\mathcal{L}^2(W_{j}-x)\right\vert \bX\right\}+ \\
   &\ \frac{1}{n^2}\sum_{j\ne k}m_{1}(X_{j})m_{2}(X_k)\IE\left\{\left.\mathcal{L}(W_{j}-x)\mathcal{L}(W_k-x)\right \vert \bX\right\}.
\end{align*}
Using earlier findings for the above two conditional expectations, we obtain that, when $U$ is ordinary smooth, 
\begin{equation}
\text{Cov}\{\hat g_{1, \text{DK}}(x), \hat g_{2, \text{DK}}(x)|\bX\}=\frac{\eta(0,0)}{nh^{1+2\beta}}\psi^*(x)f_W(x)f^2_X(x) \mu_2^2+o_p\left(\frac{1}{nh^{1+2\beta}}\right),
\label{eq:covgDKord}
\end{equation}
where $\psi^*(x)=\IE(\sin \Theta \cos \Theta|W=x)$;
and, when $U$ is super smooth, 
\begin{equation}
\text{Cov}\{\hat g_{1, \text{DK}}(x), \hat g_{2, \text{DK}}(x)|\bX\}=O_p\left(\frac{\exp(2h^{-\beta}/\gamma)}{nh^{1-2\beta_2}}\right).
\label{eq:covgDKsup}
\end{equation}
    \subsection{Asymptotic bias of \texorpdfstring{$\hat{m}_{\text{DK}}(x)$}{Lg}}
    With the first and second asymptotic moments of $\hat{g}_{\ell, \text{DK}}(x)$ derived, we are ready to derive the asymptotic bias of $\hat{m}_{\text{DK}}(x)$. Suppressing the dependence on $x$ for simpler notations, we view $\hat m_{\cdot}$ as a bivariate function with arguments $(\hat g_{1,\cdot}, \hat g_{2, \cdot})$, i.e., $\hat m_{\cdot}=\text{atan2}[\hat g_{1,\cdot}, \hat g_{2, \cdot}]$. A Taylor expansion of $\hat{m}_{\cdot}$ around $(\hat g_{1,\cdot}, \hat g_{2,\cdot})=(g_1, g_2)$ is given by
    \begin{equation}
    \label{eq:mhatseries}
    \begin{aligned}
     &\ 
     \hat{m}_{\cdot}-m \\
     = &\ \frac{1}{g^2_{1}+g^2_{2}}\left\{g_2\left(\hat{g}_{1,\cdot}-g_1\right)-g_1\left(\hat{g}_{2,\cdot}-g_{2}\right)\right\}-\frac{1}{(g^2_{1}+g^2_{2})^{2}}\times \\
     &\ \left[g_{1}g_{2}\left\{\left(\hat{g}_{1,\cdot}-g_{1}\right)^{2}-\left(\hat{g}_{2,\cdot}-g_{2}\right)^{2}\right\}   +\left(g_{2}^{2}-g_{1}^{2}\right)\left(\hat{g}_{1,\cdot}-g_{1}\right)\left(\hat{g}_{2,\cdot}-g_{2}\right)\right]+\\
     &\ O_p\left\{\left(\hat{g}_{1,\cdot}-g_{1}\right)^{3}\right\}+O_p\left\{\left(\hat{g}_{2,\cdot}-g_{2}\right)^{3}\right\}.
    \end{aligned}
    \end{equation}
Taking the expectation of both sides of \eqref{eq:mhatseries} gives 
 \begin{equation}
    \label{eq:mhatbias}
    \begin{aligned}    &\ \text{Bias}\left(\hat{m}_{\cdot}\right) \\
    = &\  \frac{1}{g^2_{1}+g^2_{2}}\left\{g_2\IE(\hat{g}_{1,\cdot}-g_1)-g_1\IE(\hat{g}_{2,\cdot}-g_{2})\right\}-\frac{1}{(g^2_{1}+g^2_{2})^{2}}\left(g_{1}g_{2}\left[\IE\left\{(\hat{g}_{1,\cdot}-g_{1})^{2}\right\}\right.\right.\\
    &\ \left.-\IE\left\{(\hat{g}_{2,\cdot}-g_{2})^{2}\right\}\right] \left. +\left(g_{2}^{2}-g_{1}^{2}\right)\IE\left\{\left(\hat{g}_{1,\cdot}-g_{1}\right)\left(\hat{g}_{2,\cdot}-g_{2}\right)\right\}\right) \\
    &\ + O\{(\hat{g}_{1,\cdot}-g_{1})^{3}\}+O\{(\hat{g}_{2,\cdot}-g_{2})^{3}\},
        \end{aligned}
    \end{equation}
We consider the asymptotic results conditioning on $\bX$ in this appendix, and \eqref{eq:mhatbias} still holds when all the moments are conditional moments given $\bX$. 
    
For the deconvoluting kernel estimator $\hat m_{\text{DK}}(x)$, $g_\ell(x)=m_\ell(x) f^2_X(x)\mu_2$, for $\ell=1, 2$, in \eqref{eq:mhatbias}; $\IE(\hat g_{\ell, \text{DK}}-g_\ell|\bX)$ as the bias of $\hat g_{\ell. \text{DK}}$ is revealed by \eqref{eq:EgDK}; $\IE\{(\hat g_{\ell, \text{DK}}-g_\ell)^2|\bX\}=\text{Var}(\hat g_{\ell, \text{DK}}|\bX)+\{\text{Bias}(\hat g_{\ell, \text{DK}}|\bX)\}^2$ is dominated by $\text{Var}(\hat g_{\ell, \text{DK}}|\bX)$ given \eqref{eq:vargDKord} and \eqref{eq:vargDKsup}, with $\{\text{Bias}(\hat g_{\ell, \text{DK}}|\bX)\}^2$ absorbed in lower-order terms in $\text{Bias}(\hat g_{\ell, \text{DK}}|\bX)$; lastly, $\IE\{(\hat g_1-g_1)(\hat g_2-g_2)|\bX\}$ is equal to $\text{Cov}(\hat g_{1, \text{DK}}, \hat g_{2, \text{DK}}|\bX)+\text{Bias}(\hat g_{1,\text{DK}}|\bX)\text{Bias}(\hat g_{2, \text{DK}}|\bX)$, which is dominated by the first tem  $\text{Cov}(\hat g_{1,\text{DK}}, \hat g_{2,\text{DK}}|\bX)$ given  in \eqref{eq:covgDKord} and \eqref{eq:covgDKsup}, with the second term absorbed in lower order terms in $\text{Bias}(\hat g_{\ell, \text{DK}}|\bX)$. Therefore, putting \eqref{eq:EgDK}, \eqref{eq:vargDKord}, and \eqref{eq:covgDKord} in \eqref{eq:mhatbias}, we have, for ordinary smooth $U$,  
    \begin{align}
     &\text{Bias}\left\{\hat{m}_{\text{DK}}(x)|\bX\right\} \nonumber \\
   = &\ \frac{h^2 \mu_2}{2} \left\{\frac{m_1^{(2)}(x)m_2(x)-m_1(x)m_2^{(2)}(x)}{m_1^2(x)+m_2^2(x)}\right\} +o(h^2)+ \label{eq:DiMarziobias}  \\
   &\ \frac{f_W(x)}{nh^{1+2\beta}f^2_X(x)\mu_2 \left\{m_1^2(x)+m_2^2(x)\right\}}\times \nonumber \\
   &\ \bigg[\{\eta(0,2)-\eta(1,1)\}\{m_1^*(x)m_2(x)-m_1(x)m_2^*(x)\} -\nonumber \\
   &\ \left. \eta(0,0)\mu_2\times \frac{m_1(x)m_2(x)\left\{\xi^*_1(x)-\xi^*_2(x)\right\}+\psi^*(x) \left\{ m_2^2(x)-m_1^2(x)\right\}}{m_1^2(x)+m_2^2(x)} \right]+\nonumber \\
   &\ o_p\left(\frac{1}{nh^{1+2\beta}}\right).\label{eq:mDKbiasme}
    \end{align}
Viewing $m(x)=\text{atan2}[m_1(x), m_2(x)]$ as a univariate function with argument $x$, one has $$m'(x)=\frac{m_1'(x)m_2(x)-m_1(x)m_2'(x)}{m_1^2(x)+m_2^2(x)},  $$following which one can show that   
\begin{align*}
  \frac{m_1^{(2)}(x)m_2(x)-m_1(x)m_2^{(2)}(x)}{m_1^2(x)+m_2^2(x)}
  & = m^{(2)}(x)+2 m'(x) \frac{C'(x)}{C(x)}\\
  & = m^{(2)}(x)+2 m'(x)\frac{m_1(x)m'_1(x)+m_2(x)m'_2(x)}{m_1^2(x)+m_2^2(x)},
\end{align*}
where $C(x)=\sqrt{m_1^2(x)+m_2^2(x)}$. This reveals that the dominating term in \eqref{eq:DiMarziobias} coincides with the dominating bias of the ideal estimator $\tilde m(x)$ stated in \cite[][see Theorem 4]{di2013non}. 
For the second part of the dominating bias that is of order $1/(nh^{1+2\beta})$ following \eqref{eq:DiMarziobias} above, one can see that parts of it reduce to zero if $W=X$, i.e., in the absence of measurement error so that, for instance, $m^*_\ell(x)=m_\ell(x)$. This second part of the dominating bias reflects the effect of measurement error on the bias of $\hat m_{\text{DK}}(x)$. Certainly, if $nh^{3+2\beta}\to \infty$ as $n\to \infty$ and $h\to 0$, then this second term is absorbed in $o(h^2)$ in \eqref{eq:DiMarziobias}, suggesting that a larger sample size is required if we want $\hat m_{\text{DK}}(x)$ to achieve a comparable level of bias as that of $\tilde m(x)$. 

Similarly, if $U$ is super smooth, incorporating \eqref{eq:EgDK}, \eqref{eq:vargDKsup}, and \eqref{eq:covgDKsup} in \eqref{eq:mhatbias}, we have 
  \begin{align}\label{DKbias}\text{Bias}\left\{\hat{m}_{\text{DK}}(x)|\bX\right\}  
  = &\ \frac{h^2 \mu_2}{2}\times \left\{\frac{m_1^{(2)}(x)m_2(x)-m_1(x)m_2^{(2)}(x)}{m_1^2(x)+m_2^2(x)}\right\}\nonumber \\ 
  &\ +o(h^2)+O_p\left(\frac{\exp(2h^{-\beta}/\gamma)}{nh^{1-2\beta_2}}\right),
   \end{align}
where the effect of measurement error is (relatively more) vaguely reflected in the last term that depends on the order of smoothness $\beta$. However it is easy to see that if $nh^{4-2\beta_2}\exp(-2h^{-\beta}/\gamma)$ $\to \infty$ as $n\to \infty$ and $h\to 0$, then the dominating bias of $\hat m_{\text{DK}}(x)$ again coincides with that of $\tilde m(x)$. In other words, in the presence of super smooth measurement error, an even larger sample size than what is needed when $U$ is ordinary smooth would be required for the bias of $m_{\text{DK}}(x)$ to be comparable with that of $\tilde m(x)$. These give the results in Theorem 1 in the main article regarding the asymptotic bias of $m_{\text{DK}}(x)$.

    \subsection{Asymptotic variance of \texorpdfstring{$\hat{m}_{\text{DK}}(x)$}{Lg}}\label{dcvar}
Now viewing $\hat m^2_{\cdot}$ as a bivariate function with arguments $(\hat g_{1,\cdot}, \hat g_{2,\cdot})$, we have a Taylor expansion of $\hat m^2_{\cdot}$ around $(\hat g_{1,\cdot}, \hat g_{2,\cdot})=(g_1, g_2)$ given by  
\begin{align*}
&\ \hat m_\cdot^2\\
= &\ \ m^2+\frac{2m}{g_1^2+g_2^2}\left\{g_2 \left(\hat g_{1,\cdot}-g_1\right)- g_1 \left(\hat g_{2,\cdot}-g_2\right)\right\} + \frac{2m}{\left(g_1^2+g_2^2\right)^2}\times \\
&\ \left[g_1g_2 \left\{\left(\hat g_{2,\cdot}-g_2\right)^2 -\left(\hat g_{1,\cdot}-g_1\right)^2\right\}+\left(g_1^2-g_2^2\right)\left(\hat g_{1,\cdot}-g_1\right)\left(\hat g_{2,\cdot}-g_2\right)
\right]+\\
&\  \frac{1}{\left(g_1^2+g_2^2\right)^2}\left\{ g_2^2 \left(\hat g_{1,\cdot}-g_1\right)^2 +g_1^2 \left(\hat g_{2,\cdot}-g_2\right)^2-2g_1g_2 \left(\hat g_{1,\cdot}-g_1\right)\left(\hat g_{2,\cdot}-g_2\right)\right\}+\\
&\ O_p\left\{\left(\hat g_{1,\cdot}-g_1\right)^3\right\}+O_p\left\{\left(\hat g_{2,\cdot}-g_2\right)^3\right\}.
\end{align*}  
Taking conditional expectations of both sides of the above equation and dropping terms that can be absorbed in lower-order terms of $\IE(\hat g_{\ell, \cdot}-g_\ell|\bX)$ {\reviseagain in \eqref{eq:EgDK}}, we have 
\begin{equation}
\label{eq:Emhatsq}
\begin{aligned}
&\ \IE\left(\hat m^2_{\cdot}\right|\bX)\\
= &\  m^2 +\frac{2m}{g_1^2+g_2^2}\left\{g_2 \text{Bias}\left( \hat g_{1,\cdot}|\bX\right)-g_1 \text{Bias}\left( \hat g_{2,\cdot}|\bX\right)\right\}+\\
&\  \frac{1}{\left(g_1^2+g_2^2\right)^2}\left\{ g_1 (2mg_2+g_1)\text{Var}\left(\hat g_{2,\cdot}|\bX\right)-g_2 (2mg_1-g_2)\text{Var}\left(\hat g_{1,\cdot}|\bX\right)\right\}+\\
&\ \frac{2}{\left(g_1^2+g_2^2\right)^2}\left\{ m\left( g_1^2-g_2^2\right)-g_1g_2\right\}\text{Cov}\left( \hat g_{1,\cdot}, \hat g_{2,\cdot}|\bX\right)+\text{lower-order terms}.
\end{aligned}
\end{equation}

Using results in Sections~\ref{sec:Eg1g2} and \ref{sec:varg}, we have, for ordinary smooth $U$, 
\begin{equation}
\label{eq:EmDK2}
\begin{aligned}
&\ \IE\left\{\hat m^2_{\text{DK}}(x)|\bX\right\} \\
= &\ m^2(x)+ h^2 m(x) \mu_2 \left\{ m^{(2)}(x)+2 m'(x) \frac{C'(x)}{C(x)}\right\}+o_p(h^2)+\\
&\ \frac{2f_W(x)}{nh^{1+2\beta}f_X^2(x) \left\{ m_1^2(x)+m_2^2(x)\right\}}\bigg \{ m(x)\{\eta(0,2)-\eta(1,1)\}\left\{m_1^*(x)m_2(x)-\right.\\
&\ \left. m_1(x)m_2^*(x)\right\}/\mu_2 +\frac{\eta(0,0)}{m_1^2(x)+m_2^2(x)} \bigg( m(x) m_1(x)m_2(x) \left\{ \xi^*_2(x)-\xi^*_1(x)\right\}+  \\
&\ \left.\frac{m_1^2(x)\xi^*_2(x)+m_2^2(x)\xi^*_1(x)}{2}\right.  + \psi^*(x) \left[m(x) \left\{m_1^2(x)-m_2^2(x)\right\}\right. \\
&\ \left. -m_1(x)m_2(x)\right] \bigg)\bigg \} +o_p\left(\frac{1}{nh^{1+2\beta}}\right), 
\end{aligned}
\end{equation}

and, for super smooth $U$,
\begin{align*}
    &\ \IE\left\{\hat m^2_{\text{DK}}(x)|\bX\right\} \\
= &\ m^2(x)+ h^2 m(x) \mu_2 \left\{ m^{(2)}(x)+2m'(x) \frac{C'(x)}{C(x)}\right\}+o_p(h^2) +O_p\left(\frac{\exp(2h^{-\beta}/\gamma)}{nh^{1-2\beta_2}}\right).
\end{align*}

Focusing on the term of order $O(h^2)$ and the term of order $O(1/(nh^{1+2\beta}))$ for ordinary smooth $U$, we have \begin{equation}\label{DKm2}\IE\{\hat m^2_{\text{DK}}(x)|\bX\}=m^2(x)+h^2\times A+\frac{1}{nh^{1+2\beta}} \times B +o(h^2)+o_p\left(\frac{1}{nh^{1+2\beta}}\right)\end{equation} and, with $\text{Bias}\{\hat m_{\text{DK}}(x)|\bX\}=h^2\times C+ \{1/(nh^{1+2\beta})\}\times D+o_p(1/(nh^{1+2\beta}))$, 
\begin{align*}
&\ \left[\IE\left\{\hat m_{\text{DK}}(x)|\bX\right\}\right]^2 \\
 = &\ \left[\text{Bias}\{\hat m_{\text{DK}}(x)|\bX\}+m(x)\right]^2\\
= &\ \text{Bias}^2\{\hat m_{\text{DK}}(x)|\bX\}+m^2(x)+2m(x)\text{Bias}\{\hat m_{\text{DK}}(x)|\bX\}\\
= &\ m^2(x)+2 m(x) \left (h^2\times C+\frac{1}{nh^{1+2\beta}}\times D\right)+o(h^2)+o_p\left(\frac{1}{nh^{1+2\beta}}\right),
\end{align*}
where $A$ and $B$ are indicated by \eqref{eq:EmDK2}, $C$ and $D$ are implied by \eqref{eq:DiMarziobias} and \eqref{eq:mDKbiasme}. In fact, $A-2m(x)C=0$, and thus, for ordinary smooth $U$, 
\begin{align*}
   &\ \text{Var}\{\hat m_{\text{DK}}(x)|\bX\} \\
   = &\ h^2 \{A-2m(x)C\}+\frac{1}{nh^{1+2\beta}}\{B-2m(x) D\}+o(h^2)+o_p\left(\frac{1}{nh^{1+2\beta}}\right)\\
  =  &\ \frac{f_W(x)\eta(0,0)}{nh^{1+2\beta}f_X^2(x)\{m_1^2(x)+m_2^2(x)\}^2} \{ m_1^2(x)\xi^*_2(x)+m_2^2(x)\xi_1^*(x)\\
  &\ -2m_1(x)m_2(x) \psi^*(x)\}+o_p\left(\frac{1}{nh^{1+2\beta}}\right);
\end{align*}
and, for super smooth $U$, 
$$\text{Var}\{\hat m_{\text{DK}}(x)|\bX\}=O_p\left(\frac{\exp(2h^{-\beta}/\gamma)}{nh^{1-2\beta_2}}\right).$$
  
\renewcommand*{\theequation}{E.\arabic{equation}}
\setcounter{equation}{0}
\renewcommand*{\thesubsection}{E.\arabic{subsection}}
\renewcommand*{\thelemma}{E.\arabic{lemma}}
\setcounter{subsection}{0}
\section*{{\revise Appendix E: Asymptotic properties of the complex error estimator $\hat m_{\text{CE}}(x)$}}

We consider the $p$-th order local polynomial weight used in $\hat m_{\text{CE}}(x)$ in this appendix, for $p\ge 0$. 
Throughout this appendix, we assume normal measurement error and conditions imposed in Lemma~\ref{lem:addce2} in the main article are satisfied, such as relevant functions being entire functions.

\subsection{Expectations of $\hat{g}_{1,\text{CE}}(x)$ and $\hat g_{2,\text{CE}}(x)$}
\label{sec:EgCE}
Recall that 
\begin{align*}      
\hat g_{1,\text{CE}}(x) &= \frac{1}{n}\sum_{j=1}^n \sin \Theta_j \mathcal{L}^*(W_j-x) \\
& = \frac{1}{n}\sum_{j=1}^n \sin \Theta_j  \times \frac{1}{B}\sum_{b=1}^{B} \mathcal{W}^*(W_{j,b}^*-x; \bW_b^*) \\
& =  \frac{1}{B}\sum_{b=1}^{B} \frac{1}{n}  \sum_{j=1}^n \sin \Theta_j \mathcal{W}^*(W_{j,b}^*-x; \bW_b^*) \\
& = \frac{1}{B}\sum_{b=1}^{B} \hat g^*_{1,b}(x),
\end{align*}
where $\hat g^*_{1,b}(x)=n^{-1}  \sum_{j=1}^n \sin \Theta_j \mathcal{W}^*(W_{j,b}^*-x; \bW_b^*)$ is the same type of estimator as the ideal local linear estimator $\tilde g_1(x)$ for $g_1(x)=m_1(x)=\IE(\sin \Theta|X=x)$, with the complex-valued weight $\mathcal{W}^*(W_{j,b}^*-x; \bW_b^*)$ in place of the normalized local linear weight $\mathcal{W}^*(X_j-x; \bX)$ in (5) in the main article, or, more generically, in place of the naive local polynomial weight of order $p$. By Lemma~\ref{lem:addce2}, $\IE\{\hat g^*_{1,b}(x)|\bX, \bTheta\}=\tilde g_1(x)$, and thus $\IE\{\hat g_{1,\text{CE}}(x)|\bX,\bTheta\}=\tilde g_1(x)$. It follows that, if $\sigma^2_1(\cdot)=\text{Var}(\sin \Theta|X=\cdot)=\xi_1(\cdot)-m_1^2(\cdot)$ is continuous in a neighborhood of $x$, and that $nh\to \infty$ as $n\to \infty$ and $h\to 0$, then
\begin{align}
&\ \IE\{\hat g_{1,\text{CE}}(x)|\bX\} \nonumber \\
= &\ \IE\left[ \left.\IE\left\{\hat g_{1,\text{CE}}(x)|\bX,\bTheta\right \}\right \vert \bX\right] \nonumber\\
= &\ \IE\left\{\tilde g_1(x)|\bX\right\}, \text{ next use Theorem 3.1 in \cite{Fan},} \nonumber    \\
= &\ 
\begin{cases}
 m_1(x) + \be_1^\top \bS^{-1} \tilde \bc_p \frac{1}{(p+2)!}\left\{m_1^{(p+2)}(x)+ \right.\\
 \left.(p+2)m_1^{(p+1)}(x) \frac{f'_X(x)}{f_X(x)}\right\}h^{p+2}
 +o_p(h^{p+2}), & \text{ if $p$ is even,}\\  
 m_1(x)+ \be_1^\top \bS^{-1}\bc_p \frac{1}{(p+1)!}m_1^{(p+1)}(x)h^{p+1}+o_p(h^{p+1}), & \text{ if $p$ is odd,}
\end{cases}
\label{eq:Egtilde1}
\end{align}
where $\bS=(\mu_{r_1+r_2})_{0\le r_1, r_2\le p}$, $\bc_p=(\mu_{p+1}, \ldots, \mu_{2p+1})^\top$,  $\tilde\bc_p=(\mu_{p+2}, \ldots, \mu_{2p+2})^\top$, and $\be_1$ is the $(p+1)\times 1$ unit-norm vector with the first entry equal to 1. 
Similarly, if $\sigma^2_2(\cdot)=\text{Var}(\cos \Theta|X=\cdot)=\xi_2(\cdot)-m_2^2(\cdot)$ is continuous in a neighborhood of $x$, and $nh\to \infty$ as $n\to \infty$ and $h\to 0$, then
\begin{align}
&\ \IE\{\hat g_{2,\text{CE}}(x)|\bX\} \nonumber \\ 
= &\  
\begin{cases}
 m_2(x) + \be_1^\top \bS^{-1} \tilde \bc_p \frac{1}{(p+2)!}\left\{m_2^{(p+2)}(x)+ \right.\\
 \left.(p+2)m_2^{(p+1)}(x) \frac{f'_X(x)}{f_X(x)}\right\}h^{p+2}
 +o_p(h^{p+2}), & \text{ if $p$ is even,}\\  
 m_2(x)+ \be_1^\top \bS^{-1}\bc_p \frac{1}{(p+1)!}m_2^{(p+1)}h^{p+1}+o_p(h^{p+1}), & \text{ if $p$ is odd,}
\end{cases}
\label{eq:Egtilde2} 
\end{align}

In the special case considered in the main article with $p=1$, we have 
\begin{equation*}
 \IE\{\hat g_{\ell,\text{CE}}(x)|\bX\} = m_\ell(x) + \frac{h^2}{2}\mu_2m_\ell^{(2)}(x)+O_p(h^4), \text{ for $\ell=1, 2$}. 
\end{equation*}

\subsection{Variances of $\hat{g}_{1,\text{CE}}(x)$ and $\hat{g}_{2,\text{CE}}(x)$} \label{sec:VargCE}
By Lemma~\ref{lem:addce2}, $\IE\{\hat g^*_{1,b}(x)|\bX, \bTheta\}=\tilde g_1(x)$,  and $\IE[\{\hat g^{*}_{1,b}(x)\}^2|\bX, \bTheta]=\tilde g^2_1(x)$, thus
\begin{align}
\text{Var}\left\{ \left. \hat g^*_{1,b}(x) \right \vert \bX\right\} & = \IE\left[\left. \left\{ \hat g^*_{1,b}(x)\right\}^2 \right \vert \bX\right]-\left[\IE\left. \left\{ \hat g^*_{1,b}(x) \right \vert \bX\right\}\right]^2 \nonumber\\
& = \IE\left\{\left. \tilde g^2_1(x) \right \vert \bX\right\}-\left[\IE\left. \left\{ \tilde g_1(x) \right \vert \bX\right\}\right]^2 \nonumber\\
& =\text{Var}\{\tilde g_1(x)|\bX\} \nonumber \\
& =\be_1 ^\top\bS^{-1} \bnu^*\bS^{-1}\be_1\frac{\sigma_1^2(x)}{nh f_X(x)}+o_p\left(\frac{1}{nh}\right),\label{eq:varg1b*}
\end{align}
where $\bnu^*=(\nu_{r_1+r_2})_{0\le r_1, r_2\le p}$, and the last equation comes from the asymptotic variance of a local polynomial estimator of $m_1(x)$ in Theorem 3.1 in \cite{Fan} under the assumption that $\sigma^2_1(\cdot)$ is continuous in a neighborhood of $x$, and that $nh\to \infty$ as $n\to \infty$ and $h\to 0$.
It follows that
\begin{align}
&\ \text{Var}\left\{\hat{g}_{1,\text{CE}}(x)|\bX\right\} \nonumber\\
=&\ \text{Var}\left\{ \left. \frac{1}{B}\sum_{b=1}^{B}\hat g^{*}_{1,b}(x) \right \vert \bX\right\}, \nonumber\\
=&\frac{1}{B^{2}}\sum_{b=1}^{B}\text{Var}\left\{\left. \hat g^{*}_{1,b}(x)\right \vert \bX\right\}+\frac{1}{B^{2}}\sum_{a \ne b}\text{Cov}\left\{\left. \hat g^{*}_{1,a}(x), \hat g^{*}_{1,b}(x) \right \vert \bX\right\}\nonumber\\
= &\frac{1}{B}\text{Var}\left\{\left. \tilde g_1(x)\right \vert \bX\right\}+\frac{B-1}{B}\left[\IE\left\{\left.\hat g^{*}_{1,a}(x)\hat g^{*}_{1,b}(x) \right \vert \bX \right\}-\right. \nonumber \\
&\ \left. \IE\left\{\left.\hat g^*_{1,a}(x)\right\vert \bX \right\}\IE\left\{\left.\hat g^*_{1,b}(x)\right\vert \bX \right\}\right], \text{ where $a\ne b$,} \nonumber \\
= &  \be_1^\top \bS^{-1}\bnu^*\bS^{-1}\be_1\frac{ \sigma_1^2(x)}{Bnh f_X(x)}+o_p\left(\frac{1}{Bnh}\right)+ \frac{B-1}{B}\IE\left\{\left.\hat g^{*}_{1,a}(x)\hat g^{*}_{1,b}(x) \right \vert \bX \right\}-\nonumber\\
&\ \frac{B-1}{B} \left[\IE\left\{\left. \tilde g_1(x)\right\vert \bX \right\}\right]^2, \text{ by \eqref{eq:varg1b*}.}
\label{eq:yet2find} 
\end{align}
Assume that $nh^{p+3}\to \infty$, if $p$ is even, or, if $p$ is odd, assume $nh^{p+2}\to \infty$, then the first two terms (of order $1/(nh)$ at a fixed $B$) are included in $o_p(h^{p+1+I(\text{$p$ is even})})$ that is the order of the last term. 
In what follows, we look into the only expectation left unsolved above, \scalebox{.94}{$\IE\{\hat{g}_{1,a}^{*}(x)\hat{g}_{1,b}^{*}(x)|\boldsymbol{X}\}$,} which requires new tactics to solve due to its dependence on two sets of complex-valued error-contaminated covariate values, $\bW^*_a=(W^*_{1, a}, \ldots, W^*_{n,a})^\top$ and \\ $\bW^*_b=(W^*_{1, b}, \ldots, W^*_{n,b})^\top$. 

Adopting notations for weighted least squares estimators as in \cite[Chapter 3,][]{Fan}, we have
    $$\hat{g}_{1,a}^{*}(x)\hat{g}_{1,b}^{*}(x)=\boldsymbol{e}_{1}^{\top}(\boldsymbol{G}_{a}^{\top}\bOmega_{a}\boldsymbol{G}_{a})^{-1}\boldsymbol{G}_{a}^{\top}\bOmega_a\boldsymbol{Y}\boldsymbol{Y}^{\top}\bOmega_b\boldsymbol{G}_{b}(\boldsymbol{G}_{b}^{\top}\boldsymbol{\Omega}_{b}\boldsymbol{G}_{b})^{-1}\boldsymbol{e}_{1},$$
    where $\boldsymbol{Y}=(\sin\Theta_{1},\ldots,\sin\Theta_{n})^{\top}$, $\bOmega_a=\text{diag}(K_{h}(W^*_{1,a}-x),\ldots, K_{h}(W^*_{n,a}-x))$, 
    $$\bG_a=
    \begin{bmatrix} 
    1 & W^*_{1,a}-x & \ldots & (W_{1,a}^*-x)^p\\
    \vdots & \vdots & \ddots & \vdots \\
    1 & W^*_{n,a}-x & \ldots & (W_{n,a}^*-x)^p
    \end{bmatrix},$$
   and $\bOmega_b$ and $\bG_b$ are similarly defined. Recall that, for $j=1, \ldots, n$, $W^*_{j,a}=W_j+ i \sigma_u Z_{j,a}$ and $W^*_{j,b}=W_j+i \sigma_u Z_{j,b}$, with $Z_{j,a} \perp Z_{j,b}$ from $N(0, 1)$ for $a\ne b$. Utilizing the law of iterated expectation, we first look into the expectation conditioning on $\bW^*_a$, $\bW^*_b$, and $\bX$,
    \begin{align}
    &\ \IE\left\{\left.\hat{g}_{1,a}^{*}(x)\hat{g}_{1,b}^{*}(x) \right \vert\bW^*_a, \bW^*_b,\boldsymbol{X}\right\} \nonumber\\
    = &\ \boldsymbol{e}_{1}^{\top}(\boldsymbol{G}_{a}^{\top}\bOmega_{a}\boldsymbol{G}_{a})^{-1}\boldsymbol{G}_{a}^{\top}\bOmega_{a}\IE(\bY\bY^{\top}|\boldsymbol{X})\bOmega_{b}\boldsymbol{G}_{b}(\boldsymbol{G}_{b}^{\top}\bOmega_{b}\boldsymbol{G}_{b})^{-1}\boldsymbol{e}_{1} \nonumber \\
    = &\ \boldsymbol{e}_{1}^{\top}(n^{-1}\boldsymbol{G}_{a}^{\top}\bOmega_{a}\boldsymbol{G}_{a})^{-1}(n^{-2}\boldsymbol{G}_{a}^{\top}\bOmega_{a}\boldsymbol{\Sigma} \bOmega_{b}\boldsymbol{G}_{b})(n^{-1}\boldsymbol{G}_{b}^{\top}\bOmega_{b}\boldsymbol{G}_{b})^{-1}\boldsymbol{e}_{1}, \label{eq:condexp1}
    \end{align}
    where
$$\boldsymbol{\Sigma}=\IE\left(\left.\bY\bY^\top \right \vert\bX\right)=
\scalebox{.9}{$\begin{bmatrix}
    \xi_1(X_{1})& m_1(X_{1})m_1(X_{2}) & \ldots &m_{1}(X_{1})m_{1}(X_{n}) \\
    m_{1}(X_{2})m_{1}(X_{1})& \xi_1(X_{2})& \ldots &m_{1}(X_{2})m_{1}(X_{n}) \\
    \vdots & \vdots & \ddots &\vdots \\
    m_{1}(X_{n})m_{1}(X_{1})& m_{1}(X_{n})m_{1}(X_{2})&\ldots & \xi_1(X_{n})
\end{bmatrix}.$}$$

Because 
$$n^{-1}\bG_{a}^{\top}\bOmega_{a}\bG_{a} = \left[\frac{1}{n}\sum_{j=1}^{n} (W^*_{j,a}-x)^{\ell+q-2}K_{h}(W^*_{j,a}-x)\right]_{1\le \ell, q\le (p+1)},$$
with the expectation of the $(\ell, q)$ entry  conditioning on $\bX$ equal to, by Lemma 1, for $\ell, q=1, \ldots,  p+1$,
\begin{align*}
 &\ \frac{1}{n}\sum_{j=1}^{n} \IE\left\{\left.(W^*_{j,a}-x)^{\ell+q-2}K_{h}(W^*_{j,a}-x)\right\vert\boldsymbol{X}\right\} \\
= &\  \frac{1}{n} \sum_{j=1}^{n}(X_j-x)^{\ell+q-2}K_{h}(X_{j}-x) \\
= &\ h^{\ell+q-2}\mu_{\ell+q-2} f_X(x)+h^{\ell+q-1}\mu_{\ell+q-1} f'_X(x)+h^{\ell+q}\mu_{\ell+q} f^{(2)}_X(x)/2\\
&\ +O_p\left(\frac{h^{\ell+q-2}}{\sqrt{nh}}\right),
\end{align*}
therefore, 
\begin{equation}
\label{eq:left&right}
\begin{aligned}
    &\ \IE\left(\left.n^{-1}\bG_a^\top \bOmega_a\bG_a \right \vert \bX\right)
     \\
     = &\ \bH\left\{f_X(x)\bS +hf'_X(x) \bS_1 +h^2 f_X^{(2)}(x)\bS_2/2+O_p\left(\frac{1}{\sqrt{nh}}\right)\right\}\bH,
\end{aligned}
\end{equation}
where $\bH=\text{diag}(1, h, \ldots, h^p)$, $\bS_\ell=(\mu_{r_1+r_2+\ell})_{0\le r_1, r_2\le p}$, for $\ell=1, 2$. This also gives $\IE\{n^{-1}\bG_b^\top \bOmega_b\bG_b|\bX\}$.

Next we study the middle portion of \eqref{eq:condexp1}, $n^{-2}\bG_a^\top\bOmega_a\bSigma \bOmega_b \bG_b$. The $(\ell, q)$- entry of this matrix is
\begin{equation}
\label{eq:middleportion}
\begin{aligned}
&\ (n^{-2}\bG_a^\top\bOmega_a\bSigma \bOmega_b \bG_b)_{\ell,q} \\
= &\
\frac{1}{n^2}\sum_{j=1}^{n}\sum_{k=1}^{n}K_{h}(W^*_{j,b}-x)K_{h}(W^*_{k,a}-x)\times\\ &\ \IE(Y_jY_k|\bX)
(W^*_{j,b}-x)^{q-1}(W^*_{k,a}-x)^{\ell-1}, \text{ for $\ell, q=1, \ldots, p+1$}.
\end{aligned}
\end{equation}
For the summand of the above double sum, we have, if $j\ne k$, the expectation of the summand given $\bX$ is 
\begin{align*}
&\ m_1(X_j)(X_j-x)^{q-1}K_{h}(X_j-x)m_1(X_k)(X_{k}-x)^{\ell-1}K_{h}(X_k-x)\\
=&\ h^{q+\ell-2}\mu_{q-1}\mu_{\ell-1}\tau^2_1(x)+h^{q+\ell-1}(\mu_{q-1}\mu_\ell+\mu_q\mu_{\ell-1})\tau_1(x)\tau'_1(x)+\\
&\ \frac{h^{q+\ell}}{2}\left[\mu_{q-1}\mu_{\ell+1}\tau_1(x)\tau_1^{(2)}(x)+\mu_{q+1}\mu_{\ell-1}\tau_1(x)\tau_1^{(2)}(x)+2 \mu_q\mu_\ell \{\tau_1'(x)\}^2\right]+\\
&\ O_p\left(h^{q+\ell-3}\right).
\end{align*}
If $j=k$, the expectation of the summand given $\bX$ is 
\begin{align}
&\ \IE\left\{\left.K_{h}(W^*_{j,b}-x)K_{h}(W^*_{j,a}-x)
(W^*_{j,b}-x)^{q-1}(W^*_{j,a}-x)^{\ell-1} \right \vert X_j\right\}\xi_1(X_j)\nonumber \\
= &\  \IE\left[ \IE\left\{\left.(W^*_{j,a}-x)^{\ell-1}K_{h}(W^*_{j,a}-x)\right\vert W_j, X_j\right\}\times \right. \nonumber \\
&\ \quad \left. \left.\IE\left\{\left.
(W^*_{j,b}-x)^{q-1}K_{h}(W^*_{j,b}-x) \right \vert W_j, X_j\right\} \right \vert X_j\right]\xi_1(X_j) \nonumber \\
= &\ \IE\left[\IE\left\{\left.(W^*_{j,a}-x)^{\ell-1}K_{h}(W^*_{j,a}-x)\right \vert W_j\right\}\times \right. \nonumber \\
&\ \quad \left. \left. \IE\left\{\left.
(W^*_{j,b}-x)^{q-1}K_{h}(W^*_{j,b}-x) \right \vert W_j\right\} \right\vert X_j\right]\xi_1(X_j).
\label{eq:varcross1}
\end{align}
To find the inner expectation in \eqref{eq:varcross1}, we consider the following Fourier transform, 
\begin{align}
 \phi_{(W^*-x)^{k}K_{h}(W^*-x)}(t) 
 & = \int (W^*-x)^{k}K_{h}(W^*-x) e^{itx}\,dx \nonumber \\
 & = \exp(itW^*)h^k \int u^k K(u) e^{-ithu}\, du \nonumber \\
 & = \exp(itW^*)h^k i^{-k} \left.\frac{\partial^k \phi_K(s)}{\partial s^k}
 \right\vert_{s=-ht} \nonumber \\
 & = \exp(itW^*)h^k i^{-k} \phi_K^{(k)}(-ht). \label{eq:charfun} 
\end{align}
Using the inverse Fourier transform, we elaborate the inner expectation in \eqref{eq:varcross1} as follows,
\begin{align*}
&\ \IE\left\{\left.(W^*_{j,a}-x)^{\ell-1}K_{h}(W^*_{j,a}-x) \right \vert W_j\right\}\\
= &\ \int (W_j+i\sigma_U z-x)^{\ell-1} K_{h}(W_j+i\sigma_U z-x)f_{Z}(z) \,dz \\
= &\ \int \frac{1}{2\pi}\int e^{-itx}\phi_{(W_j+i\sigma_U z-x)^{\ell-1}K_{h}(W_j+i \sigma_U z-x)}(t) \, dt \times \frac{1}{\sqrt{2\pi}}e^{-z^2/2}\, dz \\
= &\ \frac{1}{2\pi} h^{\ell-1} i^{-(\ell-1)}\int\exp\{it(W_j-x)\}\phi_{K}^{(\ell-1)}(-ht)\times \\
&\ \frac{1}{\sqrt{2\pi}}\int \exp(-t\sigma_U z-z^2/2)\, dzdt, \text{ by \eqref{eq:charfun},}\\
= &\ \frac{1}{2\pi }h^{\ell-1} i^{-(\ell-1)}\int \exp\{it(W_j-x)\}\phi_{K}^{(\ell-1)}(-ht) \exp(t^2\sigma_U^2/2) \, dt\\
= &\ h^{\ell-1} i^{-(\ell-1)}\frac{1}{2\pi }\int \exp\{it(W_j-x)\}\frac{\phi_{K}^{(\ell-1)}(-ht)}{\phi_U(t)}\, dt\\
= &\ h^{\ell-2} i^{-(\ell-1)}\frac{1}{2\pi} \int \exp\left\{-is \left(\frac{W_j-x}{h}\right) \right\} \frac{\phi_{K}^{(\ell-1)}(s)}{\phi_U(-s/h)}\, ds\\
= &\ h^{\ell-1}K_{U, \ell-1, h}(W_j-x), \text{ by \eqref{eq:KUlh}}.
\end{align*}
Similarly, $\IE\{(W^*_{j,b}-x)^{q-1}K_{h}(W^*_{j,b}-x)|W_j\}=h^{q-1}K_{U, q-1, h}(W_j-x)$. Hence, \eqref{eq:varcross1} reduces to, by \eqref{eq:EprodKs},
\begin{align*}
   &\ h^{\ell+q-2} \IE\left\{K_{U, \ell-1, h}(W_j-x)K_{U, q-1, h}(W_j-x)|X_j\right\}\xi_1(X_j) \\
   = &\ \xi_1(X_j) h^{\ell+q-2}\times  O_p\left(\frac{\exp(2h^{-\beta}/\gamma)}{h^{1-2\beta_2}}\right).
\end{align*}
Finally, we have the expectation of \eqref{eq:middleportion} given $\bX$ and assuming $\xi_1(X)=O_p(1)$, 
\begin{equation*}
\begin{aligned}
 &\ \IE\left\{\left.\left(n^{-2}\bG_a^\top \bOmega_a \bSigma \bOmega_b \bG_b\right)_{\ell,q}\right \vert \bX \right\}    \\
= &\ h^{q+\ell-2} \mu_{q-1}\mu_{\ell-1}\tau_1^2(x) +h^{q+\ell-1}(\mu_{q-1}\mu_\ell +\mu_q\mu_{\ell-1})\tau_1(x)\tau'_1(x)+\\
&\ \frac{h^{q+\ell}}{2}\left[\mu_{q-1}\mu_{\ell+1}\tau_1(x)\tau_1^{(2)}(x)+\mu_{q+1}\mu_{\ell-1}\tau_1(x)\tau_1^{(2)}(x)+2\mu_q\mu_\ell \{ \tau'_1(x)\}^2\right]+\\
&\ O_p\left(h^{q+\ell-3}\right)+O_p\left(\frac{\exp(2h^{-\beta}/\gamma)}{nh^{3-q-\ell-2\beta_2}}\right).
 \end{aligned}
 \end{equation*}
Putting these entries in the matrix, we have 
\begin{equation}
\label{eq:atmiddle}
\begin{aligned}
 &\ \IE\left(\left.n^{-2}\bG_a^\top \bOmega_a \bSigma \bOmega_b \bG_b\right \vert \bX\right)    \\
 = &\ \bH\bigg(\tau_1^2(x)  \bc_0\bc_0^\top +h \tau_1(x) \tau'_1(x) \left(\bc_0\bc_1^\top+\bc_1 \bc_0^\top\right)\\
     &\ +\frac{h^2}{2}  \left[ \tau_1(x)\tau_1^{(2)}(x)\left(\bc_0\bc_2^\top+\bc_2\bc_0^\top\right)+2\left\{\tau'_1(x)\right\}^2\bc_1\bc_1^\top\right] \\
     &\ \left.+O_p(1/h)+O_p\left(\frac{\exp(2h^{-\beta/\gamma})}{nh^{1-2\beta_2}}\right)\right)\bH,
\end{aligned}
\end{equation}
where $\bc_0=(\mu_0, \mu_1, \ldots, \mu_p)^\top$, $\bc_1=(\mu_1, \mu_2, \ldots, \mu_{p+1})^\top$,  $\bc_2=(\mu_2, \mu_3, \ldots, \mu_{p+1})^\top$.

Using \eqref{eq:left&right}, \eqref{eq:atmiddle}, and that fact that $\be_1^\top \bH^{-1}=\be_1^\top$, we have 
\begin{align*}
&\ \IE[\hat g^*_{1,a}(x)\hat g^*_{1,b}(x)|\bX] \\
= &\ m_1^2(x)+\be_1^\top \bS^{-1} \bigg[\frac{hm_1(x)}{f_X(x)}\left\{-m_1(x) f'_X(x)\left(\bc_0\bc_0^\top \bS^{-1}\bS_1+\bS_1\bS^{-1}\bc_0\bc_0^\top\right)+\right.\\
&\  \left. \tau'_1(x)(\bc_0\bc_1^\top+\bc_1\bc_0^\top)\right\}+\frac{h^2}{2f_X(x)}\bigg(\frac{1}{f_X(x)}\left[\tau_1(x)\tau_1^{(2)}(x)\left(\bc_0\bc_2^\top+\bc_2\bc_0^\top\right)+\right.\\
&\ \left.2\left\{\tau'_1(x)\right\}^2\bc_1\bc_1^\top\right]-\frac{2m_1(x)}{f_X(x)}f'_x(x)\tau'_1(x)(\bc_0\bc_1^\top+\bc_1\bc_0^\top)\left(\bS^{-1}\bS_1+\bS_1\bS^{-1}\right)+\\
&\ \frac{2m_1^2(x)}{f_X(x)}\left\{f'_X(x)\right\}^2\bS_1 \bS^{-1}\bc_0\bc_0^\top\bS^{-1}\bS_1 -m_1^2(x)f_X^{(2)}(x)\left(\bc_0\bc_0^\top\bS^{-1}\bS_2+\right.\\
&\ \left. \bS_2\bS^{-1}\bc_0\bc_0^\top\right)\bigg)+o(h^2)+O_p\left(\frac{\exp(2h^{-\beta}/\gamma)}{nh^{1-2\beta_2}}\right)\bigg]\bS^{-1}\be_1\\
\triangleq &\ m_1^2(x)+\be_1^\top \bS^{-1}\mathbb{C}(h, m_1, p)\bS^{-1}\be_1,
\end{align*}
where we use $\mathbb{C}(h, m_1, p)$ to denote the expression between $\be_1^\top\bS^{-1}$ and $\bS^{-1}\be_1$, with its dependence on $h$, $p$, and the function $m_1(\cdot)$ highlighted. 
Using this result in \eqref{eq:yet2find} yields, if $p$ is odd, 
\begin{equation}
\label{eq:varg1ce}
\begin{aligned}
  &\ \text{Var}\{\hat{g}_{1,\text{CE}}(x)|\mathbf{X}\} \\
= &\ \left(1-\frac{1}{B}\right)\be_1^\top \bS^{-1}\left\{\mathbb{C}(h, p, m_1)\bS^{-1}\be_1-2m_1(x)m_1^{(p+1)}(x)\frac{h^{p+1}}{(p+1)!}\bc_p\right\};
\end{aligned}  
\end{equation}
and, if $p$ is even, 
\begin{equation*}
\begin{aligned}
&\ \text{Var}\left\{\hat  g_{1, \text{CE}}(x)|\bX\right\} \\
= &\ \left(1-\frac{1}{B}\right)\be_1^\top\bS^{-1}\bigg[\mathbb{C}(h, p, m_1)\bS^{-1}\be_1-2m_1(x) \bigg\{ m_1^{(p+2)}(x)+\\
&\ \left.(p+2)m_1^{(p+1)}(x) \frac{f'_x(x)}{f_X(x)}\right\}\frac{h^{p+2}}{(p+2)!}\tilde \bc_p\bigg].
\end{aligned}
\end{equation*}
Results for $\text{Var}\{\hat g_{2,\text{CE}}(x)|\bX\}$ are given by the same expressions but with $m_2(x)$ and $m'_2(x)$ replacing $m_1(x)$ and $m'_1(x)$, respectively. 
Because $U$ is assumed to be normal when formulating $\hat g_{1,\text{CE}}(x)$ and $\hat g_{2,\text{CE}}(x)$, we have $\beta=2$, $\beta_2=0$ (because $\beta_0=0$), and $\gamma=2/\sigma_u^2$ in these asymptotic variance results.

In the special case considered in the main article with $p=1$, \eqref{eq:varg1ce} reduces to 
\begin{align*}
  &\ \text{Var}\{\hat{g}_{1,\text{CE}}(x)|\mathbf{X}\} \\
  = &\ \left(1-\frac{1}{B}\right)\frac{h^2m_1(x)}{f_X(x)}\left[(\mu_2+\mu_4/\mu_2) f^{(2)}_X(x) m_1(x) -2 f'_X(x) m'_1(x) \times \right. \\
  &\ \left.\left\{1+\mu_2\frac{f'_X(x)m_1(x)}{f_X(x) m'_1(x)}-\mu_2\right\} \right]+O_p(h^4)+O_p\bigg(\frac{\exp(2h^{-\beta}/\gamma)}{nh^{1-2\beta_{2}}}\bigg).  
\end{align*}
Similarly, 
\begin{equation*}
\begin{aligned}
 &\ \text{Var}\{\hat{g}_{2,\text{CE}}(x)|\mathbf{X}\} \\
 = &\ \left(1-\frac{1}{B}\right)\frac{h^2m_2(x)}{f_X(x)}\left[(\mu_2+\mu_4/\mu_2) f^{(2)}_X(x) m_2(x) -2 f'_X(x) m'_2(x) \times \right. \\
  &\ \left.\left\{1+\mu_2\frac{f'_X(x)m_2(x)}{f_X(x) m'_2(x)}-\mu_2\right\} \right]+O_p(h^4)+O_p\bigg(\frac{\exp(\sigma_u^2 h^{-2})}{nh}\bigg). 
\end{aligned}
\end{equation*}

\subsection{Covariance of $\hat{g}_{1,\text{CE}}(x)$ and $\hat{g}_{2,\text{CE}}(x)$}
\label{sec:CovgCE}
Using similar tactics adopted in the previous two subsections, one can derive expectations involved in the covariance of $\hat{g}_{1,\text{CE}}(x)$ and $\hat{g}_{2,\text{CE}}(x)$. In particular,
\begin{align}
&\ \text{Cov}\left\{\left.\hat{g}_{1,\text{CE}}(x), \, \hat{g}_{2,\text{CE}}(x) \right \vert \bX \right\} \nonumber \\
= &\ \IE\left\{\left. \hat g_{1,\text{CE}}(x)\hat g_{2,\text{CE}}(x)\right \vert \bX\right\}-\IE\left\{\left. \hat g_{1,\text{CE}}(x)\right \vert \bX\right\}\IE\left\{\left. \hat g_{2,\text{CE}}(x)\right \vert \bX\right\} \nonumber\\
= &\ \frac{1}{B}\IE\left\{\left. \hat g^*_{1,b}(x)\hat g^*_{2,b}(x)\right \vert \bX\right\}+\left(1-\frac{1}{B}\right)\IE\left\{\left. \hat g^*_{1,a}(x)\hat g^*_{2,b}(x)\right \vert \bX\right\} -\nonumber \\
&\ \IE\left\{\left. \tilde g_1(x)\right \vert \bX\right\}\IE\left\{\left. \tilde g_2(x)\right \vert \bX\right\}, \text{ where $a\ne b$,} \nonumber\\
= &\ \frac{1}{B}\IE\left\{\left. \tilde g_1(x)\tilde g_2(x)\right \vert \bX\right\}+\left(1-\frac{1}{B}\right)\IE\left\{\left. \hat g^*_{1,a}(x)\hat g^*_{2,b}(x)\right \vert \bX\right\} -\nonumber \\
&\ \IE\left\{\left. \tilde g_1(x)\right \vert \bX\right\}\IE\left\{\left. \tilde g_2(x)\right \vert \bX\right\}.\label{eq:cov3parts}
\end{align}
The term in \eqref{eq:cov3parts} as the products of two expectations can be elaborated using \eqref{eq:Egtilde1}. The first expectation in \eqref{eq:cov3parts} only concerns estimators in the absence of measurement error and can be derived using a similar strategy employed to derive \eqref{eq:condexp1} without involving error-contaminated covariates, which gives 
\begin{align*}
   &\ \IE\left\{\left. \tilde g_1(x)\tilde g_2(x)\right \vert \bX\right\} \\
= &\ m_1(x) m_2(x) +\frac{\be_1^\top \bS^{-1}}{f^2_{X}(x)}\bigg(h f_X(x)\left\{-m_1(x) m_2(x) f'_X(x)\left(\bc_0\bc_0^\top \bS^{-1}\bS_1+\right.\right.\\
&\ \left.\left.\bS_1\bS^{-1}\bc_0\bc_0^\top \right)+m_1(x)\tau'_2(x)\bc_1\bc_0^\top+\tau'_1(x)m_2(x) \bc_0\bc_1^\top \right\}+\frac{h^2}{2}\bigg[\tau_1(x)\tau^{(2)}_2(x)\times\\
&\ \bc_2\bc_0^\top+\tau_1^{(2)}(x)\tau_2(x) \bc_0\bc_2^\top +2 \tau'_1(x) \tau'_2(x)\bc_1\bc_1^\top -2f'_X(x)\left\{m_1(x)\tau'_2(x)\times \right. \\
&\ \left.\left( \bS_1\bS^{-1}\bc_1\bc_0^\top+\bc_1\bc_0^\top \bS^{-1}\bS_1\right) + \tau'_1(x)m_2(x) \left(\bS_1\bS^{-1}\bc_0\bc_1^\top+\bc_0\bc_1^\top \bS^{-1}\bS_1\right) \right\}+ \\
&\ 2 m_1(x)m_2(x) \left\{f'_X(x)\right\}^2 \bS_1\bS^{-1}\bc_0\bc_0^\top \bS^{-1}\bS_1-m_1(x)m_2(x) f_X(x) f_X^{(2)}(x)\times \\
&\ \left(\bc_0\bc_0^\top \bS^{-1}\bS_2+\bS_2\bS^{-1}\bc_0\bc_0^\top\right)\bigg]+o(h^2)+O_p\left(\frac{1}{h}\right)\bigg)\bS^{-1}\be_1\\
&\ \triangleq m_1(x) m_2(x) +\be_1^\top \bS^{-1}\left\{\tilde{\mathbb{C}}(h, p)+O_p\left(\frac{1}{h}\right)\right\}\bS^{-1}\be_1.
\end{align*}
The second expectation in \eqref{eq:cov3parts} is also similar to \eqref{eq:condexp1} and involves error-contaminated covariates in $\bW^*_a$ and $\bW^*_b$, and thus has the same dominating terms as those in $\IE\left\{\left. \tilde g_1(x)\tilde g_2(x)\right \vert \bX\right\}$ above but with the last term being $O_p(\exp(2h^{-\beta}/\gamma)/(nh^{1-2\beta_2}))$ in place of $O_p(1/h)$. Putting the four expectations in \eqref{eq:cov3parts} together reveals that, if $p$ is odd,  
\begin{equation}
\label{eq:covg12ce}
\begin{aligned}
&\ \text{Cov}\left\{\left.\hat{g}_{1,\text{CE}}(x), \, \hat{g}_{2,\text{CE}}(x) \right \vert \bX \right\}\\
= &\ \be_1^\top\bS^{-1}\bigg[ \left\{\tilde{\mathbb{C}}(h,p)+O_p\left(\frac{1}{Bh}\right)+O_p\left(\frac{\exp(2h^{-\beta}/\gamma)}{nh^{1-2\beta_2}}\right)\right\}\bS^{-1}\be_1-\\
&\ \left\{m_1^{(p+1)}(x) m_2(x) +m_1(x) m^{(p+1)}_2(x)\right\}\frac{h^{p+1}}{(p+1)!}\bc_p\bigg];
\end{aligned}
\end{equation}
and, if $p$ is even, 
\begin{equation}
\label{eq:covg12ce2}
\begin{aligned}
&\ \text{Cov}\left\{\left.\hat{g}_{1,\text{CE}}(x), \, \hat{g}_{2,\text{CE}}(x) \right \vert \bX \right\}\\
= &\ \be_1^\top\bS^{-1}\bigg( \left\{\tilde{\mathbb{C}}(h,p)+O_p\left(\frac{1}{Bh}\right)+O_p\left(\frac{\exp(2h^{-\beta}/\gamma)}{nh^{1-2\beta_2}}\right)\right\}\bS^{-1}\be_1-\\
&\ \bigg[ m_1(x)\left\{m_2^{(p+2)}(x)+(p+2)m_2^{(p+1)}(x) \frac{f_X'(x)}{f_X(x)} \right\}+\\
&\ m_2(x)\left\{m_1^{(p+2)}(x)+(p+2)m_1^{(p+1)}(x) \frac{f_X'(x)}{f_X(x)} \right\}\bigg]\frac{h^{p+2}}{(p+2)!}\tilde \bc_p\bigg).
\end{aligned}
\end{equation}
Let $B/n$ tend to a positive constant as $n\to  \infty$ so that the term $O_p(1/(Bh))$ in \eqref{eq:covg12ce} and \eqref{eq:covg12ce2} is included in $O_p(\exp(2h^{-\beta}/\gamma)/(nh^{1-2\beta_2}))$. This also allows dropping the factor $(1-1/B)$ in the dominating terms in \eqref{eq:varg1ce}. 
    
In the special case with $p=1$, following algebraic simplifications, we have 
\begin{equation*}
\begin{aligned}
&\ \text{Cov}\left\{\left.\hat{g}_{1,\text{CE}}(x), \, \hat{g}_{2,\text{CE}}(x) \right \vert \bX \right\}\\
= &\ \frac{h^2}{f^2_X(x)}\left(m_1(x)m_2(x) \left[(\mu_2+\mu_4/\mu_2) f_X(x) f_X^{(2)}(x) -2\mu_2 \left\{ f'_X(x)\right\}^2\right] \right. \\
&\ +\mu_2 f_X(x) \left\{f_X(x)-f'_X(x)\right\}\left\{m_1(x) m'_2(x)+m'_1(x)m_2(x) \right\} \Big)\\&+O(h^4)+O_p\left(\frac{1}{Bh}\right)+O_p\left(\frac{\exp(\sigma_u^2 h^{-2})}{nh}\right).
\end{aligned}
\end{equation*}

\subsection{Asymptotic bias and variance of $\hat{m}_{\text{CE}}(x)$}
Using the biases of $\hat g_{1, \text{CE}}(x)$ and $\hat g_{2, \text{CE}}(x)$ in Appendix~\ref{sec:EgCE}, their variances in Appendix~\ref{sec:VargCE}, and the covariance of them in Appendix~\ref{sec:CovgCE} in \eqref{eq:mhatbias}, we have, if $p$ is odd, 
\begin{align*}
    &\ \text{Bias}\left\{\hat m_{\text{CE}}(x)|\bX\right\} \\
    = &\ \be_1^\top \bS^{-1}\left(\frac{h^{p+1}}{(p+1)!\left\{m_1^2(x)+m_2^2(x)\right\}}\left(m_2(x)m_1^{(p+1)}(x)-m_1(x) m_2^{(p+1)}(x)+ \right.\right.\\
    &\ \frac{1}{m_1^2(x)+m_2^2(x)}\left[2m_1(x)m_2(x) \left\{m_1(x)m_1^{(p+1)}(x)-m_2(x)m_2^{(p+1)}(x)\right\} +\right.\\
    &\ \left.\left.\left\{m_2^2(x)-m_1^2(x)\right\}\left\{m_1^{(p+1)}(x)m_2(x)+m_1(x)m_2^{(p+1)}(x)\right\}\right]\right)\bc_p+o_p\left(h^{p+1}\right)-\\
    &\ \frac{1}{\left\{m_1^2(x)+m_2^2(x)\right\}^2}\left[m_1(x)m_2(x)\left\{\mathbb{C}(h, p, m_1)-\mathbb{C}(h, p, m_2)\right\}+\right.\\
    &\ \left.\left.\left\{m_2^2(x)-m_1^2(x)\right\}\left\{\tilde{\mathbb{C}}(h, p)+O\left(\frac{\exp(2h^{-\beta}/\gamma)}{nh^{1-2\beta_2}}\right)\right\}\right]\bS^{-1}\be_1\right);
\end{align*}
and, if $p$ is even, 
\begin{align*}
    &\ \text{Bias}\left\{\hat m_{\text{CE}}(x)|\bX\right\} \\
    = &\ \be_1^\top \bS^{-1}\left(\frac{h^{p+2}}{(p+2)!\left\{m_1^2(x)+m_2^2(x)\right\} }\bigg( m_2(x)\bigg\{m_1^{(p+2)}(x)+(p+2)m_1^{(p+1)}(x)\times\right.\\
    &\ \left.\frac{f_X'(x)}{f_X(x)}\right\}-m_1(x)\left\{m_2^{(p+2)}(x)+(p+2)m_2^{(p+1)}(x)\frac{f_X'(x)}{f_X(x)}\right\}+\\
    &\ \frac{2+m_1^2(x)-m_2^2(x)}{m_1^2(x)+m_2^2(x)}\left[m_1(x)\left\{m_1^{(p+2)}(x)+(p+2)m_1^{(p+1)}(x)\frac{f_X'(x)}{f_X(x)}\right\}\right.+\\
    &\ \left. m_2(x)\left\{m_2^{(p+2)}(x)+(p+2)m_2^{(p+1)}(x)\frac{f_X'(x)}{f_X(x)}\right\}\right]\bigg)\tilde \bc_p+o_p\left(h^{p+2}\right)-\\
    &\ \frac{1}{\left\{m_1^2(x)+m_2^2(x)\right\}^2}\bigg[m_1(x)m_2(x) \left\{\mathbb{C}(h, p, m_1)-\mathbb{C}(h, p, m_2)\right\}+\\
    &\ \left\{m_2^2(x)-m_1^2(x)\right\}\left\{\tilde {\mathbb{C}}(h, p)+O_p\left(\frac{\exp(2h^{-\beta}/\gamma)}{nh^{1-2\beta_2}}\right)\right\}\bigg]\bS^{-1}\be_1\bigg).
\end{align*}

In the special case with $p=1$
, \begin{align}\label{CEBias}
&\ \text{Bias}\left\{\hat{m}_{\text{CE}}(x)|\bX\right\} \nonumber\\
=&\ \frac{h^2 \mu_2}{2}\left\{m^{(2)}(x)+2m'(x) \frac{m_1(x)m'_1(x)+m_2(x)m_2'(x)}{m_1^2(x)+m_2^2(x)}\right\}+\nonumber\\
&\ \frac{h^2}{f_X(x)\left\{m_1^2(x)+m_2^2(x)\right\}^2}\Big[ 2 (1-\mu_2) f_X'(x) m_1(x) m_2(x) \{m_1(x) m'_1(x)\nonumber\\& -m_2(x)m'_2(x)\}+
 \mu_2 \{f_X(x) -f'_X(x)\}\{m_1^2(x)-m_2^2(x) \}\{m_1(x)m'_2(x)\nonumber\\&+m'_1(x)m_2(x)\}\Big]+O_p(h^4)+O_p(\frac{\exp(\sigma_u^2 h^{-2})}{nh}).
\end{align}
    
Using these bias results, along with the biases of $\hat g_{1, \text{CE}}(x)$ and $\hat g_{2, \text{CE}}(x)$ in Appendix~\ref{sec:EgCE}, their variances in Appendix~\ref{sec:VargCE}, and the covariance of them in Appendix~\ref{sec:CovgCE}, in \eqref{eq:Emhatsq} reveals 
$\text{Var}\{\hat m_{\text{CE}}(x)|\bX\}$. In the special case with $p=1$,
\begin{align}\label{CEVar}
    &\ \text{Var}\{\hat m_{\text{CE}}(x)|\bX\} \nonumber\\
    = &\ 2h^2\left\{(2\mu_2-1)f_X'(x)/f_X(x)-\mu_2\right\}\frac{m_1(x)m_2(x)\left\{m_1(x)m_2'(x)+m'_1(x)m_2(x)\right\}}{\left\{ m_1^2(x)+m^2_2(x)\right\}^2}\nonumber\\
    &\ +O_p(h^4)+O_p\left(\frac{\exp\left(\sigma_u^2h^{-2}\right)}{nh}\right).
\end{align}
       
\subsection{Asymptotic normality of $\hat g_{1, \text{CE}}$ and $\hat g_{2, \text{CE}}$} \label{sec:normgCE}

Our complex-error corrected estimator falls under a larger umbrella of Monte-Carlo corrected score estimators. To this end,  \cite{novickthesis} proved than these Monte-Carlo corrected score estimators are asymptotically identical to SIMEX estimators under the same model. Further, \cite{song} showed that the SIMEX estimator applied to the nonparametric regression model with normal errors, and utilizing a local linear kernel, is asymptotically normally distributed. With this in mind, it can be shown that our method is asymptotically equivalent to a Monte-Carlo corrected score, specifically one that arises from the objective function
$$\mathcal{F}=\sum_{i=1}^{n}\bigg[Y_{i}-\alpha-\beta(X_{i}-x)\bigg]^{2}K_{h}(X_{i}-x).$$
Thus the complex-error estimators, $\hat g_{1, \text{CE}}$ and $\hat g_{2, \text{CE}}$, are asymptotically equivalent to the SIMEX estimator, which is asymptotically normal.
\renewcommand*{\theequation}{F.\arabic{equation}}
\setcounter{equation}{0}
\renewcommand*{\thesubsection}{F.\arabic{subsection}}
\renewcommand*{\thelemma}{F.\arabic{lemma}}
\setcounter{subsection}{0}
\section*{{\revise Appendix F: Asymptotic properties of the one-step correction estimator $\hat m_{\text{OS}}(x)$}}

We consider the $p$-th order local polynomial weight used in $\hat m_{\text{OS}}(x)$ in this appendix, for $p\ge 0$. For this generalization, one needs to change ``$\ell=0, 1, 2, 3$'' in {\bf Condition O} stated in Section~\ref{sec:asymp} to ``$\ell=0, 1, \ldots, 2p+1$''; similarly, one needs to change ``$\ell=0,1, 2$" in {\bf Condition S} to ``$\ell=0, 1, \ldots, 2p$''. 

\subsection{The first two moments of $\hat g_{1, \text{OS}}(x)$ and $\hat g_{2, \text{OS}}(x)$}
Because the one-step correction estimator $\hat g_{\ell, \text{OS}}(x)$ for $g_\ell(x)=m_\ell(x)f_X(x)$ is precisely the estimator for, in notations used in \citet{huangzhou}, $B(x)=\IE(Y|X=x)f_X(x)$ with a linear response $Y$, the bias and variance of $\hat g_{\ell, \text{OS}}(x)$ are readily available in \citet{huangzhou}. In particular, recall that $\bS=(\mu_{r_1+r_2})_{0\le r_1, r_2\le p}$, $\bc_p=(\mu_{p+1}, \ldots, \mu_{2p+1})^\top$, and $\tilde\bc_p=(\mu_{p+2}, \ldots, \mu_{2p+2})^\top$, by Equation (A.5) in \cite{huangzhou}, we have for $\ell=1, 2$,
\begin{equation}
\label{eq:EgOS}
\begin{aligned}
&\ \IE\left\{\left.\hat{g}_{\ell, \text{OS}}(x)\right \vert \bW \right\}\\
= & 
\begin{cases}
    g_\ell(x) +(N_{\ell,p}*D)(x)+(M_{\ell, p}*D)(x)\displaystyle{\frac{h^{p+1}}{(p+1)!}}+o_p(h^{p+1}), & \text{if $p$ is odd,} \\
    g_\ell(x) +(N_{\ell,p}*D)(x)+(M_{\ell, p}*D)(x)\displaystyle{\frac{h^{p+2}}{(p+2)!}}+o_p(h^{p+2}), & \text{if $p$ is even,} 
\end{cases}
\end{aligned}
\end{equation}
where $N_{\ell,p}(w) = m^*_\ell(w)\sum_{r=1}^p f^{(r)}_W(w)\mu_r h^r/r!$, and, if $p$ is odd,
\begin{equation}
M_{\ell, p}(w)=m_\ell^*(w)f^{(p+1)}_W(w)\mu_{p+1}+m_\ell^{*(p+1)}(w) f_W(w) \be_1^\top \bS^{-1}\bc_p, \label{eq:Mlpodd}
\end{equation}
and, if $p$ is even, 
\begin{equation}
\label{eq:Mlpeven}
\begin{aligned}
    &\ M_{\ell, p}(w)\\
    = &\ m_\ell^*(w)f^{(p+2)}_W(w)\mu_{p+2}+\left\{m_\ell^{*(p+2)}(w) f_W(w)+(p+2)m_\ell^{*(p+1)}f'_W(w)\right\} \be_1^\top \bS^{-1}\tilde\bc_p,
\end{aligned} 
\end{equation}
with the Fourier transform of $N_{\ell,p}(w)$ and $M_{\ell,p}(w)$ compactly supported on $I_t$, and $D(x) = (2\pi)^{-1}\int_{I_t} e^{-itx}/\phi_U(t)dt$. Because the Fourier transform of $M_{\ell,p}(w)$, $\phi_{M_{\ell,p}}(t)$, is compactly supported on $I_t$,  the convolution of $M_{\ell,p}(\cdot)$ and $D(\cdot)$ is
\begin{align*}
    (M_{\ell,p}* D)(x) & = \int M_{\ell,p}(w)D(x-w)dw \\
    & = \frac{1}{2\pi} \int M_{\ell,p}(w) \int_{I_t} e^{-it(x-w)}/\phi_U(t)\, dt dw \\
    &= \frac{1}{2\pi} \int_{I_t} e^{-itx}\frac{1}{\phi_U(t)} \int  e^{itw} M_{\ell,p}(w)\, dw dt  \\
    &= \frac{1}{2\pi} \int_{I_t} e^{-itx}\frac{\phi_{M_{\ell,p}}(t)}{\phi_U(t)} \, dt.
\end{align*} 
In other words, $(M_{\ell,p} *D)(x)$ is the same integral transform of $M_{\ell,p}(\cdot)$ appearing in Lemma 3 in the main article but expressed in the form of convolution, that is, $(M_{\ell,p} *D)(x)=\mathcal{T}_U(M_{\ell,p})(x)$ by adopting the integral transform $\mathcal{T}_U$ defined in Section 3.3 in the main article. Similarly, $(N_{\ell,p} *D)(x)=\mathcal{T}_U(N_{\ell,p})(x)$. 

Define $\bS^*=(\eta(r_1, r_2))_{0\le r_1, r_2\le p}$. Then, according to Appendix B in \cite{huangzhou}, we have, for $\ell=1, 2$, 
\begin{equation}
\text{Var}\left\{ \left. \hat g_{\ell,\text{OS}}(x)\right \vert \bW \right\} = 
\be_1^\top \bS^{-1}\bS^*\bS^{-1}\be_1 \frac{\sigma_\ell^{*2}(x)f_W(x)}{nh^{1+2\beta}}+o_p\left(\frac{1}{nh^{1+2\beta}}\right),
\label{eq:VargOSord}
\end{equation}
if $U$ is ordinary smooth, 
and $\text{Var}\left\{ \left. \hat g_{\ell,\text{OS}}(x)\right \vert \bW \right\}$ is bounded from above by 
\begin{equation}
\be_1^\top \bS^{-1}\bS^{-1}\be_1\frac{C\sigma_\ell^{*2}(x) f_W(x)\exp(2h^{-\beta}/\gamma)}{f_X^2(x) nh^{1-2\beta_2}}+o_p\left(\frac{\exp(2h^{-\beta}/\gamma)}{nh^{1-2\beta_2}}\right),
\label{eq:VargOSsup}
\end{equation}
if $U$ is super smooth,
where $\beta_2=\beta_0I(\beta_0<0.5)$, $\sigma^{*2}_1(x)=\text{Var}(\sin \Theta|W=x)=\xi^*_1(x)-m^{*2}_1(x)$ and $\sigma^{*2}_2(x)=\text{Var}(\cos \Theta|W=x)=\xi^*_2(x)-m^{*2}_2(x)$. 

Lastly, adopting the same strategy elaborated in Appendix B in \cite{huangzhou}, one can show that 
\begin{equation}
\text{Cov}\left\{ \left. \hat g_{1,\text{OS}}(x), \hat g_{2,\text{OS}}(x)\right \vert \bW \right\}=\be_1^\top \bS^{-1}\bS^*\bS^{-1}\be_1 \frac{\phi^*(x)f_W(x)}{nh^{1+2\beta}}+o_p\left(\frac{1}{nh^{1+2\beta}}\right),
\label{eq:CovgOSord}
\end{equation}
if $U$ is ordinary smooth, and $\text{Cov}\left\{ \left. \hat g_{1,\text{OS}}(x), \hat g_{2,\text{OS}}(x)\right \vert \bW \right\}$ is bounded from above by 
\begin{equation}
\be_1^\top \bS^{-1}\bS^{-1}\be_1\frac{C\phi^*(x) f_W(x)\exp(2h^{-\beta}/\gamma)}{f_X^2(x) nh^{1-2\beta_2}}+o_p\left(\frac{\exp(2h^{-\beta}/\gamma)}{nh^{1-2\beta_2}}\right),
\label{eq:CovgOSsup}
\end{equation}
if $U$ is super smooth,
where $\phi^*(x)=\text{Cov}(\sin \Theta, \cos \Theta|W=x)=\psi^*(x)-m_1^*(x)m_2^*(x)$.

\subsection{Asymptotic bias and variance of \texorpdfstring{$\hat{m}_{\text{OS}}(x)$}{Lg}}
\label{sec:biasvarOS}
By \eqref{eq:mhatbias}, with $g_\ell(x)=m_\ell(x) f_X(x)$, for $\ell=1, 2$,
\begin{align*}
&\ \text{Bias}\left\{ \left. \hat m_{\text{OS}}(x)\right \vert \bW\right \} \\
= &\ \frac{1}{g^2_1(x)+g^2_2(x)}\left[g_2(x) \text{Bias}\left\{\left. \hat g_{1,\text{OS}}(x)\right \vert \bW\right\}-g_1(x) \text{Bias}\left\{\left. \hat g_{2,\text{OS}}(x)\right \vert \bW\right\} \right]-\\
&\  \frac{1}{\left\{g^2_1(x)+g^2_2(x)\right\}^2}\left( g_1(x)g_2(x) \left[ \text{Var}\left\{ \left. \hat g_{1,\text{OS}}(x) \right\vert \bW\right\} -\text{Var}\left\{ \left. \hat g_{2,\text{OS}}(x) \right\vert \bW\right\}\right] +\right. \\
&\ \left.\left\{ g^2_2(x)-g^2_1(x)\right\} \text{Cov}\left\{ \left. \hat g_{1,\text{OS}}(x), \hat g_{2,\text{OS}}(x) \right\vert \bW\right\}\right)+\text{lower-order terms,}
\end{align*}
{\reviseagain where the lower-order terms refer to the non-dominating terms of $\IE\{\hat g_{\ell,\text{OS}}(x)-g_\ell(x)|\bW\}$ indicated in \eqref{eq:EgOS}.}
By \eqref{eq:EgOS}, \eqref{eq:VargOSord}, and \eqref{eq:CovgOSord}, if $U$ is ordinary smooth,
\begin{align*}
    &\ \text{Bias}\left\{\hat m_{\text{OS}}(x)|\bW\right\}\\
    = &\ \frac{1}{f_X(x)\left\{m_1^2(x)+m_2^2(x)\right\}}
    \bigg[m_2(x) \mathcal{T}_U(N_{1,p})(x)-m_1(x) \mathcal{T}_U(N_{2,p})(x)\\
    &\ \left.+\left\{m_2(x) \mathcal{T}_U(M_{1,p})(x)-m_1(x) \mathcal{T}_U(M_{2,p})(x)\right\}
    \frac{h^{p+1+I(\text{$p$ is even})}}{(p+1+I(\text{$p$ is even}))!}\right]\\
    &\ +o_p\left(h^{p+1+I(\text{$p$ is even})}\right)+\frac{\be_1^\top \bS^{-1}\bS^*\bS^{-1}\be_1 f_W(x)}{nh^{1+2\beta}f^2_X(x)\left\{m_1^2(x)+m_2^2(x)\right\}^2}\times \\
    &\ \left[m_1(x)m_2(x) \left\{\sigma^{*2}_1(x)-\sigma_2^{*2}(x)\right\}+\phi^*(x)\left\{m_2^2(x)-m_1^2(x)\right\}\right]+o_p\left(\frac{1}{nh^{1+2\beta}}\right);
\end{align*}
and, if $U$ is super smooth, using \eqref{eq:EgOS}, \eqref{eq:VargOSsup}, and \eqref{eq:CovgOSsup}, we have 
\begin{align*}
    &\ \text{Bias}\left\{\hat m_{\text{OS}}(x)|\bW\right\}\\
    = &\ \frac{1}{f_X(x)\left\{m_1^2(x)+m_2^2(x)\right\}}
    \bigg[m_2(x) \mathcal{T}_U(N_{1,p})(x)-m_1(x) \mathcal{T}_U(N_{2,p})(x)\\
    &\ \left.+\left\{m_2(x) \mathcal{T}_U(M_{1,p})(x)-m_1(x) \mathcal{T}_U(M_{2,p})(x)\right\}
    \frac{h^{p+1+I(\text{$p$ is even})}}{(p+1+I(\text{$p$ is even}))!}\right]\\
    &\ +o_p\left(h^{p+1+I(\text{$p$ is even})}\right)+O_p\left(\frac{\exp(2h^{-\beta}/\gamma)}{nh^{1-2\beta_2}}\right).
\end{align*}
In the special case we focus on in our study with a local linear weight, i.e., with $p=1$, because $N_{\ell, 1}(w)=m_\ell^*(w)f'_W(w)\mu_1h=0$, for $\ell=1, 2$, we have, if $U$ is ordinary smooth, 
\begin{align*}
&\ \text{Bias}\left\{ \left. \hat m_{\text{OS}}(x)\right \vert \bW\right \} \\
=&\  \frac{h^2 \mu_2}{2f_X(x) \left\{m^2_1(x)+m^2_2(x) \right\}}\left\{m_2(x)\mathcal{T}_U(M_{1,1})(x)-m_1(x) \mathcal{T}_U(M_{2,1})(x)\right\} +o_p\left(h^2\right) \\
&\ + \frac{f_W(x)\eta(0,0)\left[m_1(x)m_2(x) \left\{ \sigma^{*2}_1(x)-\sigma^{*2}_2(x)\right\} 
+\phi^*(x)\left\{m^2_2(x)-m^2_1(x)\right\}\right]}{nh^{1+2\beta}f_X^2(x) \left\{m_1^2(x)+m^2_2(x)\right\}^2} \\ &+ o_p\left(\frac{1}{nh^{1+2\beta}} \right);
\end{align*}
and, if $U$ is super smooth, 
\begin{align}\label{OSBias}
\ \text{Bias}\left\{ \left. \hat m_{\text{OS}}(x)\right \vert \bW\right \} 
=&\  \frac{h^2 \mu_2}{2f_X(x) \left\{m^2_1(x)+m^2_2(x) \right\}}\{m_2(x)\mathcal{T}_U(M_{1,1})(x)\nonumber\\&-m_1(x) \mathcal{T}_U(M_{2,p})(x)\} +o_p\left(h^2\right)+O_p\left( \frac{\exp(2h^{-\beta}/\gamma)}{nh^{1-2\beta_2}}\right).
\end{align}
Since the dominating terms in these asymptotic bias do not depend on $\bW$, they can be interpreted as the unconditional asymptotic dominating bias. 

Using \eqref{eq:EgOS}, \eqref{eq:VargOSord}, and \eqref{eq:CovgOSord} in \eqref{eq:Emhatsq} gives that, for ordinary smooth $U$, 
\begin{align*}
     &\ \IE\left\{\left. 
     \hat m^2_{\text{OS}}(x)\right \vert \bW \right\}\\
     = &\ m^2(x)+\frac{2m(x)}{f_X(x)\left\{m_1^2(x)+m_2^2(x)\right\}}\bigg[m_2(x) \mathcal{T}_U(N_{1,p})(x)-m_1(x) \mathcal{T}_U(N_{2,p})(x)\\
    &\ \left.+\left\{m_2(x) \mathcal{T}_U(M_{1,p})(x)-m_1(x) \mathcal{T}_U(M_{2,p})(x)\right\}
    \frac{h^{p+1+I(\text{$p$ is even})}}{(p+1+I(\text{$p$ is even}))!}\right]\\
    &\ +o_p\left(h^{p+1+I(\text{$p$ is even})}\right)+\frac{2\be_1^\top \bS^{-1}\bS^*\bS^{-1}\be_1 f_W(x)}{nh^{1+2\beta}f^2_X(x)\left\{m_1^2(x)+m_2^2(x)\right\}^2}\times \\
    &\ \left( m(x) m_1(x)m_2(x) \left\{\sigma_2^{*2}(x)-\sigma_1^{2*}(x) \right\}+\frac{1}{2}\left\{m_1^2(x)\sigma_2^{*2}(x)+m_2^2(x) \sigma_1^{*2}(x)\right\} \right. \\
    &\ +\phi^*(x) \left[m(x) \left\{m_1^2(x)-m_2^2(x)\right\}-m_1(x)m_2(x)\right]\bigg)+o_p\left(\frac{1}{nh^{1+2\beta}}\right);
\end{align*}
and, for super smooth $U$, 
\begin{align*}
     &\ \IE\left\{\left. 
     \hat m^2_{\text{OS}}(x)\right \vert \bW \right\}\\
     = &\ m^2(x)+\frac{2m(x)}{f_X(x)\left\{m_1^2(x)+m_2^2(x)\right\}}\bigg[m_2(x) \mathcal{T}_U(N_{1,p})(x)-m_1(x) \mathcal{T}_U(N_{2,p})(x)\\
    &\ \left.+\left\{m_2(x) \mathcal{T}_U(M_{1,p})(x)-m_1(x) \mathcal{T}_U(M_{2,p})(x)\right\}
    \frac{h^{p+1+I(\text{$p$ is even})}}{(p+1+I(\text{$p$ is even}))!}\right]\\
    &\ +o_p\left(h^{p+1+I(\text{$p$ is even})}\right)+O_p\left(\frac{\exp(2h^{-\beta}/\gamma)}{nh^{1-2\beta_2}}\right).
\end{align*}
It follows that, for ordinary smooth $U$, 
\begin{align*}
    &\ \text{Var}\left\{\hat m_{\text{OS}}(x)|\bW\right\} \\
 = &\ \frac{\be_1^\top \bS^{-1}\bS^*\bS^{-1}\be_1 f_W(x)}{nh^{1+2\beta}f_X^2(x)\left\{m_1^2(x)+m_2^2(x)\right\}^2}\bigg(4m(x)m_1(x)m_2(x)\left\{\sigma^{*2}_2(x)-\sigma^{*2}_1(x)\right\} \\
 &\ +2\phi^*(x)\left[2m(x) \left\{m_1^2(x)-m_2^2(x)\right\}-m_1(x)m_2(x)\right]+m_1^2(x)\sigma^{*2}_2(x)\\
 &\ +m_2^2(x) \sigma^{*2}_1(x)-\frac{2}{f_X(x)\left\{m_1^2(x)+m_2^2(x)\right\}}\bigg[m_1(x)m_2(x) \left\{\sigma_1^{*2}(x)-\sigma_2^{*2}(x)\right\}\\
 &\ +\phi^*(x)\left\{m_2^2(x)-m_1^2(x)\right\}\bigg]\left\{m_2(x) \mathcal{T}_U(N_{1,p})(x)-m_1(x) \mathcal{T}_U(N_{2,p})(x)\right\}\bigg)\\
 &\ +o_p\left(h^{p+1+I(\text{$p$ is even)}}\right)+o_p\left(\frac{1}{nh^{1+2\beta}}\right),
\end{align*}
which reduces to, when $p=1$,
\begin{equation}
\label{OSVar}
\begin{aligned}
  &\   \text{Var}\left\{\left. \hat m_{\text{OS}}(x)\right \vert \bW\right\}\\ 
= &\ \frac{f_W(x) \eta(0,0)}{nh^{1+2\beta} f^2_X(x) \{m^2_1(x)+m^2_2(x)\}^2}\bigg( m^2_1(x) \xi^*_2(x)+m^2_2(x) \xi^*_1(x)\\
&\ 2m_1(x)m_2(x) \psi^*(x)  -\left\{m_1(x)m^*_2(x)-m^*_1(x) m_2(x) \right\}^2 \\
&\ +4m(x)\bigg[m_1(x)m_2(x) \left\{\sigma^{*2}_2(x)-\sigma^{*2}_1(x)\right\} +\phi^*(x) \{m^2_1(x)\\
&\ -m^2_2(x)\}\bigg]\bigg) +o_p\left( \frac{1}{nh^{1+2\beta}}\right).
\end{aligned}
\end{equation}
For super smooth $U$, 
\begin{align*}
    &\ \text{Var}\left\{\hat m_{\text{OS}}(x)|\bW\right\}\\
 = &\ O_p\left(\frac{\exp(2h^{-\beta}/\gamma}{nh^{1-2\beta_2}}\right)-\frac{m_2(x)\mathcal{T}_U(N_{1,p})(x)-m_1(x)\mathcal{T}(N_{2,p})(x)}{f^2_X(x)
 \left\{m_1^2(x)+m_2^2(x)\right\}^2}\times \\
 &\ \bigg[m_2(x)\mathcal{T}_U(N_{1,p})(x)-m_1(x)\mathcal{T}(N_{2,p})(x)+2\{m_2(x)\mathcal{T}_U(M_{1,p})(x)\\
 &\ -m_1(x)\mathcal{T}(M_{2,p})(x)\}\frac{h^{p+1+I(\text{$p$ is even})}}{(p+1+I(\text{$p$ is even}))!}\bigg]+o_p\left(h^{p+1+I(\text{$p$ is even})}\right),
\end{align*}
which reduces to, when $p=1$,
$$\text{Var}\left\{\left. \hat m_{\text{OS}}(x)\right \vert \bW\right\} =O_p\left(\frac{\exp(2h^{-\beta}/\gamma)}{nh^{1-2\beta_2}}\right).$$

\renewcommand*{\theequation}{G.\arabic{equation}}
\setcounter{equation}{0}
\renewcommand*{\thesubsection}{G.\arabic{subsection}}
\renewcommand*{\thelemma}{G.\arabic{lemma}}
\setcounter{subsection}{0}
\section*{Appendix G: Asymptotic normality of $\hat m_{\cdot}(x)$}
Let ``$\stackrel{\mathscr{L}}{\longrightarrow}$'' denote ``converges in law to,'' i.e., ``converges in distribution to,'' and let ``$\stackrel{\mathscr{P}}{\longrightarrow}$'' denote ``converges in probability to.'' A lemma is established next in preparation for proving the asymptotic normality of our proposed estimators for $m(x)$.
\begin{lemma}\label{lem:asynorm}
Consider three sequences of random variables indexed by $n$, $A_n$, $B_n$, and $C_n=\text{atan2}[A_n, B_n]$. If $d_n (A_n-\mu_A, B_n-\mu_B)^\top\stackrel{\mathscr{L}}{\longrightarrow}N(\bzero_2, \bSigma)$ as $n\to \infty$, where $\mu^2_A+\mu^2_B\ne 0$, $\bzero_2$ is a $2\times 1$ vector of zeros, $\bSigma$ is a $2\times 2$ variance-covariance matrix, and $d_n$ is a positive non-random sequence that diverges to $\infty$ as $n\to \infty$, then $d_n(C_n-\text{atan2}[\mu_A, \mu_B]) \stackrel{\mathscr{L}}{\longrightarrow} N(0, (\mu_A^2\bSigma[2, 2]+\mu_B^2\bSigma[1,1]-2 \mu_A\mu_B \bSigma[1,2])/(\mu_A^2+\mu_B^2)^2)$ as $n\rightarrow \infty$.
\end{lemma}

\begin{proof}
Because $d_n(A_n-\mu_A)\stackrel{\mathscr{L}}{\longrightarrow}N(0, \bSigma[1,1])$, where $d_n\to \infty$ as $n\to \infty$, by Theorem 2.3.4 in \cite{lehmann1999elements}, $A_n\stackrel{\mathscr{P}}{\longrightarrow} \mu_A$. Similarly, $B_n\stackrel{\mathscr{P}}{\longrightarrow} \mu_B$. By the continuous mapping theorem \citep[][Theorem 2.1.4]{lehmann1999elements}, $(A_n-\mu_A)^\ell\stackrel{\mathscr{P}}{\longrightarrow}0$ and $(B_n-\mu_B)^\ell\stackrel{\mathscr{P}}{\longrightarrow}0$, for $\ell\in \mathbb{N}$, as $n\to \infty$. 

Using a Taylor expansion of $C_n$ around $(A_n, B_n)=(\mu_A, \mu_B)$, we have 
\begin{align}
  & \ d_n(C_n-\text{atan2}[\mu_A, \mu_B])\\ 
   =&\ 
     \frac{1}{\mu_A^2+\mu_B^2}\{\mu_B d_n(A_n-\mu_A)-\mu_A d_n (B_n-\mu_B)\} \label{eq:line1}\\
     &\scalebox{.8}{$ -\frac{1}{(\mu_A^2+\mu_B^2)^{2}}\bigg[\mu_A\mu_B\left\{d_n(A_n-\mu_A)^{2}-d_n(B_n-\mu_B)^{2}\right\}+(\mu_B^{2}-\mu_A^{2})d_n(A_n-\mu_A)(B_n-\mu_B)\bigg] \label{eq:line2}$}\\
     &\ +O_p\{d_n(A_n-\mu_A)^{3}\}+O_p\{d_n(B_n-\mu_B)^{3}\} \nonumber.
\end{align}
By Slutsky's theorem, \eqref{eq:line1}, which is equal to $(\mu_B, -\mu_A)d_n (A_n-\mu_A, B_n-\mu_B)^\top/(\mu^2_A+\mu^2_B)$, converges in distribution to 
$N(0, (\mu_B, -\mu_A)\bSigma (\mu_B, -\mu_A)^\top/(\mu_A^2+\mu_B^2)^2)$. In \eqref{eq:line2}, since $d_n(A_n-\mu_A)^2=d_n(A_n-\mu_A)\times (A_n-\mu_A)$, with $d_n(A_n-\mu_A) \stackrel{\mathscr{L}}{\longrightarrow} N(0, \bSigma[1,1])$ and $A_n-\mu_A \stackrel{\mathscr{P}}{\longrightarrow}0$, we have  $d_n(A_n-\mu_A)^2\stackrel{\mathscr{P}}{\longrightarrow} 0$ by Slutsky's theorem. Similarly, $d_n(B_n-\mu_B)^2\stackrel{\mathscr{P}}{\longrightarrow} 0$ and  $d_n(A_n-\mu_A)(B_n-\mu_B)\stackrel{\mathscr{P}}{\longrightarrow} 0$. Hence, \eqref{eq:line2} and the remainder terms following \eqref{eq:line2} all converge in probability to zero. Putting these together and using Slutsky's theorem again, we establish the limiting distribution of $d_n(C_n-\text{atan2}[\mu_A, \mu_B])$.
\end{proof}

Now return to our context of estimating the circular mean function $m(x)$. By letting $A_n=\hat g_{1, \cdot}(x)$ and $B_n=\hat g_{2, \cdot}(x)$ in Lemma~\ref{lem:asynorm}, we have $C_n=\hat m_\cdot(x)$, where ``$\cdot$'' refers to one of the acronyms (``DK,'' ``CE,'' and ``OS'') relating to our proposed estimators for $m(x)$. According to Lemma~\ref{lem:asynorm}, the asymptotic normality of our proposed estimators follows directly from the established asymptotic normality of $\hat g_{1, \cdot}(x)$ and $\hat g_{2,\cdot}(x)$. For example, by the asymptotic normality of the one-step correction estimators established in \cite{huangzhou} for $\hat g_{1, \text{OS}}(x)$ and $\hat g_{2,\text{OS}}(x)$, we can immediately conclude that $d_n\{\hat m_{\text{OS}}(x)-\text{atan2}[m_1(x)f_X(x), \, m_2(x)f_X(x)]\}=d_n\{\hat m_{\text{OS}}(x)-m(x)\}$ converges in distribution to a mean-zero normal distribution as $n\to \infty$, where $d_n=\sqrt{n}h^{1/2+\beta}$ if $U$ is ordinary smooth of order $\beta$, and $d_n=\sqrt{n}h^{1-\beta_2}\exp(-h^{-\beta}/\gamma)$ if $U$ is super smooth of order $\beta$. Similar conclusions for $\hat m_{\text{DK}}(x)$ and $\hat m_{\text{CE}}(x)$ can be established by the asymptotic normality of $\hat g_{\ell, \text{DK}}(x)$ according to \cite{Delaigle} and that of $\hat g_{\ell, \text{CE}}(x)$ shown in Section~\ref{sec:normgCE}.
\renewcommand*{\theequation}{H.\arabic{equation}}
\setcounter{equation}{0}
\renewcommand*{\thesubsection}{H.\arabic{subsection}}
\renewcommand*{\thelemma}{H.\arabic{lemma}}
\setcounter{subsection}{0}
\section*{{\revise Appendix H: Mean integrated square error and optimal bandwidths}}
\label{sec:miseopt}
A commonly employed strategy for choosing the bandwidth in local polynomial estimation is through the minimization of the mean integrated square error (MISE),  
\begin{align}\label{MISE}
    \text{MISE}
    &=\int_{\mathbb{R}} \left(\left[\text{Bias}\{\hat{m}_{\cdot}(x)|\bX\}\right]^{2}+\text{Var}\{\hat{m}_{\cdot}(x)|\bX\}\right)w(x)\,dx,
\end{align}
following the definition in \citet[Section 2.1,][]{Fan}, where  $w(x)\ge 0$ is a weight function that we set to one in our study, and $\hat{m}_{\cdot}(x)$ generically refers to one of our proposed estimator. With the asymptotic bias and variance of our proposed estimators in the presence of ordinary smooth measurement error derived, the asymptotic MISE (AMISE) of $\hat m_{\text{DK}}(x)$ and $\hat m_{\text{OS}}(x)$ in this case can be used to find an asymptotically optimal bandwidth. In the presence of super smooth measurement error, we only have the order for some upper bound of the asymptotic variance of $\hat m_{\text{DK}}(x)$ and $\hat m_{\text{OS}}(x)$, which is less useful for optimal bandwidth derivation based on AMISE. Even though the dominating bias and dominating variance of $\hat m_{\text{CE}}(x)$ presented in Theorems~\ref{thm:asympbias} and \ref{thm:asympvar} include elaborated terms of order $O(h^2)$, they also implicitly depend on terms of order $O_p(\exp(\sigma_u^2h^{-2})/(nh))$ that prohibit one from recovering a useful AMISE for this estimator. In what follows, we present the AMISEs of $\hat m_{\text{DK}}(x)$ and $\hat m_{\text{OS}}(x)$, and the corresponding optimal bandwidths when local linear weights are used while assuming ordinary smooth measurement errors. 

By Theorem~\ref{thm:asympbias}, the dominating bias of $\hat m_{\text{DK}}(x)$ and $\hat m_{\text{OS}}(x)$ are of the same form given by $ \mathbb{DB}(x)h^2+\mathbb{DB}^*(x)/(nh^{1+2\beta})$,
where $\mathbb{DB}(x)$ is the dominating terms in $\hat m_{\text{DK}}(x)$ or $\hat m_{\text{OS}}(x)$ that multiply $h^2$, which depends on functionals of $m(x)$ and $\mu_2$, and $\mathbb{DB}^*(x)$ is the dominating terms in these estimators that multiply $1/(nh^{1+2\beta})$, which depends on naive counterparts of functionals of $m(x)$ besides these functionals. By Theorem~\ref{thm:asympvar}, the dominating variance of $\hat m_{\text{DK}}(x)$ and $\hat m_{\text{OS}}(x)$ also share a common form given by $\mathbb{DV}(x)/(nh^{1+2\beta})$, with $\mathbb{DV}(x)$ involving many of the aforementioned functionals. These dominating terms, $\mathbb{DB}(x)$, $\mathbb{DB}^*(x)$, and $\mathbb{DV}(x)$, are free of $n$ and $h$. 
 
Using these generic forms of the dominating bias and dominating variance of an estimator in \eqref{MISE}, we have 
\begin{align*}
    \text{AMISE} & = \int_{\mathbb{R}} \left[\left\{\mathbb{DB}(x) h^2+\frac{\mathbb{DB}^*(x)}{nh^{1+2\beta}}\right\}^2+\frac{\mathbb{DV}(x)}{nh^{1+2\beta}}\right]\, dx.
\end{align*}
Using the dominating terms in the above AMISE, we find  an asymptotic optimal $h$ given by 
\begin{equation}
h_{\text{Optimal}}=\bigg[\frac{(1+2\beta)}{4n}\frac{\int_{\mathbb{R}} \mathbb{DV}(x)\, dx}{\int_{\mathbb{R}}\left\{\mathbb{DB}(x)\right\}^2  \, dx}\bigg]^{1/(5+2\beta)}.
\label{eq:optimalh}
\end{equation}
Using this asymptotic optimal bandwidth in the bias results in Theorems~\ref{thm:asympbias} and \ref{thm:asympvar} gives the bias of $\hat m_{\text{DK}}(x)$ and $\hat m_{\text{OS}}(x)$ both in the order of $O(n^{-2/(5+2\beta)})$, whereas the variance of them are in the order of $O_p(n^{-4/(5+2\beta)})$ when $U$ is ordinary smooth.

\renewcommand*{\thefigure}{I.\arabic{figure}}
\setcounter{figure}{0}
\section*{{\revise Appendix I: Using naive bandwidths in a non-naive estimator}}

Focusing on the deconvoluting kernel estimator, $\hat m_{\text{DK}}(x)$, we conducted a simulation study to compare the naive 5-fold cross-validation bandwidth selection, which ignores measurement errors in the naive loss function, and our proposed CV-SIMEX, with the optimal bandwidth $h_{\text{opt}}$ defined in Section~\ref{sec:simudesign} serving as a reference. In this experiment, we have $m(x)=2\text{atan}(1/x)$ as the regression function,  with $X\sim N(0,4)$ and $U\sim N(0,1)$, yielding $\lambda=0.8$. 
Figure \ref{cvbox} presents the empirical risk of $\hat m_{\text{DK}}(x)$ when different bandwidths are used based on 100 Monte Carlo replicates at each considered sample size. The empirical risk is computed based on fitted values of $m(x)$ at a sequence of $x$ ranging from $-3$ to 3 at increments of 0.1. The improvement in the estimation when using a bandwidth chosen by CV-SIMEX is evident compared to when one implements a naive CV. As $n$ increases, CV-SIMEX becomes closer in performance to the optimal choice of bandwidth. 

\begin{figure}[!h]
\centering
   \includegraphics[scale=.55]{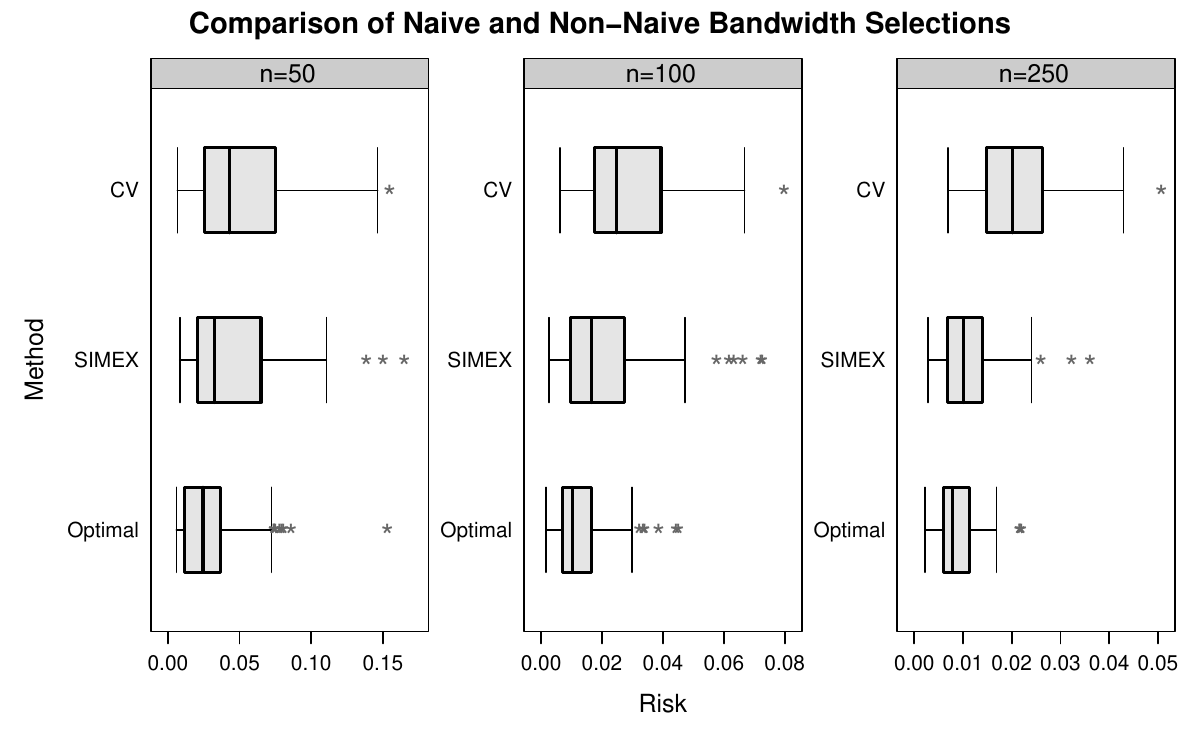}
    \caption{\label{cvbox}Boxplots of empirical risk across 100 Monte Carlo replicates at each level of $n\in \{50, 100, 250\}$ for $\hat m_{\text{DK}}(x)$ when three bandwidth selection methods are used: a naive cross-validation (CV), the proposed CV-SIMEX (SIMEX), and the optimal bandwidth (Optimal).}
\end{figure}

\end{appendix}

\bibliographystyle{apalike}
\bibliography{refs.bib}  

\end{document}